\documentclass[preprint,11pt,3p]{elsarticle}
\usepackage{setspace}
\usepackage{amsmath} 
\usepackage{amssymb} %
\usepackage{bm} 
\usepackage{amsthm} 
\usepackage{amsfonts} 

\usepackage[linewidth=1pt]{mdframed} 

\usepackage{graphicx,psfrag,epsf} 
\usepackage{rotating} 

\biboptions{comma,round,authoryear}

\usepackage{url} 

\usepackage{multirow}
\usepackage{booktabs}
\usepackage{array}
\usepackage{arydshln} 

\usepackage[labelsep=period]{caption}

\usepackage{enumitem}

\usepackage{thmtools} 
\usepackage{thm-restate}

\usepackage{appendix}

\theoremstyle{plain}

\newtheorem{prop}{Proposition}

\theoremstyle{remark}

\journal{European Journal of Operational Research}

\begin{document}

\begin{frontmatter}

\title{Direction Selection in Stochastic Directional Distance Functions}

\author[label1]{Kevin Layer}
\address[label1]{Department of Industrial and Systems Engineering, Texas A\&M University, College Station, TX, USA.}

\author[label1,label2]{Andrew L. Johnson\corref{cor1}}
\address[label2]{School of Information Science and Technology, Osaka University, Suita, Japan.}
\cortext[cor1]{Communicates regarding this paper can be sent to Andrew L. Johnson, ajohnson@tamu.edu}

\author[label3]{Robin C. Sickles}
\address[label3]{Department of Economics, Rice University, Houston, TX, USA.}

\author[label4]{Gary D. Ferrier}
\address[label4]{Department of Economics, University of Arkansas, Fayetteville, AR, USA.}

\begin{abstract}
Researchers rely on the distance function to model multiple product production using multiple inputs. A stochastic directional distance function (SDDF) allows for noise in potentially all input and output variables. Yet, when estimated, the direction selected will affect the functional estimates because deviations from the estimated function are minimized in the specified direction. Specifically, the parameters of the parametric SDDF are point identified when the direction is specified; we show that the parameters of the parametric SDDF are set identified when multiple directions are considered. Further, the set of identified parameters can be narrowed via data-driven approaches to restrict the directions considered. We demonstrate a similar narrowing of the identified parameter set for a shape constrained nonparametric method, where the shape constraints impose standard features of a cost function such as monotonicity and convexity. 

Our Monte Carlo simulation studies reveal significant improvements, as measured by out of sample radial mean squared error, in functional estimates when we use a directional distance function with an appropriately selected direction and the errors are uncorrelated across variables. We show that these benefits increase as the correlation in error terms across variables increase. This correlation is a type of endogeneity that is common in production settings. From our Monte Carlo simulations we conclude that selecting a direction that is approximately orthogonal to the estimated function in the central region of the data gives significantly better estimates relative to the directions commonly used in the literature. For practitioners, our results imply that selecting a direction vector that has non-zero components for all variables that may have measurement error provides a significant improvement in the estimator's performance. We illustrate these results using cost and production data from samples of approximately 500 US hospitals per year operating in 2007, 2008, and 2009, respectively, and find that the shape constrained nonparametric methods provide a significant increase in flexibility over second order local approximation parametric methods.    
\end{abstract}

\begin{keyword}
Nonparametric regression \sep Shape Constraints \sep Data Envelopment Analysis \sep Hospital production.
\end{keyword}

\end{frontmatter}

\section{Introduction}
\label{sec:1.Intro}

The focus of this paper is direction selection in stochastic directional distance functions (SDDF).\footnote{Here we use the term stochastic in reference to a model with a noise term.} While  the DDF is typically used to measure efficiency, in this paper we use a nonparametric shape constrained SDDF to model the conditional mean behavior of production. The stochastic distance function (SDF) was introduced by \cite{lovell1994} and was used in a series of early empirical studies by \cite{coelli1999comparison, coelli2000technical} and \cite{sickles2002}. The parameters of a parametric distance function are point identified; however, if the direction in the DDF is not specified, then the parameters of a parametric DDF are set identified.\footnote{Let $\phi$ be what is known (e.g., via assumptions and restrictions) about the data generating process (DGP).  Let $\theta$ represent the parameters to be identified, let $\Theta$ denote all possible values of $\theta$, and let $\theta_0$ be the true but unknown value of $\theta$.  Then the vector $\theta$ of unknown parameters is point identified if it is uniquely determined from $\phi$. However, $\theta$ is set identified if some of the possible values of $\theta$ are observationally equivalent to $\theta_0$ (\cite{lewbel2016identification}).} A set of axiomatic properties related to production and cost functions, such as monotonicity and convexity in the case of a cost function, are well established in the production literature (\cite{shephard1970,chambers1988applied}). Although the stochastic distance function literature acknowledges the axiomatic properties necessary for duality, it does not impose them globally. Instead, authors typically impose them only on a particular point in the data (e.g., \cite{Atkinson2003}). Recognizing these issues, we provide an axiomatic nonparametric estimator of the SDDF and a method to restrict the pool of the directions to choose from for the SDDF, thereby reducing the size of the set identified parameter set. 
Most empirical studies that use establishment or hospital level data to estimate production or cost functions either assume a specific parametric form or ignore noise, or both (\citep{hollingsworth2003non}). 
In contrast, we use an axiomatic nonparametric SDDF estimator and the proposed method to determine a set of acceptable directions to estimate a cost function that maintains global axiomatic properties for the US hospital industry. Furthermore, we demonstrate the importance of global axiomatic properties for the estimation of most productive scale size and marginal costs.

A few papers have attempted to implement the directional distance function in a stochastic setting (see, for example, \cite{fare2005}, \cite{fare2010functional}, and \cite{fare2016}). The latter two papers discuss the challenges of selecting a parametric functional form that does not violate the axioms typically assumed in production economics. Based on their observations, \cite{fare2016} use a quadratic functional specification.\footnote{As \cite{kuosmanenjohnson2017} note, the translog function used for multi-output production cannot satisfy the standard assumptions for the production technology $T$ globally for any parameter values. The quadratic functional form does not have this shortcoming.} Yet several papers show a loss of flexibility in parametric functional forms, such as the translog or the quadratic functional form, when shape constraints are imposed (e.g., \cite{DiewertWales1987}). Also important to implementation, the selection of the direction vector in the SDDF has been discussed in \cite{FarePasurkaVardanyan2017} and \cite{AtkinsonTsionas2016}, among others. These papers focus on selecting the direction corresponding to a particular interpretation of the inefficiency measure, based on the distance to the economically efficient point. In contrast, we consider \cite{kuosmanenjohnson2017}'s multi-step efficiency analysis and focus on the first step, estimating a conditional mean function. Our goal is to select the direction that best recovers the underlying technology while acknowledging that the data is likely to contain noise in potentially all variables.\footnote{For researchers interested in productivity measurement and productivity variation (e.g., \cite{syverson2011determines}), the results from this paper can be used directly. For authors interested in efficiency analysis, the insights from this paper could be used to improve the estimates from the first stage of \cite{kuosmanenjohnson2017}'s three-step procedure where efficiency is estimated in the third step. }

To model multi-product production,  \cite{kuosmanenjohnson2017} have proposed the use of axiomatic nonparametric methods to estimate the SDDF which they name Directional Convex Nonparametric Least Squares (CNLS-d), a type of sieve estimator. Their methods have the benefits of relaxing standard functional form assumptions for production, cost, or distance functions, but also improve the interpretability and finite sample efficiency over nonparametric methods such as kernel regression (\cite{yagi2018shape}). A variety of models can be interpreted as special cases of \cite{kuosmanenjohnson2017}, among these are a set of models that specify the direction (e.g., \cite{johnson2011one,kuosmanen2012stochastic}). All CNLS models are sieve estimators and fall into the category of partially identified or set identified estimators discussed in \cite{manski2003partial} and \cite{Tamer2010PartialID}. The guidance our paper provides in selecting a direction will reduce the size of the set identified for CNLS-d and other DDF estimators with flexible direction specifications.
 
Much of the production function literature concerns endogeneity issues, for example see \cite{OlleyPakes96}, \cite{levinsohn2003estimating}, and \cite{ackerberg2015}. These methods are often referred to as proxy variable approaches. The argument for endogeneity is typically that decisions regarding variable inputs such as labor are made with some knowledge of the factors included in the unobserved residuals. Recently, these methods have been reinterpreted as instrumental variable approaches (\cite{wooldridge2009}), or control function approaches (\cite{ackerberg2015}). Unfortunately, the assumptions on the particular timing of input decisions is not innocuous. Indeed every firm must adjust its inputs in exactly the same way, otherwise the moment restrictions needed for point identification are violated. For an alternative in the stochastic frontier setting, see \cite{kutlu2018distribution}.   

\cite{kuosmanenjohnson2017} have shown that a production function estimated using a stochastic distance function under a constant returns-to-scale assumption is robust to endogeneity issues because the normalization by one of the inputs or outputs causes the errors-in-variables to cancel each other. In this paper we consider the more general case of a convex technology that does not necessarily satisfy constant returns-to-scale, and show that when errors across variables are highly correlated, a specific type of endogeneity, the SDDF improves estimation performance significantly over the typical alternative of ignoring the endogeneity.

When considering alternative directions in the DDF, we show that the direction that performs the best is often related to the particular performance measure used. We use an out-of-sample mean squared error (MSE) that is measured radially to address this issue. This measure is motivated by the results of our Monte Carlo simulations and is natural for a function that satisfies monotonicity and convexity, assuring the true function and the estimated function are close in the areas were most data are observed. 

We analyze US hospital data and characterize the most productive scale size and marginal costs for the US hospital sector. We demonstrate that out-of-sample MSE is reduced significantly by relaxing parametric functional form restrictions. We also observe the advantage of imposing axioms that allow the estimated function to still be interpretable. Concerning the direction selection, we find, for this data set, that the exact direction selected is not very critical in terms of MSE performance, but some commonly used directions should be avoided.

The remainder of this paper is organized as follows. Section 2 introduces the statistical model and the production model. Section 3 describes the estimators used for the analysis. Section 4 outlines our reasons for the MSE measure we propose. Section 5 highlights the importance of the direction selection through Monte Carlo experiments. Section 6 describes our direction selection method. Section 7 demonstrates the benefits of using non-parametric shape-constrained estimators with an appropriately selected direction for US hospital data. Section 8 concludes.

\section{Models}
\label{sec:2.Mods}

\subsection{Statistical Model}
\label{sec:2.1.StatMod}

We consider a statistical model that allows for measurement error in potentially all of the input and output variables. Let $\bm{\tilde{x}}_i\in \bm{X}\subset \mathbb{R}_+^d, d\geq 1$, be a vector of random input variables of length $d$ and $\bm{\tilde{y}}_i\in \bm{Y}\subset \mathbb{R}_+^Q$, $Q\geq 1$, be a vector of random output variables of length $Q$, where $i$ indexes observations. Let $\bm{\epsilon}^x_{i}\in  \mathbb{R}^d$, $d\geq 1$, be a vector of random error variables of length $d$ and $\bm{\epsilon}^y_{i} \in \mathbb{R}^Q$, $Q\geq 1$, be a vector of random error variables of length $Q$. 
One way of modeling the errors-in-variable (EIV) is:
\begin{equation}
    \label{eq:prod1}
    \binom{\bm{x}_i}{\bm{y}_i} = \binom{\bm{\tilde{x}}_i}{\bm{\tilde{y}}_i} + \binom{\bm{\epsilon}^x_i}{\bm{\epsilon}^y_i}.
\end{equation}
Equation \eqref{eq:prod1} is only identified when multiple measurements exist for the same vector of regressors or when a subsample of observations exists in which the regressors are measured exactly (\cite{Carroll2006}). \cite{Carroll2006} discussed a standard regression setting, not a multi-input/multi-output production process. Thus, repeated measurement requires all but one of the netputs to be identical across at least two observations.\footnote{Here we use the term \textit{netputs} to describe the union of the input and output vectors.} Neither of of these conditions is likely to hold for typical production data sets; therefore, we develop an alternative approach to identification. 

As our starting point, we use the alternative, but equivalent, representation of the EIV model proposed by \cite{kuosmanenjohnson2017}:

\begin{equation}
\label{eq:prod2}
    \binom{\bm{x}_i}{\bm{y}_i} = \binom{\bm{\tilde{x}}_i}{\bm{\tilde{y}}_i} + e_i \binom{\bm{g}_i^x}{\bm{g}_i^y}.
\end{equation}
Clearly, the representations of \cite{Carroll2006} and \cite{kuosmanenjohnson2017} are equivalent if:
\begin{equation}
\label{eq:EIVequivalence}
    \binom{\bm{\epsilon}^x_i}{\bm{\epsilon}^y_i} = e_i \binom{\bm{g}_i^x}{\bm{g}_i^y}.
\end{equation}
We define the following normalization:
\begin{equation}
\label{eq:EIVnorm}
e_i = \sqrt{\sum_{j=1}^{d}{(\epsilon^x_{ij})^2} + \sum_{j=1}^{Q}{(\epsilon^y_{ij})^2}},
\end{equation}
which implies:
\begin{equation}
\label{eq:EIVnorm2}
\sqrt{\sum_{j=1}^{d}{(g^x_{ij})^2} + \sum_{j=1}^{Q}{(g^y_{ij})^2}} = 1.
\end{equation}

\noindent We refer to $(\bm{g}_i^x,\bm{g}_i^y)$ as the \textit{true} noise direction and in the most general case we allow the direction to be observation specific.\footnote{When the noise direction is observation specific and random, all inputs and outputs potentially contain noise and therefore are endogeneous variables. If some components of the $(\bm{g}^x,\bm{g}^y)$ vector are zero, this implies the associated variables are exogeneous and measured with certainty. See \cite{kuosmanenjohnson2017} for more details.} 
 The estimation methods to consider noise in potentially all inputs will depend on our assumptions about the production technology, which are discussed in the following subsection. 

\bigskip
\subsection{Production Model} 
\label{sec:2.2.ProdMod}


Researchers use production function models, cost function models, or distance function models to characterize production technologies. Considering a general production process with multiple inputs used to produce multiple outputs, we define the production possibility set as:

\begin{equation}
\label{eq:DefTechno}
T = \left \{\left(\bm{\tilde{x}},\bm{\tilde{y}}\right) \in \mathbb{R}_{+}^{d+Q} \ | \ \bm{\tilde{x}} \ \text{can produce} \ \bm{\tilde{y}}\right\}.
\end{equation}

\noindent Following \cite{shephard1970}, we adopt the following standard assumptions to assure that $T$ represents a production technology:
\begin{enumerate}
\item{\textit{T} is closed};
\item{\textit{T} is convex};
\item{Free Disposability of inputs and outputs; i.e., if \(\left(\bm{\tilde{x}}^l,\bm{\tilde{y}}^l\right) \in T\)} and \(\left(\bm{\tilde{x}}^k,-\bm{\tilde{y}}^k\right) \geq \left(\bm{\tilde{x}}^l, -\bm{\tilde{y}}^l\right)\), then \(\left(\bm{\tilde{x}}^k,\bm{\tilde{y}}^k\right) \in T\).
\end{enumerate}
For an alternative representation, see, for example, \cite{frisch1964theory}.

Developing methods to estimate characteristics of the production technology while imposing these standard axioms was a popular and fruitful topic from the early 1950's until the early 1980's, generating such classic papers as \cite{koopmans1951}, \cite{shephard1953,shephard1970}, \cite{afriat1972efficiency}, \cite{charnes1978measuring},\footnote{Data Envelopment Analysis is perhaps one of the largest success stories and has become an extremely popular method in the OR toolbox for studying efficiency.} and  \cite{varian1984nonparametric}. Unfortunately, these methods are deterministic in the sense that they rely on a strong assumption that the data do not contain any measurement errors, omitted variables, or other sources of random noise. Furthermore, for some research communities linear programs were seen as harder to implement than parametric regression which could be calculated via normal equations. Thus, most econometricians and applied economists have chosen to use parametric models, sacrificing flexibility for ease of estimation and the inclusion of noise in the model.

Here we focus our attention on the distance function because it allows the joint production of multi-outputs using multi-inputs. The production function and cost functions can be seen as special cases of the distance function in which there is either a single output or a single input (cost), respectively. Further, motivated by our discussion of EIV models above, we consider a directional distance function which allows for measurement error in potentially all variables. We try to relax both the parametric and deterministic assumptions common in earlier approaches to modeling multi-output/multi-input technologies. We do this by building on an emerging literature that revisits the axiomatic nonparametric approach incorporating standard statistical structures including noise (\cite{kuosmanen2008representation};\cite{Kuosmanen2010}). 

\subsubsection{The Deterministic Directional Distance Function (DDF)}
\label{sec:2.2.1.DDF}

\cite{luenberger1992} and \cite{chambers1996distance,chambers1998DDF} introduced the directional distance function, defined for a technology \textit{T} as: 
\begin{equation}
\label{eq:DefDDF}
\overrightarrow{D}_T\left(\bm{\tilde{x}},\bm{\tilde{y}};\bm{g}^x,\bm{g}^y\right) = \max{\{\delta \in \mathbb{R}: \left(\bm{\tilde{x}} - \delta \; \bm{g}^x,\ \bm{\tilde{y}}+\delta \; \bm{g}^y\right) \in T\}}, 
\end{equation}
where $\bm{\tilde{x}}$ and $\bm{\tilde{y}}$ are the observed input and output vectors, such that   \(\bm{\tilde{x}} \in \mathbb{R}_+^{d}\) and \(\bm{\tilde{y}} \in \mathbb{R}_+^{Q}\) are assumed to be observed without noise and fully describe the resources used in production and the goods or services generated from production.  \(\bm{g}^x \in \mathbb{R}_+^{d}\) is the direction vector in the input space, \(\bm{g}^y \in \mathbb{R}_+^{Q}\) is the direction vector in the output space, and  \(\left(\bm{g}^x,\bm{g}^y\right) \in \mathbb{R}_+^{d+Q}\) defines the direction from the point \(\left(\bm{\tilde{x}},\bm{\tilde{y}}\right)\) in which the distance function is measured.\footnote{We assume \(\left(\bm{g}^x,\bm{g}^y\right) \neq \bm{0}\); i.e., at least one of the components of either $\bm{g}^x$ or $\bm{g}^y$ is non-zero.} $\delta$ is commonly interpreted as a measure of inefficiency by quantifying the number of bundles of size $\left(\bm{g}^x,\bm{g}^y\right)$ needed to move the observed point $\left(\bm{\tilde{x}},\bm{\tilde{y}}\right)$ to the boundary of the technology in a deterministic setting. 


\cite{chambers1998DDF} explained how the directional distance function characterizes the technology \textit{T} for a given direction vector \(\left(\bm{g}^x,\bm{g}^y\right) \); specifically:
\begin{equation}
\label{eq:DDFTechnology}
\overrightarrow{D}_T\left(\bm{\tilde{x}},\bm{\tilde{y}};\bm{g}^x,\bm{g}^y\right) \geq 0, \text{if and only if} 
\left(\bm{\tilde{x}},\bm{\tilde{y}}\right) \in T
.
\end{equation}

\noindent If \textit{T} satisfies the assumptions stated in Section \ref{sec:2.2.ProdMod}, then the directional distance function \(\overrightarrow{D}_T: \mathbb{R}_+^{d} \times \mathbb{R}_+^{Q} \times \mathbb{R}_+^{d} \times \mathbb{R}_+^{Q} \to \mathbb{R}_+\) has the following properties (see \cite{chambers1998DDF}):

\begin{enumerate}[label=(\alph*)]
\item{\(\overrightarrow{D}_T\left(\bm{\tilde{x}},\bm{\tilde{y}};\bm{g}^x,\bm{g}^y\right)\) is upper semicontinuous in \(\bm{\tilde{x}}\) and \(\bm{\tilde{y}}\) (jointly);}
\item{\(\overrightarrow{D}_T\left(\bm{\tilde{x}},\bm{\tilde{y}};\lambda \  \bm{g}^x,\lambda \ \bm{g}^y\right) = \left(1/\lambda\right)\overrightarrow{D}_T\left(\bm{\tilde{x}},\bm{\tilde{y}};\bm{g}^x,\bm{g}^y\right), \lambda > 0\);}
\item{\(\bm{\tilde{y}'} \geq \bm{\tilde{y}} \Rightarrow 
\overrightarrow{D}_T\left(\bm{\tilde{x}},\bm{\tilde{y}'};\bm{g}^x,\bm{g}^y\right) \leq \overrightarrow{D}_T\left(\bm{\tilde{x}},\bm{\tilde{y}};\bm{g}^x,\bm{g}^y\right)\);}
\item{\(\bm{\tilde{x}'} \geq \bm{\tilde{x}} \Rightarrow 
\overrightarrow{D}_T\left(\bm{\tilde{x}'},\bm{\tilde{y}};\bm{g}^x,\bm{g}^y\right) \geq \overrightarrow{D}_T\left(\bm{\tilde{x}},\bm{\tilde{y}};\bm{g}^x,\bm{g}^y\right)\);}
\item{If T is convex, then \(\overrightarrow{D}_T\left(\bm{\tilde{x}},\bm{\tilde{y}};\bm{g}^x,\bm{g}^y\right)\) is concave in \(\bm{\tilde{x}} \ \text{and} \ \bm{\tilde{y}}\).}
\end{enumerate}
An additional property of the DDF is the translation invariance:
\begin{enumerate}[resume,label=(\alph*)]
\item{\(\overrightarrow{D}_T\left(\bm{\tilde{x}}-\alpha \bm{g}^x, \bm{\tilde{y}}+\alpha \bm{g}^y; \bm{g}^x,\bm{g}^y  \right) =
\overrightarrow{D}_T\left(\bm{\tilde{x}},\bm{\tilde{y}};\bm{g}^x,\bm{g}^y\right) - \alpha\).}
\end{enumerate}

Several theoretical contributions have been made to extend the deterministic DDF, see for example \cite{Fare2010,Aparicio2017,Kapelko2017}, and \cite{Roshdi2018}. The deterministic DDF has been used in several recent applications, including \cite{Balezentis2015,Adler2016}, and \cite{Fukuyama2018}.

\subsubsection{The Stochastic Directional Distance Function}
\label{sec:2.2.2.StochDistFunc}
 
The properties of the deterministic DDF also apply for the stochastic DDF (\cite{FarePasurkaVardanyan2017}). Here we focus on estimating a stochastic DDF considering a residual which is mean zero.\footnote{Two models are possible, 1) a mean zero residual indicating that the residual contains only noise used to pursue a productivity analysis, or 2) a composed residual with both inefficiency and noise. Our direction selection analysis is used in the first step of Kuosmanen and Johnson's three step procedure in which a conditional mean is estimated.}
This is represented in Figure \ref{fig:SDDF_Ex}.

\begin{figure}[h!]
	\begin{center}
	    \includegraphics[width=3in]{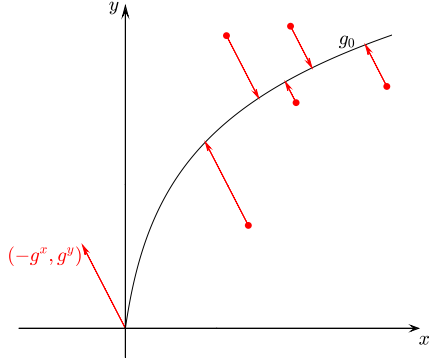}
	\end{center}
	\caption{SDDF in mean zero case}
	\label{fig:SDDF_Ex}
\end{figure}

Using the statistical model in Section \ref{sec:2.1.StatMod} and the functional representation of technology in Section \ref{sec:2.2.ProdMod}, we restate Proposition 2 in \cite{kuosmanenjohnson2017} as:

	\begin{prop}
		\label{prop2}
		If the observed data are generated according to the statistical model described in Section \ref{sec:2.1.StatMod}, then the value of the DDF in the observed data point $(\bm{x}_i,\bm{y}_i)$ is equal to the realization of the random variable $\epsilon_i$ with mean zero; specifically 
		\[\overrightarrow{D}_T(\bm{x}_i,\bm{y}_i,\bm{g}^x,\bm{g}^y)=\epsilon_i \ \forall i. \]
	\end{prop}

In the stochastic distance function literature, the translation property, (f) above, is commonly invoked to move an arbitrarily chosen netput variable out of the distance function to the left-hand side of the equation, yielding an equation that looks like a standard regression model; see, for example, \cite{lovell1994} and \cite{kuosmanenjohnson2017}. Instead, we write the SDDF with all of the outputs on one side to emphasize that all netputs are treated symmetrically.

Under the assumption of constant returns to scale, normalizing by one of the netputs causes the noise terms to cancel for the regressors, thus eliminating the issue of endogeneity (e.g., \cite{coelli2000econometric}, \cite{kuosmanenjohnson2017}). However, since we relax the constant returns to scale assumption, endogeneity can still be an issue.\footnote{If the endogeneity is caused by correlations in the errors across variables, it can be addressed by selecting an appropriate direction for the directional distance function. This is the direction we explore in the Monte Carlo simulation below in Section \ref{sec:4.1.ExpIllustrate}.}

\cite{FarePasurkaVardanyan2017}, among others, have recognized that the selection of the direction vector affects the parameter estimates of the production function. In \ref{sec:Ax.A.1.ImpofDisc}, for the linear parametric DDF defined below, we prove that alternative directions lead to distinct parameter estimates. 


\section{Estimation}
\label{sec:3.Est}

We now describe the estimation of the DDF under a specific parametric functional form and under nonparametric shape constrained methods. 

\subsection{Parametric Estimation and the DDF}
\label{sec:3.1.DescParamDir}

Consider data composed of \(n\) observations where the inputs are defined by \(\bm{x_i}, \ i = 1,...,n\) and the outputs by \(\bm{y_i}, \ i = 1,...,n\). The estimator minimizes the squared residuals for a DDF with an arbitrary prespecified direction \(\left(-\bm{g}^x,\bm{g}^y\right)\). For a linear production function, we formulate the estimator as:

\begin{subequations} 
\label{eq:ParamDDF}
    \begin{align}
       & & \min\limits_{\alpha, \bm{\beta}, \bm{\gamma}, \bm{\epsilon}}{\sum_{i=1}^{n}{\epsilon_i^2}} \tag{\ref{eq:ParamDDF}}\\
        && \text{s.t.} \ \ \
         \bm{\gamma}'\,\bm{y_i}  = \alpha + \bm{\beta}' \, \bm{x_i}  - \epsilon_i , 
        &\ \text{for} \ i = 1,\ldots,n \label{subeq:P0.expr}\\
        && \bm{\beta}' \, \bm{g}^x + \bm{\gamma}'\,\bm{g}^y  = 1 ,
        &\  \label{subeq:P1.TransProp}
    \end{align}
\end{subequations}

\noindent where \(\alpha\) is the intercept, \(\bm{\beta}\) and \(\bm{\gamma}\) are the vectors of the marginal effects of the inputs and outputs, respectively, and the \(\epsilon_i, \, i = 1,...,n\) are the residuals.    

Equation (\ref{subeq:P1.TransProp}) enforces the translation property described in \cite{chambers1998DDF}; i.e., scaling the netput vector by \(\delta\) in the direction $(-\bm{g}^x,\bm{g}^y)$ causes the distance function to decrease by \(\delta\). The combination of Equation (\ref{subeq:P0.expr}) and Equation (\ref{subeq:P1.TransProp}) ensures that the residual is computed along the direction $(-\bm{g}^x,\bm{g}^y)$. Intuitively this is because the $\bm{\beta}$ and $\bm{\gamma}$ are rescaled proportionally to the direction $(-\bm{g}^x,\bm{g}^y)$ in Equation (\ref{subeq:P1.TransProp}). For a formal proof, see \cite{kuosmanenjohnson2017}, Proposition 2.

\subsection{The CNLS-d Estimator}
\label{sec:3.2.DescCNLSd}

Convex Nonparametric Least Squares (CNLS) is a non-parametric estimator that imposes the axiomatic properties, such as monotonicity and concavity, on the production technology. The estimator CNLS-d is the directional distance function generalization of CNLS (\cite{hildreth1954point}, \cite{kuosmanen2008representation}). While CNLS allows for just a single output, CNLS-d permits  multiple outputs. In CNLS the direction along which residuals are computed is specified \textit{a priori} and is typically measured in terms of the unique output, \(\bm{y}\).  This corresponds to the assumption that noise is only present in \(\bm{y}\) and that all other variables, \(\bm{\tilde{x}}\), do not contain noise. CNLS-d allows the residual to be measured in an arbitrary prespecified direction. If all components of the direction vector are non-zero, this corresponds to an assumption that noise is present in all inputs.

Using the same input-output data defined in Section \ref{sec:2.1.StatMod}, the CNLS-d estimator is given by:

\begin{subequations} 
\label{eq:CNLSd}
    \begin{align}
       & & \min\limits_{\bm{\alpha}, \bm{\beta}, \bm{\gamma}, \bm{\epsilon}}{\sum_{i=1}^{n}{\epsilon_i^2}} \tag{\ref{eq:CNLSd}}\\
        && \text{s.t.} \ \ \
         \bm{\gamma_i}'\,\bm{y_i}  = \alpha_i + \bm{\beta_i}' \, \bm{x_i}  - \epsilon_i , 
        &\ \text{for} \ i = 1,\ldots,n \label{subeq:0.expr}\\
        && \alpha_i + \bm{\beta_i}' \, \bm{x_i}  - \bm{\gamma_i}'\,\bm{y_i}  \leq \alpha_j + \bm{\beta_j}' \, \bm{x_i}  - \bm{\gamma_j}'\,\bm{y_i}, 
        &\ \text{for} \ i,j = 1,\ldots,n, \, i \neq j \label{subeq:1.AfriatIneq}\\  
        &&  \bm{\beta_i}  \geq 0, 
        &\ \text{for} \ i = 1,\ldots,n \label{subeq:2.IneqBeta}\\
        && \bm{\beta_i}' \, \bm{g}^x + \bm{\gamma_i}'\,\bm{g}^y  = 1, 
        &\ \text{for} \ i = 1,\ldots,n \label{subeq:3.TransProp}\\
        && \bm{\gamma_i} \geq 0, 
        &\ \text{for} \ i = 1,\ldots,n, \label{subeq:4.IneqGamma}       
    \end{align}
\end{subequations}

\noindent where \(\alpha_i, \, i = 1,...,n\) is the vector of the intercept terms, \(\bm{\beta_i}, \, i = 1,..,n\) and \(\bm{\gamma_i}, \, i = 1,..,n\) are the matrices of the marginal effects of the inputs and the outputs, respectively, and \(\epsilon_i, \, i = 1,...,n\) is the vector of the residuals \citep{kuosmanenjohnson2017}.  

Equation \eqref{subeq:0.expr} is similar to \eqref{subeq:P0.expr} with the notable different that $(\alpha_i,\bm{\beta_i},\bm{\gamma_i})$ are indexed by $i$ indicating each observation has their own hyperplane defined by the triplet $(\alpha_i,\bm{\beta_i},\bm{\gamma_i})$. Equation \eqref{subeq:1.AfriatIneq}, which corresponds to the Afriat inequalities, imposes concavity. Given Equation \eqref{subeq:1.AfriatIneq}, Equation \eqref{subeq:2.IneqBeta} imposes the monotonicity of the estimated frontier relative to the inputs. Equation \eqref{subeq:3.TransProp} enforces the translation property described in \cite{chambers1998DDF} and has the same interpretation as Equation \eqref{subeq:P1.TransProp}. Similar to Equation \eqref{subeq:2.IneqBeta}, the combination of Equation \eqref{subeq:1.AfriatIneq} and Equation \eqref{subeq:4.IneqGamma}  imposes the monotonicity of the DDF relative to the outputs. In Equation \eqref{eq:CNLSd}, we specify the CNLS-d estimator with a single common direction, \(\left(-\bm{g}^x,\bm{g}^y\right)\).\footnote{Alternatively, some researchers may be interested in using observation specific directions or perhaps group specific directions (\cite{daraio2016efficiency}). In \ref{sec:Ax.A.2.ImpofDisc}, we derive the conditions under which multiple directions can be used in CNLS-d while still maintaining the axiomatic property of global convexity of the production technology. Consider two groups each with their own direction used in the directional distance function. Essentially, the convexity constraint holds as long as the noise is orthogonal to the difference of the two directions used in the estimation. A simple example of this situation is all the noise being in one dimension and the difference between the two directions for this dimension is zero. However, this condition is restrictive when noise is potentially present in all variables. Thus, specifying multiple directions in CNLS-d while maintaining the axiomatic properties of the estimator, specifically, the convexity of the production possibility set, is still an open research question.} 


\section{Measuring MSE under Alternative Directions}
\label{sec:4.TrainTest}

\subsection{Illustrative Example}
\label{sec:4.1.ExpIllustrate}

\paragraph{Data Generation Process}
\label{sec:4.1.Par_DGP}

For our illustrative example, we use a simple linear cost function and a directional distance linear parametric estimator. We consider two noise generation processes: a random noise direction and a fixed noise direction. Here we discuss the random noise direction case, but direct the reader to \ref{sec:Ax.B.AdExp} for a discussion of the fixed noise direction case.


For our example we consider a single output cost function where the observations \(\left(y_{i},c_{i}\right) , i = 1,\ldots,n\), are created by the Data Generation Process (DGP) outlined in Algorithm 1: 

\begin{mdframed}[nobreak=true]
Algorithm 1
\begin{enumerate}
    \item{Output, \(\tilde{y}_{i}\),  is drawn from the continuous uniform distribution \(U\left[0,1\right]\)}.
    \item{Cost is calculated as \(\tilde{c}_{i} = \beta_0 \ \tilde{y}_{i}\), where $\beta_0 = 1$.}
   \item{The noise terms, $\epsilon_{y_{i}}, \epsilon_{c_{i}}$,  are constructed as follows:}
            \begin{enumerate}
            \item{
            $\epsilon_0$ is calculated as: 
            \begin{equation}
            \label{eq:ep0}
            \epsilon_0 = \frac{1}{2} \left[ \; \sqrt{\frac{1}{n-1}\sum_{i=1}^n\left(\tilde{y}_{i}-\bar{y}\right)^2} +\sqrt{\frac{1}{n-1}\sum_{i=1}^n\left(\tilde{c}_{i}-\bar{c}\right)^2} \; \right],
            \end{equation}
            \noindent where \(\bar{y} = \frac{1}{n}\sum_{i=1}^{n} {\tilde{y}_{i}}\) and \(\bar{c} = \frac{1}{n}\sum_{i=1}^{n}{\tilde{c}_{i}}\) are the means of the output and cost without noise, respectively.}

            \item{The scalar length of the noise is rescaled by the vector, $v_{q\epsilon_{i}}$, in each dimension. These scaling factors are calculated as \(v_{q\epsilon_{i}} = \frac{v^*_{q\epsilon_{i}}}{\lvert\lvert\bm{v^*_{\epsilon_{i}}}\rvert\rvert_2}, q = \{1,2\}\) where  \(v^*_{q\epsilon_{i}}\) are drawn from a continuous uniform distribution \(U[-1,1]\).}
            
            \item{\(\left(\epsilon_{y_{i}},\epsilon_{c_{i}}\right) = l_{\epsilon_i} \, \bm{v_{\epsilon_i}}, i = 1,\ldots,n\), where \(l_{\epsilon_i}\) is a scalar length drawn from the normal distribution,  \(N\left(0,\lambda\,\epsilon_0\right)\), where \(\lambda\) is prespecified initial value for the standard deviation and \(\bm{v_{\epsilon_i}} = \left[v_{1\epsilon_{i}},v_{2\epsilon_{i}}\right]\)\ is a normalized direction vector.}
            
            \end{enumerate}

    \item{The observations with noise are obtained by appending the noise terms to the generated data: 
    \begin{equation}
    \label{eq:DGPObsAndNoise}
    \binom{y_{i}}{c_{i}} = \binom{\tilde{y}_{i}}{\tilde{c}_{i}} + \binom{\epsilon_{y_{i}}}{\epsilon_{c_{i}}} , i = 1,\ldots,n.
    \end{equation}
    }
\end{enumerate}
\end{mdframed}

\begin{figure}[h]
	\caption{Algorithm 1: Linear function data generation process with random noise directions}
	\label{fig:LinearDGPStatement}
\end{figure}


\noindent Figure \ref{fig:DGPExLinRand+Fix} illustrates the results for two cases of the data generating process; in the first case the direction of the noise is random, while in the second case the direction of the noise is fixed.

\begin{figure}[h]
\centerline{
\includegraphics[width=0.5\textwidth]{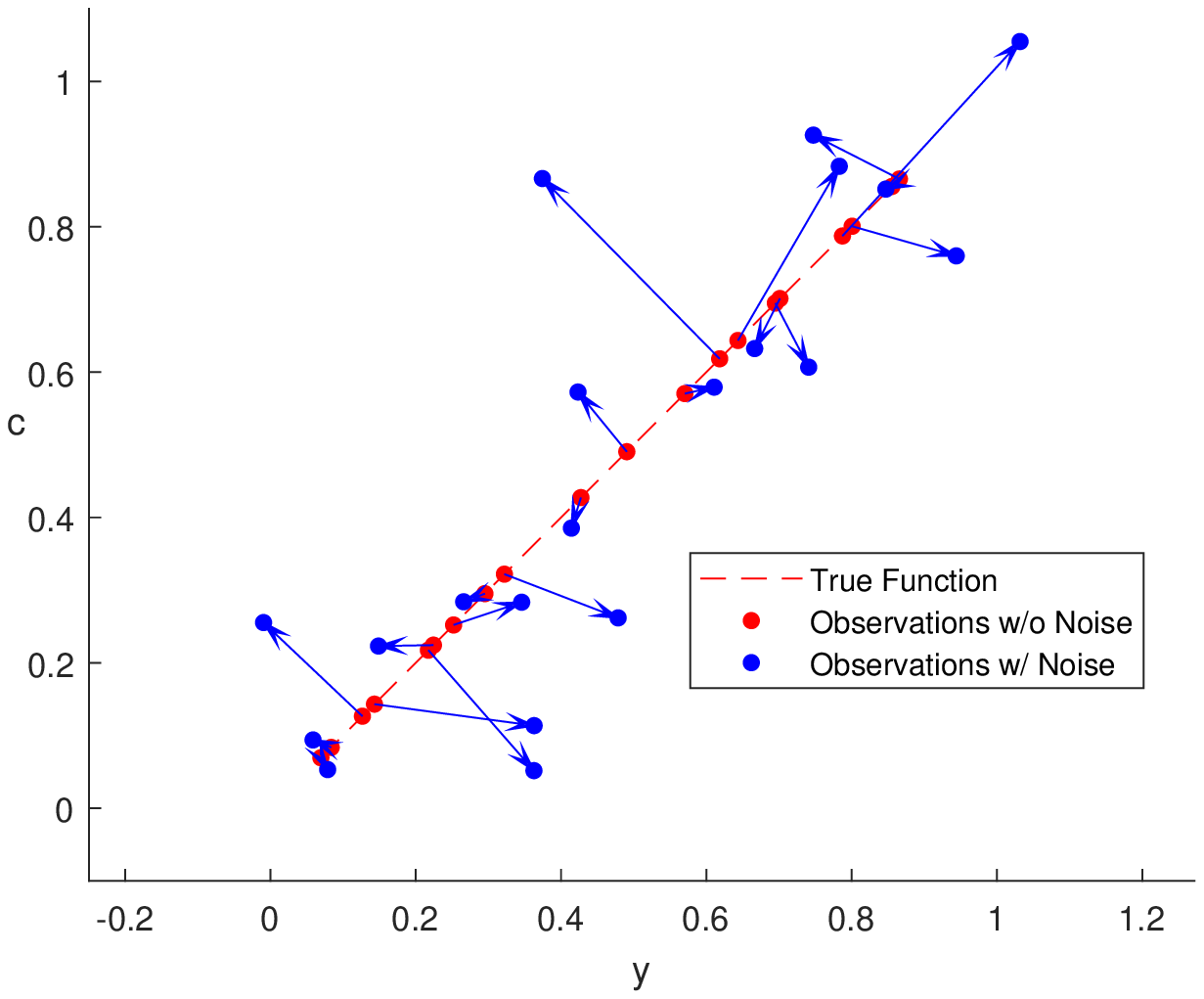}
\includegraphics[width=0.5\textwidth]{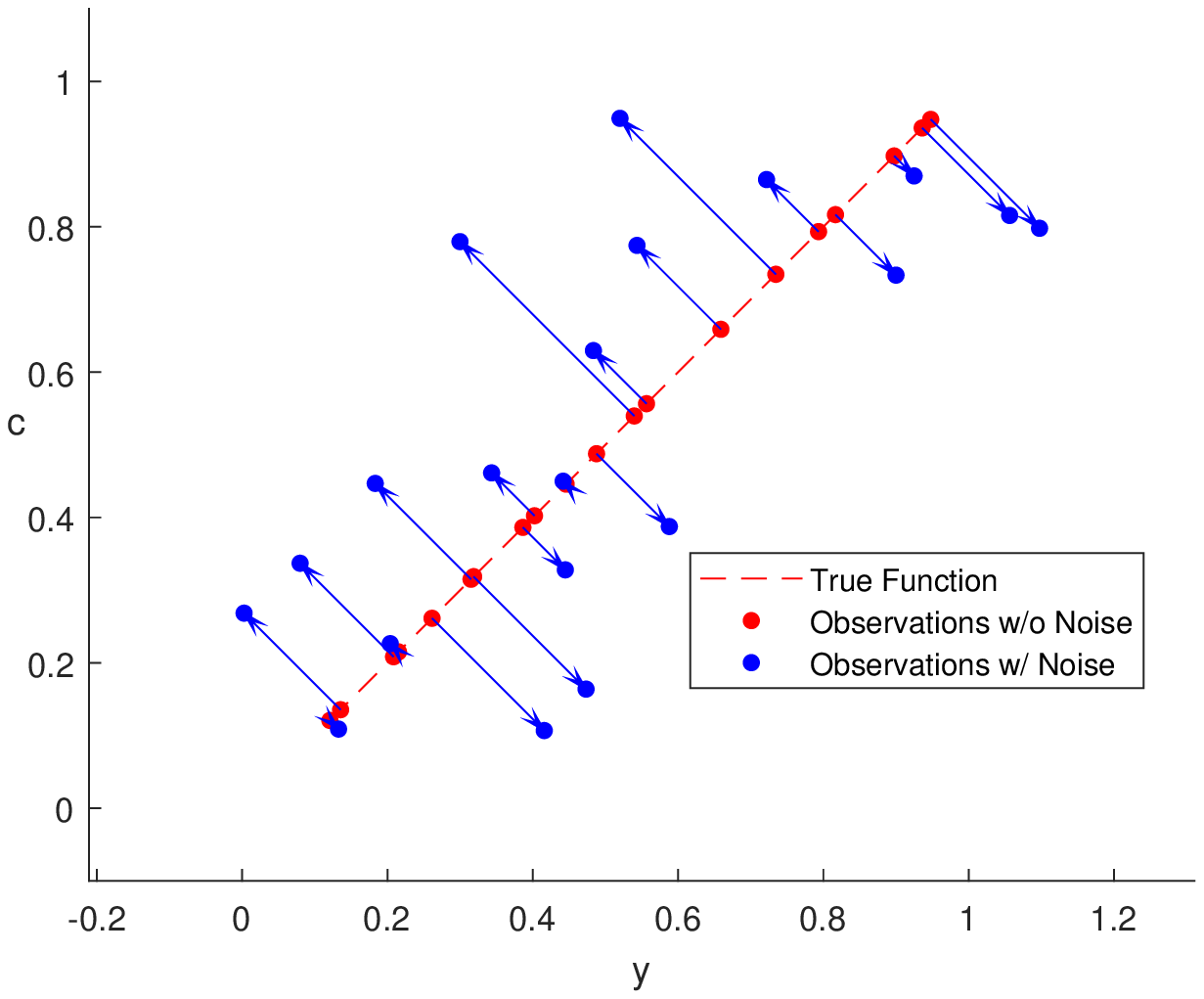}
}
\caption{Linear Case with Random Noise Direction (left), Linear Case with Fixed Noise Direction (right)}
\label{fig:DGPExLinRand+Fix}
\end{figure}



\paragraph{Evaluating the Parametric Estimator's Performance}
\label{sec:4.1.1.MSEComp}

We use two criteria to assess the performance of the parametric estimator: 1) Mean Squared Error (MSE) comparing the true function to the estimated function, and 2) MSE comparing the estimated function to a testing data set. While we can calculate both metrics for our Monte Carlo simulations, only the second metric can be used with our application data below.  

To calculate deviations, we use the MSE direction \(\left(g_{MSE}^y, g_{MSE}^c\right)\). For any particular point of the testing set, \(\left(y_{ts_{i}},c_{ts_{i}}\right), i = 1,\ldots,n\), we determine the estimates, \(\left(\hat{y}_{ts_{i}},\hat{c}_{ts_{i}}\right), i = 1,\ldots,n\) , defined as the intersection of the estimated function characterized by the coefficients \(\left(\hat{\alpha},\hat{\beta}\right)\) and the line passing through  \(\left(y_{ts_{i}},c_{ts_{i}}\right), i = 1,\ldots,n\), and direction vector \(\left(g_{MSE}^y, g_{MSE}^c\right)\). We evaluate the value of the MSE as: 
\begin{equation}
\label{eq:MSEDef}
    \textit{MSE} = \frac{1}{n} 
    \sum_{i=1}^{n} { \left(
    \left(\hat{y}_{ts_{i}} - y_{ts_{i}}\right)^2
    + \left(\hat{c}_{ts_{i}} - c_{ts_{i}}\right)^2
    \right)
    .
}   
\end{equation}

To compare the true function to the estimated function, we use the Linear Function Data Generation Process, Algorithm 1, steps 1 and 2, to construct our testing data set \(\left(y_{ts_{i}},c_{ts_{i}}\right), i = 1,\ldots,n\).
To evaluate the estimated function without knowing the true function the testing set is built using the full Linear Function Data Generation Process.

Figure \ref{fig:LinDiffMSE} show the MSE computations.

\begin{figure}[h]
\centerline{%
\includegraphics[width=0.5\textwidth]{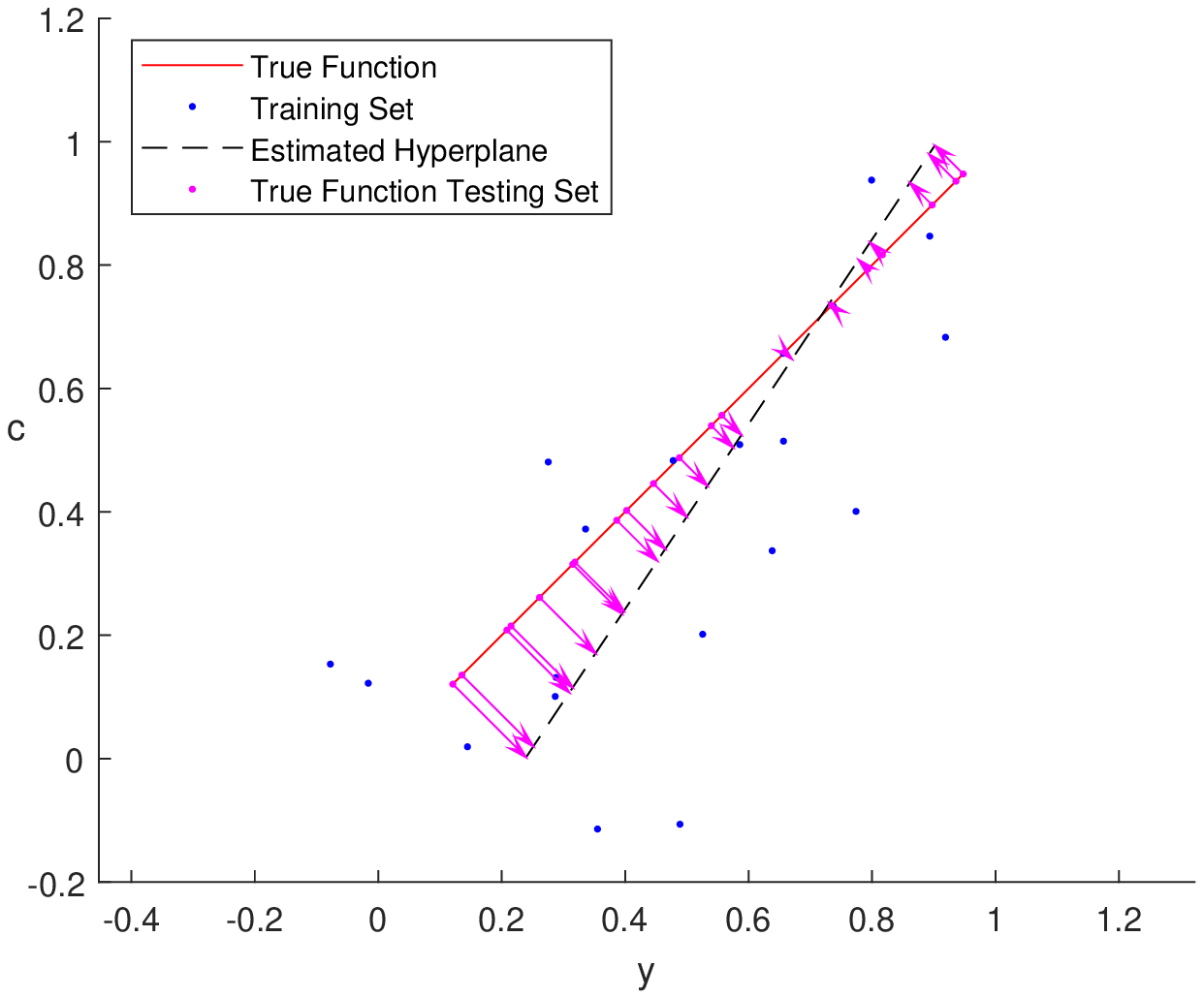}%
\includegraphics[width=0.5\textwidth]{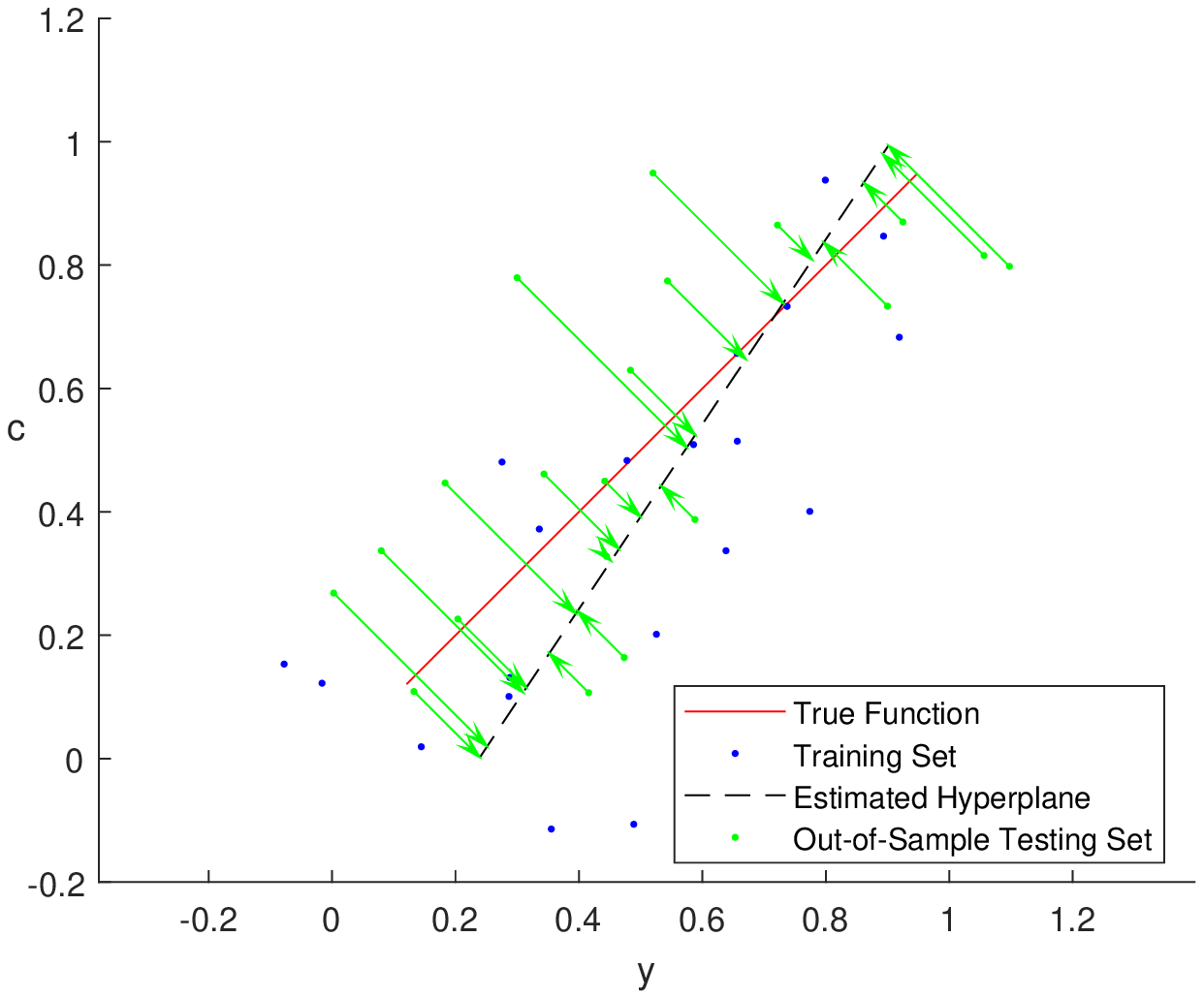}%
}%
\caption{MSE calculated relative to the True Function in the MSE direction \(\pi/4\) (left), MSE calculated using a testing data set in the MSE direction \(\pi/4\) (right)}
\label{fig:LinDiffMSE}
\end{figure}




\paragraph{Additional Information Describing the Simulations}
\label{sec:4.1.InfOnSimu}

We apply the DGP described above to generate a training set, \(\left(y_{tr_{i}},c_{tr_{i}}\right), i = 1,\ldots,n_{tr}\), and a testing set \(\left(y_{ts_{i}},c_{ts_{i}}\right), i = 1,\ldots,n_{ts}\), in which noise is introduced to the observations in random directions.
We set the noise scaling coefficient to $\lambda = 0.6$ and the number of observations to $n_{tr} = n_{ts} = 100$. We run $100$ repetitions of the simulation for each experiment on a computer with a processor Intel Core i7 CPU 860 2.80 GHz and 8 GB RAM. We use the quadratic solver on MATLAB 2017a.  

For the estimator, we define the direction vector used in the parametric DDF as a function of an angular variable $\theta$, which allows us to investigate alternative directions. Specifically, the direction vector used in the DDF is  \(\left(g^y, g^c\right) = \left(\cos(\theta_t),\sin(\theta_t)\right)\). We examine the set of directions corresponding to the angles \(\theta_t \in \left\{0,\ \pi/8,\ \pi/4,\ 3\pi/8,\ \pi/2\right\}\).

\paragraph{Results: Random Noise Directions}
\label{sec:4.1.1.Res}

Table \ref{tab:LinRandomNoiseTrueFunction} and Table \ref{tab:LinRandomNoiseOSMSE} show results corresponding to the two performance criteria introduced above and shown in Figure \ref{fig:LinDiffMSE}, the MSE relative to the true function and the MSE relative to a testing data set, respectively.  Table \ref{tab:LinRandomNoiseTrueFunction} shows that the direction corresponding to the angle \( \pi / 4\), \(\left(g^y=0.707, g^c=0.707\right)\), produces the smallest values of MSE (shown in bold in the table) regardless of the direction used for the MSE computation. However, the estimator's quality diminishes if we select the extreme directions corresponding to the angles $0$ and $\pi/2$. Table \ref{tab:LinRandomNoiseOSMSE} reports performance via a testing set, the direction corresponding to the smallest MSE value (shown in bold) is always the one matching the direction used in the MSE computation. In applications, using a testing set is necessary because the true function is unknown. Table \ref{tab:LinRandomNoiseOSMSE} shows the benefits of matching the direction of MSE evaluation direction outweigh the benefits of selecting a direction based on the properties of the function being estimated.


\begin{table}[h!]
  \centering
  \caption{Average MSE over 100 simulations for the Linear Estimator compared to the true function with a DGP using random noise directions}
    \begin{tabular}{crrrrr}
          \multicolumn{1}{r|}{} & \multicolumn{5}{c}{Avg MSE: Comparison} \\
          \multicolumn{1}{r|}{} & \multicolumn{5}{c}{to the True Function} \\
          \multicolumn{1}{r|}{} & \multicolumn{5}{c}{DDF Angle $\theta_t$ } \\
    \cmidrule{2-6}    
    \multicolumn{1}{c|} {MSE Dir Angle $\theta_{\textit{MSE}}$} & \multicolumn{1}{c}{\(0\)}  &  \multicolumn{1}{c}{\( \pi / 8\)} &   \multicolumn{1}{c}{\( \pi / 4\)}    &  \multicolumn{1}{c}{\( 3\pi / 8\)}  &   \multicolumn{1}{c}{\( \pi / 2\)}   \\
    \midrule
    \multicolumn{1}{c|}{\(0\)} & {2.09} & {0.75} & {\textbf{0.56}} & {1.16} & {3.68} \\
    \multicolumn{1}{c|}{\(\pi/8\)} & {1.36} & {0.46} & {\textbf{0.32}} & {0.63} & {1.89} \\
    \multicolumn{1}{c|}{\(\pi/4\)} & {1.25} & {0.41} & {\textbf{0.28}} & {0.51} & {1.48} \\
    \multicolumn{1}{c|}{\(3\pi/8\)} & {1.59} & {0.50} & {\textbf{0.32}} & {0.57} & {1.60} \\
    \multicolumn{1}{c|}{\(\pi/2\)} & {3.06} & {0.91} & {\textbf{0.55}} & {0.92} & {2.44} \\
    \midrule
    \multicolumn{6}{l}{Note: Displayed are measured values multiplied by \(10^3\).} \\
    \end{tabular}%
  \label{tab:LinRandomNoiseTrueFunction}%
\end{table}%


\begin{table}[h!]
  \centering
  \caption{Average MSE over 100 simulations for the Linear Estimator compared to an out-of-sample testing set with a DGP using random noise directions}
    \begin{tabular}{crrrrr}
          \multicolumn{1}{r|}{} & \multicolumn{5}{c}{Avg MSE: Comparison} \\
          \multicolumn{1}{r|}{} & \multicolumn{5}{c}{to Out-of-Sample} \\
          \multicolumn{1}{r|}{} & \multicolumn{5}{c}{DDF Angle $\theta_t$ } \\
    \cmidrule{2-6}   
    \multicolumn{1}{c|} {MSE Dir Angle $\theta_{\textit{MSE}}$} & \multicolumn{1}{c}{\(0\)}  &  \multicolumn{1}{c}{\( \pi / 8\)} &   \multicolumn{1}{c}{\( \pi / 4\)}    &  \multicolumn{1}{c}{\( 3\pi / 8\)}  &   \multicolumn{1}{c}{\( \pi / 2\)}   \\
    \midrule
    \multicolumn{1}{c|}{\(0\)} & {\textbf{28.28}} & {29.43} & {31.29} & {34.23} & {40.67} \\
    \multicolumn{1}{c|}{\(\pi/8\)} & {18.03} & {\textbf{17.79}} & {18.19} & {19.09} & {21.32} \\
    \multicolumn{1}{c|}{\(\pi/4\)} & {16.38} & {15.55} & {\textbf{15.45}} & {15.77} & {16.90} \\
    \multicolumn{1}{c|}{\(3\pi/8\)} & {20.50} & {18.67} & {18.04} & {\textbf{17.90}} & {18.46} \\
    \multicolumn{1}{c|}{\(\pi/2\)} & {38.63} & {33.07} & {30.68} & {29.29} & {\textbf{28.70}} \\
    \midrule
    \multicolumn{6}{l}{Note: Displayed are measured values multiplied by  \(10^3\).} \\
    \end{tabular}%
  \label{tab:LinRandomNoiseOSMSE}%
\end{table}%

For the out-of-sample testing set, the direction that provides the smallest MSE value is the direction used for the MSE computation. Because the functional estimate is optimized for the direction specified in the SDDF, it is perhaps expected that using the same direction that will be used in the MSE evaluation would produce a relatively low MSE compared to other directions. However, when the functional estimate is compared to the true function, the MSE values are around ten times smaller than the out-of-sample testing case. In out-of-sample testing the presence of noise in the observations causes a deviation regardless of the quality of the estimator or the number of observations. 
The DDF direction corresponding to the smallest MSE is the direction orthogonal to the true function (i.e., \(\pi/4\) for our DGP). This direction provides the shortest distance from the observations to the true function. We conclude that, in this experiment, it is preferable to select a direction orthogonal to the true function (see  Section \ref{sec:5.DirMatters} for further experiments).

From the fixed noise direction experiments (see \ref{sec:Ax.B.AdExp.1}), we observe that using a direction for the estimator that matches the direction used for the noise generation significantly reduces the MSE values compared to the true function. From this, we infer that when endogeneity is severe, using a direction that matches the characteristics of this endogeneity significantly improves the fit of the estimator; i.e., the MSE is $50 \%$ smaller for the matching direction than for the second best direction in $70 \%$ of the cases (see Section \ref{sec:5.DirMatters} for the details).

Finally, we need to solve the problem of evaluating alternative directions when the true function is unknown so that we can evaluate alternative directions in the application data. Below, we describe our proposed alternative measure of fit.



\subsection{Radial MSE Measure}
\label{sec:4.2.MSERadPres}

MSE is typically measured by the average sum of squared errors in the dimension of a single variable, such as cost or output. As explained in Section \ref{sec:4.1.ExpIllustrate}, when we compare out-of-sample performance, we find that the best direction to use in estimating a SDDF is the direction used for MSE evaluation regardless of the direction of noise in the DGP or any other characteristics of the DGP. To avoid this relationship between the direction of estimation and the direction of evaluation, we propose a \textbf{radial MSE measure}.

We begin by normalizing the data to a unit cube and consider a case of \(Q\) outputs and \(n\) observations, where the original observations are:
\[
(y_{i1},\ldots,y_{iQ},c_{i}), \ i=1,\ldots,n.
\]

The normalized observations are:
\begin{eqnarray}
\label{eq:NormObs}
    \breve{y}_{ij}  &=&
    \frac{ y_{ij} - \min_{k}{y_{kj}}}
    {\max_{k}{y_{kj}} - \min_{k}{y_{kj}}}
    , \ j=1,\ldots,Q, \ i,k=1,\ldots,n ,\\
    \breve{c}_{i}  &=&
    \frac{ y_{i} - \min_{k}{c_{k}}}
    {\max_{k}{c_{k}} - \min_{k}{c_{k}}} 
    , \ i,k=1,\ldots,n .
\end{eqnarray}

Our radial MSE measure is the distance from the testing set observation to the estimated function measured along a ray from the testing set observations to the \textbf{center \(C\)}.
Having normalized the data, the center for the radial measure is
\(
C = [ \breve{y}_1, …,\breve{y}_Q,\breve{c}] = \left[
\overbrace{0,\ldots,0}^{Q},1
\right].
\)

The radial MSE measure is the average of the distance from each testing set observation to the estimated function measured radially. Figure \ref{fig:RadMeasureEx} illustrates this measure. For a convex function, a radial measure reduces the bias in the measure for extreme values in the domain.

\begin{figure}[h!]
	\begin{center}
	    \includegraphics[width=4in]{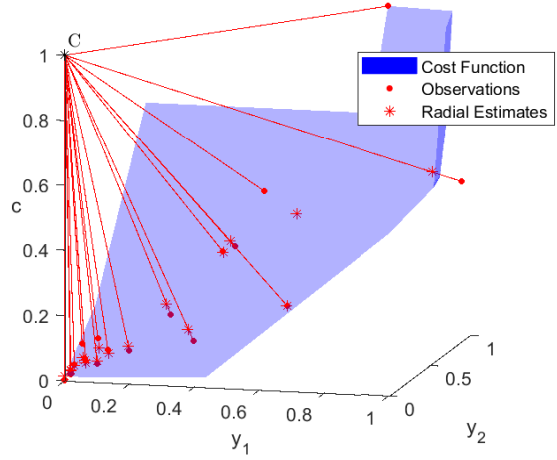}
	\end{center}
	\caption{A Radial MSE Measure on a Cost Function with Two Outputs}
	\label{fig:RadMeasureEx}
\end{figure}

\clearpage

\bigskip
\section{Monte Carlo Simulations}
\label{sec:5.DirMatters}

We next examine how different DGPs affect the optimal direction for the DDF estimator based on a set of Monte Carlo simulations. We consider both random noise directions for each observation and a fixed noise direction representing a high endogeneity case. We consider the effects of the different variance levels for the noise and changes in the underlying distribution of the production data. Using the simplest case of two outputs and a fixed cost level for all observed units allows us to separate the effects of the data and of the function.

\subsection{CNLS-d Formulation for Cost Isoquant Estimation}
\label{sec:CNLSdIso}

Before describing our experiments, we first outline the CNLS-d for estimating the iso-cost level set. It is based on the following optimization problem:

\begin{subequations}
\label{eq:CNLSd_iso}
\begin{align}
&\min_{\bm{\gamma}, \bm{\epsilon}}{\sum_{i=1}^{n}{\epsilon_i^2}} \tag{\ref{eq:CNLSd_iso}}\\
\text{s.t.} \ \ \  -\epsilon_j + \epsilon_i - \bm{\gamma_i}'\,\left(\bm{y_i} - \bm{y_j} \right) \leq 0, & \ \text{for} \ i,j = 1,\ldots,n, \ i \neq j \\
\bm{\gamma_i}'\,\bm{g}^y  = 1, & \ \text{for} \ i = 1,\ldots,n\\
\bm{\gamma_i} \geq 0, & \ \text{for} \ i = 1,\ldots,n.
\end{align}
\end{subequations}

Note all observations, $\bm{y}_i$, have a common cost level. This allows us to focus on a 2-dimensional estimation problem. For results related to 3-dimensional estimation problems see \ref{sec:Ax.B.AdExp.2}, Experiment 6.

We can recover the fitted values, $\hat{y}_i$, and the coefficient, \(\alpha_i, \ i = 1,\ldots,n\), using:
\begin{eqnarray}
\label{eq:CNLSDcoeffs}
    \bm{\hat{y}_i} &=& \bm{y_i} - \epsilon_i \, \bm{g}^y, \ \text{for} \ i = 1,\ldots,n \\
    \alpha_i &=& \bm{\gamma_i}' \, \bm{y_i} + \epsilon_i, \ \text{for} \ i = 1,\ldots,n.
\end{eqnarray}

\subsection{Experiments}
\label{sec:Exps}


We conducted several experiments to investigate the optimal direction for the DDF estimator. Four experiments' results are shown in the main text of the paper with two additional experiments described in the appendix.

\paragraph{Experiment 1 - Base case: A two output circular isoquant with uniformly distributed angle parameters and random noise direction} \mbox{} 
\label{sec:exp1}

For the base case, we consider a fixed cost level and approximate a two output isoquant; i.e., \(Q=2\). Indexing the outputs by $q$ and observations by $i$, we generate the output variables as: 

\begin{equation}
\label{eq:Exp1_TrueVars}
    y_{qi}=\tilde{y}_{qi}+\epsilon_{qi}, \ q=1,\ldots,Q , \ i=1,\ldots,n , 
\end{equation}

\noindent where \(\bm{\tilde{y}_i}\) is the observation on the isoquant and \(\bm{\epsilon_i}\) is the noise. We generate the output levels \(\tilde{y}_{qi}, \ q=1,\ldots,Q \ ,i=1,\ldots,n\) as:

\begin{eqnarray}
\label{eq:Exp1_DetailedOutputs}
    \tilde{y}_{1i} = \cos(\theta_i), \ i=1,\ldots,n \\
    \tilde{y}_{2i} = \sin(\theta_i), \ i=1,\ldots,n,
\end{eqnarray}

\noindent where \(\theta_i, \ i=1,\ldots,n\), is drawn randomly from a continuous uniform distribution, \(U\left[0,\frac{\pi}{2}\right]\). The noise terms, \(\epsilon_{qi}, \ q=1,\ldots,Q ,\ i=1,\ldots,n\), have the following expressions:

\begin{eqnarray}
\label{eq:Exp1_DetailedNoise}
    \epsilon_{1i}= l\; \cos(\theta_{\epsilon_i}), \ i=1,\ldots,n \\
    \epsilon_{2i}= l\; \sin(\theta_{\epsilon_i}), \ i=1,\ldots,n,
\end{eqnarray}

\noindent where the length \(l\) is drawn from the normal distribution \(N\left(0,\lambda \right)\), the angle \(\theta_{\epsilon_i}\) is observation specific and characterizes the noise direction for each observation, and \(\theta_{\epsilon_i}\) is drawn from a continuous uniform distribution \(U\left[-\frac{\pi}{2},\frac{\pi}{2}\right]\). The values considered for the directions in CNLS-d estimator are \(\theta_{\textit{CNLS-d}} \in \{0,\frac{\pi}{8}, \frac{\pi}{4}, \frac{3\pi}{8}, \frac{\pi}{2}\}\). The standard deviation of the normal distribution is \(\lambda = 0.1\). We perform the experiment \(100\) times for each parameter setting. 

Table \ref{tab:UnifRand} reports the radial MSE values from a testing set of $n$ observations lying on the true function.


\begin{table}[h!]
\centering
\caption{Experiment 1: Values of the radial MSE relative to the true function. The angle used in CNLS-d estimator varies and the noise direction is randomly selected. In the DGP, the standard deviation of the noise distribution, \(\lambda\), is 0.1.}
\begin{tabular}{cccccc}
\toprule
\multirow{2}[2]{*}{} & \multicolumn{5}{c}{CNLS-d Direction Angle} \\
    & \(0\)  &  \(\pi/8\) &   \(\pi/4\)    &  \(3\pi/8\)  &   \(\pi/2\)\\
\midrule
Average MSE across simulations & 13.90  & 4.65  & \textbf{3.32} & 4.49  & 13.93 \\
\midrule
\multicolumn{6}{l}{\textit{Note: Displayed are measured values multiplied by }\(10^4\).} \\
\end{tabular}%
\label{tab:UnifRand}%
\end{table}%

\noindent As shown in Table \ref{tab:UnifRand}, the angle corresponding to the smallest MSE (shown in bold) is the one that gives an orthogonal direction to the center of the true function, \(\frac{\pi}{4}\), and that the MSE values differ significantly, increasing at similar rates as the direction angle deviates from \(\frac{\pi}{4}\) in either direction.

\paragraph{Experiment 2 - The base case with fixed noise directions} \mbox{} 
\label{sec:exp2}

In this experiment, \(\theta_{\epsilon_i}\), which characterizes the noise direction for each observation, is constant for all observations, \(\theta_{\epsilon}\). The values used for \(\theta_{\epsilon}\) and the directions in CNLS-d estimator are the same, \(0,\frac{\pi}{8}, \frac{\pi}{4}, \frac{3\pi}{8}, \frac{\pi}{2}\). The standard deviation of the normal distribution is again  \(\lambda = 0.1\). We perform the experiment \(100\) times for each parameter settings. Table \ref{tab:Uniffixednoise} reports the results.


Each row in the Table \ref{tab:Uniffixednoise} corresponds to a different noise direction in DGP. The bold numbers identify the directions in CNLS-d estimator that obtain the smallest MSE for each noise direction. We confirm our previous insight, from the parametric estimator and fixed noise direction case described in \ref{sec:Ax.B.AdExp.1}, that the bold values appearing on the diagonal (from the upper-left to the lower-right of Table \ref{tab:Uniffixednoise}) correspond to the directions used in CNLS-d. This result indicates that selecting the direction in the SDDF that matches the underlying noise direction in the DGP results in improved functional estimates. 

\begin{table}[h!]
  \centering
  \caption{Experiment 2: Values of radial MSE relative to the true function varying the DGP noise direction and the CNLS-d estimator direction. In the DGP, the standard deviation of the noise distribution, \(\lambda\), is 0.1.}
    \begin{tabular}{crrrrr}
    \toprule
          & \multicolumn{5}{c}{CNLS-d Direction Angle} \\
    \cmidrule{2-6}    
    \multicolumn{1}{l}{Noise Direction Angle} &      \multicolumn{1}{c}{\(0\)}  &  \multicolumn{1}{c}{\( \pi / 8\)} &   \multicolumn{1}{c}{\( \pi / 4\)}    &  \multicolumn{1}{c}{\( 3\pi / 8\)}  &   \multicolumn{1}{c}{\(\pi / 2\)}  \\
    \midrule
    \(0\)           & \textbf{2.69} & 3.03  & 4.49  & 8.86  & 25.47 \\
    \( \pi / 8\)    & 7.49  & \textbf{3.44} & 4.00  & 8.07  & 28.83 \\
    \( \pi / 4\)    & 20.28 & 5.79  & \textbf{4.30} & 5.80  & 19.06 \\
    \( 3\pi / 8\)    & 25.58 & 7.80  & 4.18  & \textbf{3.51} & 6.84 \\
    \( \pi / 2\)    & 25.90 & 9.09  & 4.73  & 3.10  & \textbf{2.57} \\
    \midrule
    \multicolumn{6}{l}{\textit{Note: Displayed are measured values multiplied by }\(10^4\).} \\
    \end{tabular}%
  \label{tab:Uniffixednoise}
\end{table}%

\bigskip

\paragraph{Experiment 3. Base case with fixed noise direction and different noise levels} \mbox{} 
\label{sec:exp3} 

\noindent In Experiment 3, we vary the noise term by changing the \(\lambda\) coefficient
. Table \ref{tab:UnifFixedSmaller} reports the results for \(\lambda = 0.05\).

\begin{table}[h!]
  \centering
  \caption{Experiment 3--Less Noise: Values of radial MSE relative to the true function varying the DGP noise direction and the CNLS-d direction. In the DGP, the standard deviation of the noise distribution, \(\lambda\), is 0.05.}
    \begin{tabular}{crrrrr}
    \toprule
          & \multicolumn{5}{c}{CNLS-d Direction Angle} \\
    \cmidrule{2-6}    
    \multicolumn{1}{l}{Noise Direction Angle} &      \multicolumn{1}{c}{\(0\)}  &  \multicolumn{1}{c}{\( \pi / 8\)} &   \multicolumn{1}{c}{\( \pi / 4\)}    &  \multicolumn{1}{c}{\( 3\pi / 8\)}  &   \multicolumn{1}{c}{\(\pi / 2\)}  \\
    \midrule
    \(0\)         & 0.92  & \textbf{0.82} & 0.96  & 1.53  & 5.12 \\
    \( \pi / 8\)  & 1.83  & \textbf{1.09} & \textbf{1.09} & 1.47  & 5.45 \\
    \( \pi / 4\)  & 3.70  & 1.41  & \textbf{1.29} & 1.43  & 3.93 \\
    \( 3\pi / 8\) & 5.75  & 1.68  & 1.27  & \textbf{1.18} & 1.86 \\
    \( \pi / 2\)  & 4.61  & 1.40  & 0.95  & \textbf{0.79} & 0.90 \\
    \midrule
    \multicolumn{6}{l}{\textit{Note: Displayed are measured values multiplied by }\(10^4\).} \\
    \end{tabular}%
  \label{tab:UnifFixedSmaller}%
\end{table}%

In Table \ref{tab:UnifFixedSmaller} (Experiment 3, with \(\lambda = 0.05\)), we do not observe the same diagonal pattern observed in Experiment 2, and the best direction for CNLS-d estimator does not match the direction selected for the noise. 
This leads us to hypothesize that when the noise level is small, data characteristics, such as the distribution of the regressors or the shape of the function, affect the estimation whereas when the noise level is large, regressors' relative variability becomes a more dominant factor in determining the best direction for the CNSL-d estimator.   


However, with \(\lambda = 0.2\) the results of Experiment 3 are consistent with those from Experiment 2; i.e., the best direction always coincides with the noise direction selected. The results of Experiment 3 with \(\lambda = 0.2\) are reported in \ref{sec:Ax.B.AdExp}, Table \ref{tab:UnifFixedLarger} (Experiment 3 with \(\lambda = 0.2\)).



\paragraph{Experiment 4: Base case with different distributions for the initial observations on the true function} \mbox{} 
\label{sec:exp4}
In Experiment 4, we seek to understand how changing the DGP for the angle, \(\theta_i, \ i = 1,\ldots,n\), affects the optimal direction. We consider the three normal distributions with different parameters: \(N\left[\frac{\pi}{8},\frac{\pi}{16}\right]\), \(N\left[\frac{\pi}{4},\frac{\pi}{16}\right]\) and \(N\left[\frac{3\pi}{8},\frac{\pi}{16}\right]\). We truncate the tails of the distribution so that the generated angles fall in the range \(\left[0, \pi/2\right]\). Noise is specified as in  Experiment 1. Table \ref{tab:NormDist} reports the results of this experiment.

\begin{table}[h!]
  \centering
  \caption{Experiment 4: Values of radial MSE relative to the true function varying the CNLS-d direction and the mean of the normal distribution used in the DGP.}
    \begin{tabular}{crrrrr}
    \toprule
     \multicolumn{1}{c}{Mean of the}  &  \multicolumn{5}{c}{CNLS-d Direction angle} \\
    \cmidrule{2-6} 
    \multicolumn{1}{c}{Normal Distribution} (\(\bar{\theta}\)) & \multicolumn{1}{c}{\(0\)}  &  \multicolumn{1}{c}{\( \pi / 8\)} &   \multicolumn{1}{c}{\( \pi / 4\)}    &  \multicolumn{1}{c}{\( 3\pi / 8\)}  &   \multicolumn{1}{c}{\(\pi / 2\)}  \\
    \midrule
    \multicolumn{1}{c}{\(\pi/8\)} & 3.19  & \textbf{2.21} & 3.89  & 10.28 & 46.47 \\
    \multicolumn{1}{c}{\(\pi/4\)}  & 8.44  & 2.92  & \textbf{1.98} & 3.17  & 9.00 \\
    \multicolumn{1}{c}{\(3\pi/8\)} & 45.64 & 10.25 & 4.02  & \textbf{2.43} & 3.07 \\
    \midrule
    \multicolumn{6}{l}{\textit{Note: Displayed are measured values multiplied by }\(10^4\).} \\
    \end{tabular}%
  \label{tab:NormDist}%
\end{table}%

In Table {\ref{tab:NormDist}}, we observe that selecting a direction in the SDDF to match \(\bar{\theta}\), the mean of the distribution for the angle variable used in the DGP, corresponds to the smallest MSE value. This result suggests that the estimator's performance improves when we select a direction that points to the ``center'' of the data. 

\ref{sec:Ax.B.AdExp.2} presents additional experiments, varying the distribution of the observations and considering three outputs with a fixed costed level. These experiments lend further support to the strategy of selecting a direction pointed to the ``center'' of the data. 

\bigskip

\section{Proposed Approach to Direction Selection}
\label{sec:6.DirSelApp}


Based on Monte Carlo simulations, we found that the optimal direction depends on the shape of the function and the distribution of the observed data. This of itself is not surprising. However, by assuming a unimodal distribution for the data generation process, a direction that aims towards the ``center" of the data and is perpendicular to the true function at that point tends to outperform other directions. To apply this finding for a data set with \(Q\) outputs and \(n\) observations, \( (y_{i1},\ldots,y_{iQ},c_{i}), \ i=1,\ldots,n \), we suggest selecting the direction for the DDF as follows: 

\begin{mdframed}
\begin{enumerate}
    \item 
Normalize the data:
\begin{eqnarray}
    \breve{y}_{ij}  &=&
    \frac{ y_{ij} - \min_{k}{y_{kj}}}
    {\max_{k}{y_{kj}} - \min_{k}{y_{kj}}}
    , \ j=1,\ldots,Q, \ i,k=1,\ldots,n \\
    \breve{c}_{i}  &=&
    \frac{ y_{i} - \min_{k}{c_{k}}}
    {\max_{k}{c_{k}} - \min_{k}{c_{k}}} 
    , \ i,k=1,\ldots,n
\end{eqnarray}
\item Select the direction: 
\begin{align}  
    \begin{bmatrix}
       g^{y_1} \\           
       \vdots \\
       g^{y_Q} \\
       g^{c}
      \end{bmatrix} =
      \begin{bmatrix}
       \text{median}\left(\breve{y}_{i1}\right)\\
       \vdots \\
       \text{median}\left(\breve{y}_{iQ}\right) \\
       1-\text{median}\left(\breve{c}_{i}\right).
    \end{bmatrix}
\end{align}

\end{enumerate}
\end{mdframed}


This provides a method for direction selection that can be used in applications when the true direction is unknown.\footnote{A cost function is convex with respect to the point $[ \breve{y}_1, …,\breve{y}_Q,\breve{C}]=[ 0, …,0,1]$. Therefore, to have a ray that points from the point $[ 0, …,0,1]$ to the median of the data, the directional vector $[\text{median} (\breve{y}_{i1}),...,\text{median} (\breve{y}_{iQ}),1 - \text{median} (\breve{c}_{i})]$ is needed.} We test the proposed method by estimating a cost function for a US hospital data set. 

\section{Cost Function Estimation of the US Hospital Sector}
\label{sec:7.HospApp}

We analyze the cost variation across US hospitals using a conditional mean estimate of the cost function. We estimate a multi-output cost function for the US hospital sector by implementing our data-driven method for selecting the direction vector for the DDF. We report most productive scale size and marginal cost estimates. 

\subsection{Description of the Data Set}
\label{sec:7.1.DescDataset}

We obtain cost data from the American Hospital Association's (AHA) Annual Survey Databases from 2007 to 2009. The costs reported include payroll, employee benefits, depreciation, interest, supply expenses and other expenses. We estimate a cost function which can be interpreted as a distance function with a single input when hospitals face the same input prices\footnote{Unfortunately we do not observe input prices. We chose to estimate a cost function and make the assumption of common input prices rather than impose an arbitrary division of the cost.}. We obtain hospital output data from the Healthcare Cost and Utilization Project (HCUP) National Inpatient Sample (NIS) core file that captures data annually for all discharges for a 20\% sample of US community hospitals. The hospital sample changes every year. For each patient discharged, all procedures received are recorded as International Classification of Diseases, Ninth Revision, Clinical Modification (ICD9-CM) codes. The typical hospital in the US relies on these detailed codes to quantify the medical services it provides 
(\cite{zuckerman1994measuring}). We map the codes to four categories of procedures, specifically the procedure categories are ``Minor Diagnostic," ``Minor Therapeutic," ``Major Diagnostic," and ``Major Therapeutic" which are standard output categories in the literature (\cite{PopeJohnson2013}). The number of procedures is each category are summed for each hospital by year to construct the output variables. The total number of hospitals sampled is around 1,000 per year from 2007 to 2009.\footnote{The NIS survey is a stratified systematic random sample. The strata criteria are urban or rural location, teaching status, ownership, and bed size. This stratification ensures a more representative sample of discharges than a simple random sample would yield. For details see \url{https://www.hcup-us.ahrq.gov/tech_assist/sampledesign/508_compliance/508course.htm#{463754B8-A305-47E3-B7EE-A43953AA9478}}.} However, mapping between the two databases is only possible for approximately 50\% of the hospitals in the HCUP data, resulting in approximately 450 to 525 observations available each year.

\begin{table}[htbp]
  \centering
  \caption{Summary Statistics of the Hospital Data Set}
    \begin{tabular}{c|ccccc}
          & \multicolumn{5}{c}{\textbf{2007}}   \\
          & \multicolumn{5}{c}{(523 observations)}   \\
          & Cost (\$)  & MajDiag & MajTher & MinDiag & MinTher  \\
    Mean  & 146M & 162   & 4083  & 3499  & 7299   \\
    Skewness & 3.51  & 2.89  & 2.63  & 5.19  & 3.28   \\
    25-percentile & 24M & 9     & 277   & 108   & 512    \\
    50-percentile & 72M & 73    & 1688  & 938   & 3108   \\
    75-percentile & 182M & 207   & 5443  & 4082  & 9628   \\
    \midrule
        & \multicolumn{5}{c}{\textbf{2008}}   \\
        & \multicolumn{5}{c}{(511 observations)}   \\
        & Cost (\$)  & MajDiag & MajTher & MinDiag & MinTher  \\
     Mean    & 163M & 175   & 4433  & 3688  & 7657   \\
    Skewness   & 4.19  & 3.80  & 2.97  & 4.87  & 2.82   \\
    25-percentile  & 28M & 10    & 325   & 120   & 545    \\
    50-percentile   & 83M & 76    & 1809  & 1013  & 3350   \\
    75-percentile   & 189M & 246   & 5984  & 4569  & 10781  \\
    \midrule
        & \multicolumn{5}{c}{\textbf{2009}}   \\
        & \multicolumn{5}{c}{(458 observations)}   \\
        & Cost (\$)  & MajDiag & MajTher & MinDiag & MinTher  \\
     Mean &    175M & 161   & 4471  & 3615  & 7905\\
    Skewness & 3.39  & 3.78  & 2.43  & 4.68  & 2.41   \\
    25-percentile & 31M & 12    & 420   & 148   & 713    \\
    50-percentile & 91M & 69    & 1737  & 1136  & 3458   \\
    75-percentile & 220M & 230   & 6402  & 4694  & 10989  \\
    \bottomrule
    \end{tabular}%
  \label{tab:Hosp_SumStats}%
\end{table}%


\subsection{Pre-Analysis of the Data Set}
\label{sec:7.2.PresAnalysis}

\subsubsection{Testing the Relevance of the Regressors}
\label{sec:7.2.1.TestRegRel}


We begin by testing the statistical significance of our four output variables, \(\bm{y} = \left(y_1, y_2, y_3, y_4\right)\), for predicting cost. While the variables selected have been used in previous studies, we use these tests to evaluate whether this variable specification can be rejected for the current data set of U.S. hospitals from 2007-2009. 

The null hypothesis stated for the $q$th output is:
\[
    H_0 : P\left[
                    E\left(c \, | \, \bm{y}-\left\{y_q\right\}\right) = E\left(c \, | \, \bm{y}\right)
                \right] = 1
\]
against:\footnote{Where the notation $\bm{y}-\left\{y_q\right\}$ implies the vector $\bm{y}$ excluding the $q$th component.} 
\[
    H_1 : P\left[
                    E\left(c \, | \, \bm{y}-\left\{y_q\right\}\right) = E\left(c \, | \, \bm{y}\right)
                \right] < 1.
\]

We implement the test with a Local Constant Least Squares (LCLS) estimator described in \cite{HendersonParmeter2015}, calculating bandwidths using least-squares cross-validation. We use 399 wild bootstraps. 
We found that all output variables were highly statistically significant for all years.

\subsection{Results}
\label{sec:7.3.AppResults}

\paragraph{CNLS-d and Different Directions}
\label{sec:CompDirCNLSD}

We analyze each year of data as a separate cross-section because, as noted above, the HCUP does not track the same set of hospitals across years. 
To illuminate the direction's effect on the functional estimates, we graph ``Cost" as a function of ``Major Diagnostic Procedures" and ``Major Therapeutic Procedures" holding ``Minor Diagnostic Procedures'' and ``Minor Therapeutic Procedures'' constant at their median values. Figure \ref{fig:HospAppDiffDirs} illustrates the estimates for three different directions, one with only a cost component, one with only a component in Major Therapeutic Procedures, and one that comes from our median approach. Visual inspection indicates that the estimates with different directions produce significantly different estimates, highlighting the importance of considering the question of direction selection. 

\begin{figure}[h!]
	\begin{center}
	    \includegraphics[width=5in]{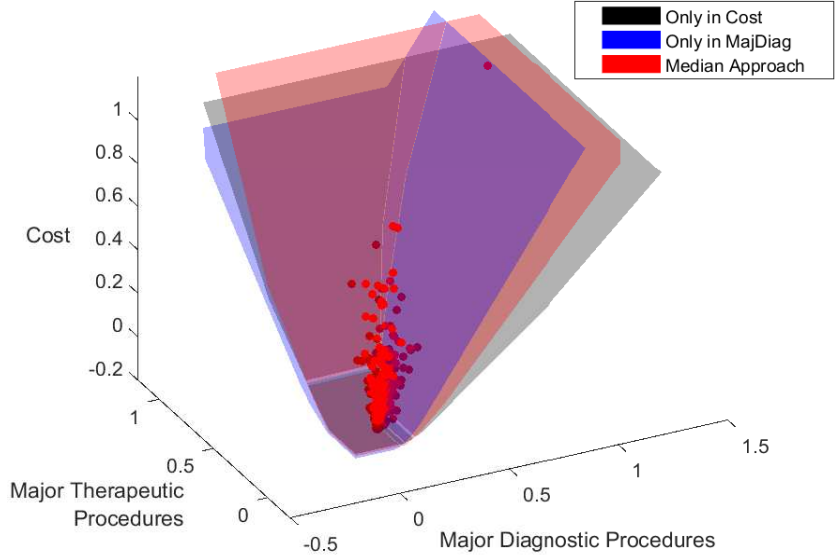}
	\end{center}
	\caption{US Hospital Cost Function Estimates for Three Directions}
	\label{fig:HospAppDiffDirs}
\end{figure}

We compare the estimator's performance when using different directions. Table \ref{tab:HospAv} reports the MSE for three sample directions in each year. We define our direction vector as \((g^{y1},g^{y2},g^{y3},g^{y4},g^c)\).\footnote{We focus on types of directions found to be competitive in our Monte Carlo simulations.}

\begin{table}[h!]
  \centering
  \caption{Results of the radial MSE values for different directions by year}
    \begin{tabular}{l|ccc}
    \multicolumn{1}{c|}{Direction} & \multicolumn{3}{c}{Year} \\
    \multicolumn{1}{c|}{(\(g^{y_1},g^{y_2},g^{y_3},g^{y_4},g^c)\)} & 2007  & 2008  & 2009 \\
    \midrule
     (0.45, 0.45, 0.45, 0.45, 0.45) & 2.10  & \textbf{1.30}  & 1.50 \\
    (0.35, 0.35, 0.35, 0.35, 0.71)  & 2.15  & 1.65  & \textbf{1.29} \\
    Median Direction  & \textbf{1.79}  & 1.55  & 1.34 \\
    \midrule
    \multicolumn{4}{l}{\textit{Note: Displayed are the measured values }}\\ 
    \multicolumn{4}{l}{\textit{multiplied by} \(10^3\)} \\
    \end{tabular}%
  \label{tab:HospAv}%
\end{table}%

We pick two directions, one with equal components in all dimensions, and a second direction that has a cost component that is double the value of the output components. The median vector is \((0.014,0.041,0.033,0.038,0.998)\), which is very close to the cost-only direction. The MSE varies by 15-30\% over the different directions. We observe that there is no clear dominant direction; however, the median direction performs reasonably well in all cases. We conclude that as long as a direction with non-zero components for all variables that could contain noise is selected, then the precise direction selected is not critical to obtaining improved estimation results.


\paragraph{Comparison with other estimators}
\label{sec:CompOtherEst_MSE}

We compare three methods to estimate a cost function: 1) a quadratic functional form (without the cross-product terms), \cite{fare2010functional}; 2) CNLS-d with the direction selection method proposed in Section \ref{sec:6.DirSelApp}; and 3) lower bound estimate calculated using a local linear kernel regression with a Gaussian kernel and leave one-out cross-validation for bandwidth selection, \cite{li2007nonparametric}.\footnote{For CNLS-d, we select a value for an upper bound through a tuning process, \(\text{Ubound} = 0.5\), and impose the upper bound on the slope coefficients estimated \citep{lim2014convergence}.} We select these estimators because a quadratic functional form to model production has been used in recent productivity and efficiency analysis of healthcare. See, for example, \cite{Ferrier2018}. The local linear kernel is selected because it is an extremely flexible nonparametric estimator and provides a lower bound for the performance of a functional estimate. However, note that the local linear kernel does not satisfy standard properties of a cost function; i.e., cost is monotonic in output and marginal costs are increasing as output increases. 

We will use the criteria of K-fold average MSE with $k=5$ to compare the approaches. This means we split the data equally into 5 parts. We use 4 of the 5 parts for estimation (training) and evaluate the performance of the estimator on the 5th part (testing). We do this for all 5 parts and average the results. 
The values presented in Table \ref{tab:HospDiffMethods} correspond to the average across folds.



\begin{table}[h!]
  \centering
  \caption{US Hospital K-fold Average MSE in Cost to the Cost Function Estimates for the Three Functional Specifications by Year}
    \begin{tabular}{c|cccc}
      & Quadratic  & CNLS-d &  Lower Bound\\
     Year  & Regression  &  (Median Direction) & Estimator \\
    \midrule
    2007  & 3.43    & 2.44  & \textbf{2.35} \\
    2008  & 2.76    & 1.93  & \textbf{1.48} \\
    2009  & 2.43    & 1.80  & \textbf{1.53} \\
    \midrule
    \multicolumn{5}{l}{\textit{Note: The MSE values displayed are the measured}} \\ \multicolumn{5}{l}{\textit{values multiplied by }\(10^3\)}\\
    \end{tabular}%
  \label{tab:HospDiffMethods}%
\end{table}%



While the average MSEs for all years are lowest for the lower bound estimator, CNLS-d performs relatively well as it is close to the lower bound in terms of fitting performance while imposing standard axioms of a cost function. As is true of most production data, the hospital data are very noisy. The shape restrictions imposed in CNLS-d improves the interpretability. The CNLS-d estimator outperforms the parametric approach, indicating the general benefits of nonparametric estimators. 


\paragraph{Description of Functional Estimates - MPSS and Marginal Costs}
\label{sec:CompOtherEst_MPSS+MC}

We report the most productive scale size (MPSS) and the marginal costs for the a quadratic parametric estimator, the CNLS-d estimator with our proposed direction selection method, and an alternative.\footnote{Here most productive scale size is measured on each ray from the origin (fixing the output ratios) and is defined as the cost level that maximizes the ratio of aggregate output to cost. Marginal cost is measured on each ray from the origin (fixing the output ratios) and is defined as the cost to increase aggregate output by one unit.} 
These metrics are determined on the averaged K-fold estimations for each estimation method. For the MPSS, we present the cost levels obtained for different ratios of \textit{Minor Therapeutic} procedures (MinTher) and \textit{Major Therapeutic} procedures (MajTher), with the minor and major diagnostics held constant at their median levels.

MPSS results are presented in Table \ref{tab:MPSS_c-level_ratios} and the values for CNLS-d (Median Direction) are illustrated in Figure \ref{fig:HospMPSS_CNLSdMed}. We observe small variations across both years and estimators. The differences across years are in part due to the sample changing across years. Most hospitals are small and operate close to the MPSS. 
However, there are several large hospitals that are operating significantly above MPSS. Hospitals might choose to operate at larger scales and provide a large array of services allowing consumers to fulfill multiple healthcare needs. 

For marginal costs, we present the values for different percentiles of the MinTher and MajTher, with the minor and major diagnostics held constant at their median levels. A more exhaustive comparison across all outputs is presented in \ref{sec:Ax.C.HospAppAnnex}. Marginal cost information can be used by hospital decision makers to select the types of improvements that are likely to result in higher productivity with minimal cost increase. For example, consider a hospital that is in the $50^{th}$ percentile of the data set for all four outputs in 2008 and the hospital manager has the option to expand operations for either minor or major diagnostic procedures. Results reported in Tables \ref{tab:MC_MajDiag_9rows} and \ref{tab:MC_MajTher_9rows} indicate that an increase of 1 minor therapeutic procedures would result in a $\$ 4.9k$ increase in cost. Alternatively, an increase of 1 major therapeutic procedures would result in a $\$ 7.7k$ increase in cost. 
A decision maker would want to consider the revenue generated by the different procedures; however, these estimates provide insights regarding the incremental cost of additional major and minor therapeutic procedures. 

CNLS-d is the most flexible of the estimators and allows MPSS values to fluctuate significantly across percentiles. CNLS-d does not smooth variation, rather it minimizes the distance from each observation to the shape constrained estimator. In \ref{sec:Ax.C.HospAppAnnex}, results for the local linear kernel estimator are also presented. Even though the local linear kernel bandwidths are selected via cross-validation, relatively large values are selected due to the relatively noisy data and the highly skewed distribution of output. These large bandwidths and the parametric nature of the quadratic function make these two estimators relatively less flexible compared to CNLS-d. A feature of performance that is captured only by CNLS-d is that, hospitals specializing in either minor or major therapeutics maximize productivity at a larger scales of operation as illustrated in Figure \ref{fig:HospMPSS_CNLSdMed}.


\begin{table}[h!]
  \centering
  \caption{Most Productive Scale Size measured in cost ($\$M$) conditional on Minor Therapeutic procedures (MinTher) and Major Therapeutic procedures (MajTher), Minor Diagnostic procedures (MinDiag) and Major Diagnostic procedures (MajDiag) held constant at the 50th percentile}
    \begin{tabular}{c|ccc|ccc|ccc|}
    \toprule
    Ratio & \multicolumn{3}{c|}{Quadratic Regression} & \multicolumn{3}{c|}{CNLS-d (median)} & \multicolumn{3}{c|}{CNLS-d (equal)} \\
\cmidrule{2-10}    MajTher/MinTher & 2007  & 2008  & 2009  & 2007  & 2008  & 2009  & 2007  & 2008  & 2009 \\
    \midrule
    20\%  & 13    & 379   & 252   & 210   & 61    & 88    & 224   & 137   & 106 \\
    30\%  & 17    & 861   & 640   & 146   & 66    & 83    & 134   & 129   & 148 \\
    40\%  & 272   & 377   & 1090  & 107   & 56    & 77    & 127   & 85    & 135 \\
    50\%  & 870   & 249   & 1552  & 112   & 64    & 85    & 124   & 126   & 134 \\
    60\%  & 360   & 210   & 276   & 90    & 70    & 120   & 88    & 96    & 142 \\
    70\%  & 205   & 182   & 187   & 111   & 66    & 184   & 132   & 104   & 104 \\
    80\%  & 151   & 170   & 150   & 174   & 69    & 286   & 221   & 110   & 111 \\
    \midrule
    \multicolumn{10}{l}{\textit{Note: The values displayed are in \$M}}\\
    \end{tabular}%
  \label{tab:MPSS_c-level_ratios}%
\end{table}%

\begin{figure}[h!]
	\begin{center}
	    \includegraphics[width=7in]{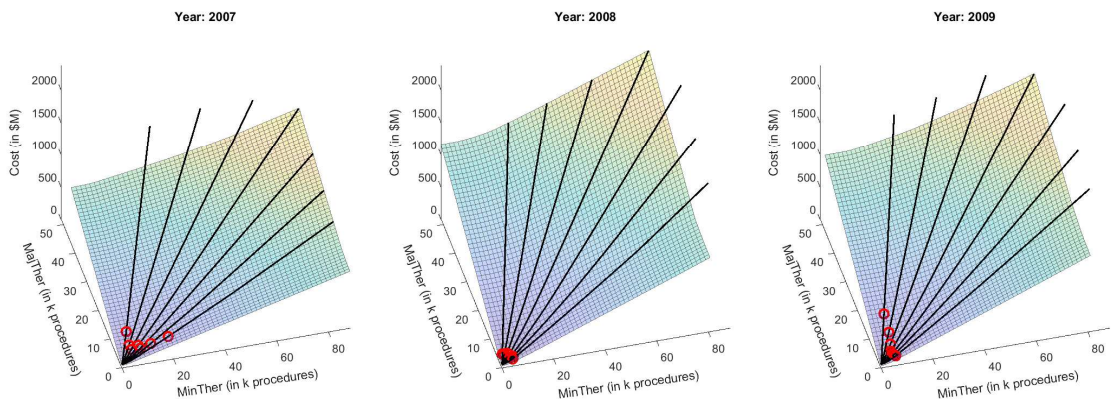}
	\end{center}
	\caption{Most Productive Scale Size (in red) on the estimated function by ``CNLS-d (Med)", CNLS-d using the median approach for the direction, for different ratios of Major Therapeutic Procedures over Minor Therapeutic Procedures}
	\label{fig:HospMPSS_CNLSdMed}
\end{figure}

\begin{table}[h!]
  \centering
  \caption{Marginal Cost of Minor Therapeutic Procedures}
    \begin{tabular}{cc|ccc|ccc|ccc}
    \toprule
    \multicolumn{2}{c|}{Percentile} & \multicolumn{3}{c|}{Quadratic Regression} & \multicolumn{3}{c|}{CNLS-d (median)} & \multicolumn{3}{c}{CNLS-d (equal)} \\
    \midrule
     MinTher & MajTher & 2007  & 2008  & 2009  & 2007  & 2008  & 2009  & 2007  & 2008  & 2009 \\
    \midrule
    25    & 25    & 8.9   & 6.5   & 13.2  & 0.03  & 0.03  & 0.03  & 0.2   & 0.02  & 0.1 \\
    25    & 50    & 8.9   & 6.5   & 13.2  & 0.05  & 0.1   & 0.1   & 0.04  & 0.1   & 0.04 \\
    25    & 75    & 8.9   & 6.5   & 13.2  & 0.2   & 0.04  & 0.03  & 0.1   & 0.02  & 0.02 \\
    50    & 25    & 8.1   & 6.1   & 12.4  & 6.9   & 5.5   & 7.4   & 5.9   & 6.3   & 7.8 \\
    50    & 50    & 8.1   & 6.1   & 12.4  & 4.3   & 4.9   & 7.8   & 2.1   & 3.7   & 7.4 \\
    50    & 75    & 8.1   & 6.1   & 12.4  & 0.2   & 0.4   & 0.03  & 0.1   & 0.02  & 0.02 \\
    75    & 25    & 6.0   & 5.0   & 10.4  & 9.6   & 13.5  & 14.0  & 9.5   & 10.9  & 14.1 \\
    75    & 50    & 6.0   & 5.0   & 10.4  & 9.6   & 13.5  & 14.3  & 9.6   & 10.9  & 13.8 \\
    75    & 75    & 6.0   & 5.0   & 10.4  & 5.7   & 10.1  & 6.4   & 4.6   & 8.7   & 6.4 \\
    \midrule
    \multicolumn{11}{l}{\textit{Note: The values displayed are in \$k }}\\
    \end{tabular}%
  \label{tab:MC_MajDiag_9rows}%
\end{table}%

\begin{table}[h!]
  \centering
  \caption{Marginal Cost of Major Therapeutic Procedures}
    \begin{tabular}{cc|ccc|ccc|ccc}
    \toprule
    \multicolumn{2}{c|}{Percentile} & \multicolumn{3}{c|}{Quadratic Regression} & \multicolumn{3}{c|}{CNLS-d (median)} & \multicolumn{3}{c}{CNLS-d (equal)} \\
    \midrule
     MinTher & MajTher & 2007  & 2008  & 2009  & 2007  & 2008  & 2009  & 2007  & 2008  & 2009 \\
    \midrule
    25    & 25    & 10.5  & 11.5  & 9.8   & 0.1   & 0.04  & 0.1   & 0.2   & 0.03  & 0.1 \\
    25    & 50    & 11.7  & 13.0  & 10.8  & 11.3  & 11.8  & 15.7  & 10.5  & 10.3  & 14.6 \\
    25    & 75    & 15.1  & 17.2  & 14.5  & 19.8  & 22.1  & 24.6  & 19.8  & 21.8  & 24.0 \\
    50    & 25    & 10.5  & 11.5  & 9.8   & 0.4   & 0.2   & 0.5   & 0.1   & 0.1   & 0.4 \\
    50    & 50    & 11.7  & 13.0  & 10.8  & 3.7   & 7.7   & 1.7   & 6.9   & 7.1   & 3.7 \\
    50    & 75    & 15.1  & 17.2  & 14.5  & 19.8  & 22.0  & 24.6  & 19.8  & 21.8  & 24.0 \\
    75    & 25    & 10.5  & 11.5  & 9.8   & 0.2   & 0.03  & 0.1   & 0.0   & 0.1   & 0.1 \\
    75    & 50    & 11.7  & 13.0  & 10.8  & 0.2   & 0.2   & 0.4   & 0.8   & 0.1   & 0.3 \\
    75    & 75    & 15.1  & 17.2  & 14.5  & 18.3  & 12.4  & 19.8  & 16.2  & 11.0  & 15.2 \\
    \midrule
    \multicolumn{11}{l}{\textit{Note: The values displayed are in \$k}}\\
    \end{tabular}%
  \label{tab:MC_MajTher_9rows}%
\end{table}%

\begin{figure}[h!]
\centerline{%
\includegraphics[width=0.5\textwidth]{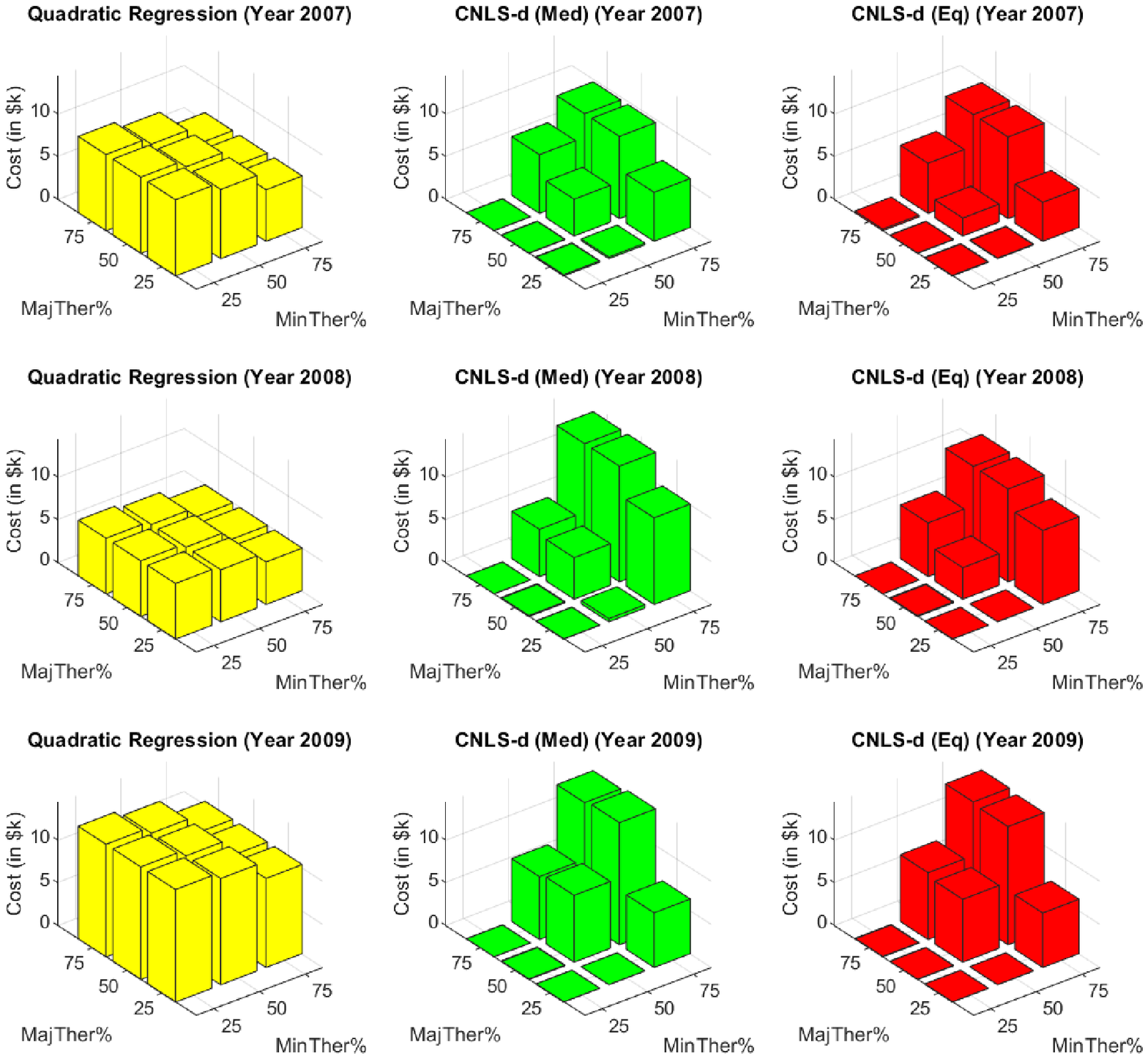}%
\includegraphics[width=0.5\textwidth]{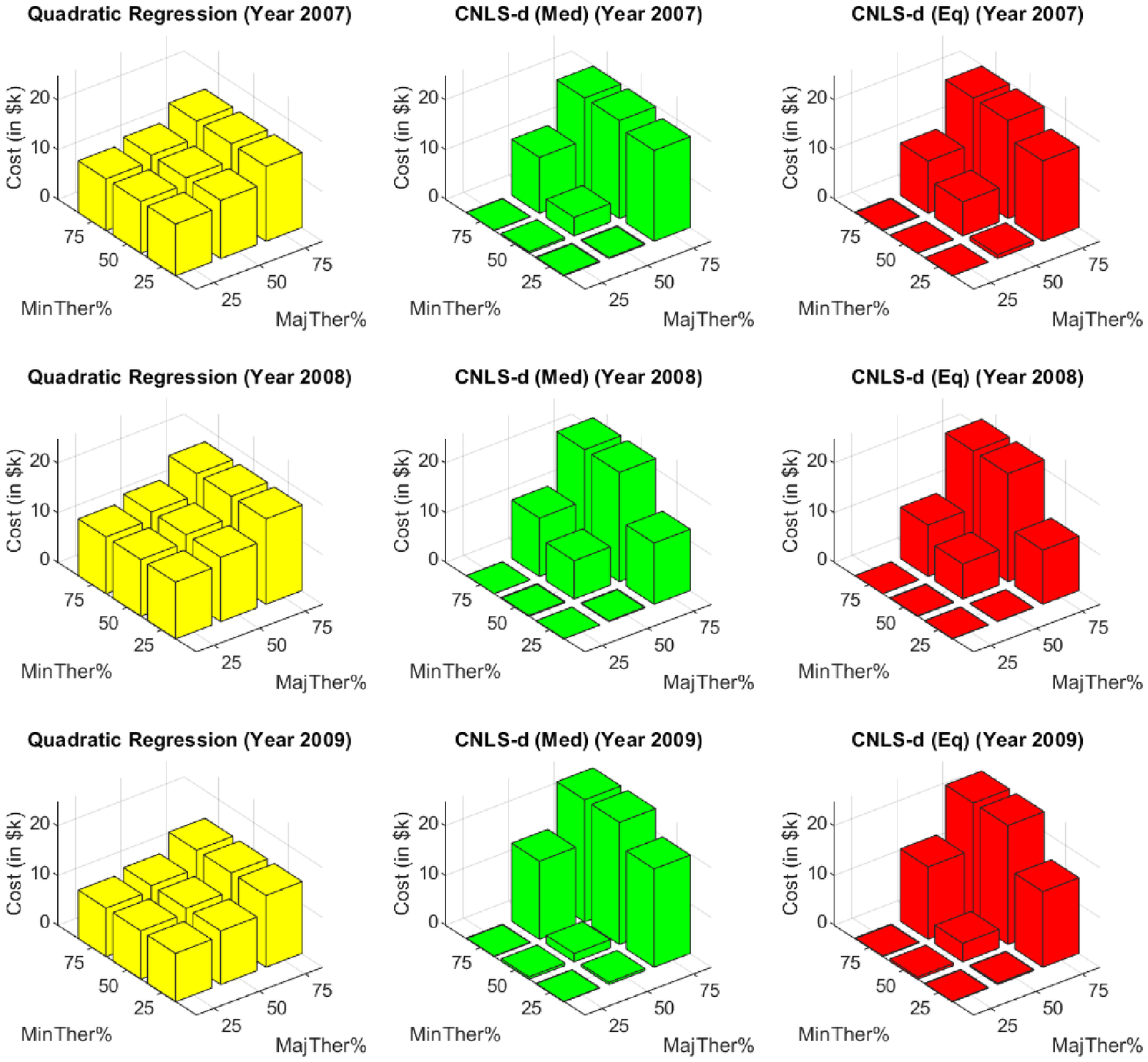}%
}%
\caption{Marginal Cost of the Minor Therapeutic procedures (left) and Marginal Cost of the Major Therapeutic procedures (right) (``CNLS-d (Med)" corresponds to CNLS-d using the median approach for the direction and ``CNLS-d (Eq)'' corresponds to CNLS-d using the direction  with equal components in all netputs}
\label{fig:MCMajDiag+MCMajTher}
\end{figure}



The marginal cost results for Minor Therapeutic procedures are presented in Table \ref{tab:MC_MajDiag_9rows} and Figure \ref{fig:MCMajDiag+MCMajTher} (left) and the marginal cost results for Major Therapeutic procedures are reported in Table \ref{tab:MC_MajTher_9rows} and Figure \ref{fig:MCMajDiag+MCMajTher} (right). 
As was the case for MPSS (see Table \ref{tab:MPSS_c-level_ratios}), CNLS-d is more flexible and its marginal cost estimates vary significantly across percentiles. The CNLS-d with different directions provides very similar marginal costs estimates. However, the CNLS-d estimates differ significantly from the marginal cost estimates obtained with the parametric estimator. 
For CNLS-d the marginal costs results are in line with the theory that marginal costs are increasing with scale. This property can also be violated if using a non-parametric estimator without any shape constraints imposed. For example this can be seen in the marginal costs of minor therapeutic procedures for the parametric (quadratic) regression estimator, Figure \ref{fig:MCMajDiag+MCMajTher}. 


Our data set, which combines AHA cost data with AHRQ output data for a broad sample of hospitals from across the US, is unique to the best of our knowledge. However, the marginal cost estimates are broadly in line with marginal cost estimates for US hospitals for similar time periods. \cite{gowrisankaran2015mergers} studied a considerably smaller set of Northern Virginia hospitals observed in 2006 that, on average, were larger that hospitals in our data set. Due to the differences in the measures of output the marginal cost levels are not directly comparable. However, conditional on the size variation, the variation in marginal costs is similar to the variation we observe for the parametric (quadratic) regression specification applied to our data. \cite{boussemart2015two} analyzed data on nearly 150 hospitals located in Florida observed in 2005. The authors use a different output specification and a translog model; however, their distribution of hospital size is similar to our data set and we observe similar variances in marginal costs with the parametric (quadratic) regression specification applied to our data.

\section{Conclusions}
\label{sec:8.Conclusion}

This paper investigated the improvement in functional estimates when specifying a particular direction in CNLS-d. 
Based on Monte Carlo experiments, two primary findings emerged from our analysis. First, directions close to the average orthogonal direction to the true function performed well. Second, when the data are noisy, selecting a direction that matched the noise direction of the DGP improves estimator performance. Our simulations indicate that CNLS-d with a direction orthogonal to the data is preferable if the noise level is not too large and that a direction that matches the noise direction of the DGP is preferred if the noise level is large. 
Thus, if users know the shape of the data or the characteristics of the noise, they can use CNLS-d with a direction orthogonal to the data if the noise coefficient is small. Or if the noise coefficient is large, the user can select a direction close to the true noise direction, with non-zero components in all variables that potentially have noise. Our application to US hospital data shows that CNLS-d performs similarly across different directions that all include non-zero components of the direction vector for variables that potentially have noise in their measurement. 

In future research, we propose developing an alternative estimator that incorporates multiple directions in CNLS-d while maintaining the concavity axiom. This would permit treating subgroups within the data, allowing different assumptions to be made across subgroups (e.g., for-profit vs. not-for-profit hospitals). 






\bibliography{References}

\pagebreak

\bigskip
\appendix

\appendixpage
\addappheadtotoc

The appendix is composed of the following parts:
\begin{itemize}
    \item Properties of Directional Distance Functions and CNLS-d (\ref{sec:Ax.A.ApprDetails})
    
    \item Monte Carlo, Additional Experiments (\ref{sec:Ax.B.AdExp})
    
    \item Detailed Results for the Hospital Application (\ref{sec:Ax.C.HospAppAnnex})
    
\end{itemize}

\pagebreak

\section{Properties of Directional Distance Functions and CNLS-d}
\label{sec:Ax.A.ApprDetails}


\subsection{Direction Selection in Directional Distance Functions}
\label{sec:Ax.A.1.ImpofDisc}

In this appendix we prove that the direction vector affects the functional estimates. Let $\bm{g}^{x,y}=(\bm{g}^x,\bm{g}^y)$, then we can state the following theorem:

\begin{restatable}{thm}{effectdir}
    \label{thm:effectdir}
		Suppose that two direction vectors exist, $\bm{g}^{x,y}_a$ and $\bm{g}^{x,y}_b$, such that $\bm{g}^{x,y}_a \neq \bm{g}^{x,y}_b$. Then the directional distance function estimates using these two different directions are not equal, $D(\bm{X},\bm{Y};\bm{g}^{x,y}_a) \neq D(\bm{X},\bm{Y};\bm{g}^{x,y}_b)$.
\end{restatable}

\begin{proof}

Rewrite Problem (\ref{eq:CNLSd}) from Section 3.2 as
\begin{subequations}
\label{eq:CNLSd_simp}
\begin{align}
\min_{\bm{\alpha}, \bm{\beta}, \bm{\gamma}} \sum_{i=1}^{n}{(\alpha_i + \bm{\beta_i}' \, \bm{x_i} - \gamma_i}'\,\bm{y_i})^2  \tag{\ref{eq:CNLSd_simp}}   \\
\text{s.t.}\ \ \
     \alpha_i + \bm{\beta_i}' \, \bm{x_i}  - \bm{\gamma_i}'\,\bm{y_i}  \leq \alpha_j + \bm{\beta_j}' \, \bm{x_i}  - \bm{\gamma_j}'\,\bm{y_i}, &  \ \text{for} \ i,j=1,\ldots,n, \  i \neq j \\
\bm{\beta_i},\bm{\gamma_i}  \geq 0, & \  \text{for} \ i=1,\ldots,n \\
\bm{\beta_i}' \, \bm{g}^x + \bm{\gamma_i}'\,\bm{g}^y  = 1, & \ \text{for} \ i=1,\ldots,n
\end{align}
\end{subequations}



\noindent Observe that all decision variables appear in the objective function and that the objective function is a quadratic function while the constraints define a convex solution space; i.e., this optimization problem has a unique solution (\cite{bertsekas1999nonlinear}). If we solve Problem \eqref{eq:CNLSd_simp} with $\bm{g}^{x,y}_a$, then the resulting solution vector is $(\bm{\alpha}_a, \bm{\beta}_a, \bm{\gamma}_a)$. Changing the direction vector from $\bm{g}^{x,y}_a$ to $\bm{g}^{x,y}_b$ the normalization constraint $\bm{\beta_i}' \, \bm{g}^x_b + \bm{\gamma_i}'\,\bm{g}^y_b  = 1$ no longer holds for $\bm{\beta}_a$ and  $\bm{\gamma}_a$. However, the previous argument holds for the uniqueness of $(\bm{\alpha}_b, \bm{\beta}_b, \bm{\gamma}_b)$. Thus, $(\bm{\alpha}_a, \bm{\beta}_a, \bm{\gamma}_a) \neq (\bm{\alpha}_b, \bm{\beta}_b, \bm{\gamma}_b)$.

\end{proof}

\subsection{Details of CNLS-d}
\label{sec:6.1.CNLSd_details}

An alternative expression for CNLS-d (cf. equations (16)-(16c) from Section 5.1) is given by:

\begin{subequations}
\label{eq:CNLSd_details}
\begin{align}
\min_{\bm{\alpha}, \bm{\beta}, \bm{\gamma}} \sum_{i=1}^{n}{\epsilon_i^2} \tag{\ref{eq:CNLSd_details}}   \\
\text{s.t.}\ \ \
     -\epsilon_j + \epsilon_i + \bm{\beta_i}' \, \left(\bm{x_i} - \bm{x_j} \right)  - \bm{\gamma_i}'\,\left(\bm{y_i} - \bm{y_j} \right) \leq 0, \ &
    \text{for} \ i,j = 1,\ldots,n, \ i \neq j \\
    \bm{\beta_i}' \, \bm{g}^x + \bm{\gamma_i}'\,\bm{g}^y  = 1, \ & 
    \text{for} \ i = 1,\ldots,n\\
    \bm{\beta_i}, \,\bm{\gamma_i} \geq 0, \  & \text{for} \ i = 1,\ldots,n.
\end{align}
\end{subequations}


It's possible to recover \(\alpha_i, i = 1,\ldots,n\), and the final estimates using the following relations:
\begin{eqnarray}
\label{eq:Vars}
    \bm{\hat{x}_i} &=& \bm{x_i} + \epsilon_i \, \bm{g}^x, \ 
    \text{for} \ i = 1,\ldots,n \\
    \bm{\hat{y}_i} &=& \bm{y_i} - \epsilon_i \, \bm{g}^y, \ 
    \text{for} \ i = 1,\ldots,n \\
    \alpha_i &=& -\bm{\beta_i}' \, \bm{x_i} + \bm{\gamma_i}' \, \bm{y_i} + \epsilon_i, \
    \text{for} \ i = 1,\ldots,n.
\end{eqnarray}
\\

\subsection{Different Directions for Different Groups in CNLS-d}
\label{sec:Ax.A.2.ImpofDisc}


Consider the case where all observations have the same input level and produce two outputs and estimate the isoquant. Define two groups of observations $G_1$ and $G_2$ such that $\lvert G_1 \cup G_2 \rvert= n$  and $G_1 \cap G_2 = \emptyset$.\footnote{The notation $\lvert \cdot \rvert$ corresponds to the cardinality of the set.} Using the notation in \ref{sec:Ax.A.1.ImpofDisc}, the direction vector for the first group of observations $G_1$ is \(\bm{g}^{y_{G_1}}\) and it's \(\bm{g}^{y_{G_2}}\) for the second group of observations $G_2$.

For either a fixed input vector, $\bm{X}$, or a fixed cost level, $c$, formulate the iso-cost estimator for $G_1$ and $G_2$ with different directions vectors as:

\begin{subequations} 
\label{eq:CNLSd_multi}
    \begin{align}
       & & \min\limits_{\bm{\alpha}, \bm{\beta}, \bm{\gamma}, \bm{\epsilon}}{\sum_{i=1}^{n}{\epsilon_i^2}} \tag{\ref{eq:CNLSd_multi}}\\
        && \text{s.t.} \ \ \
        -\epsilon_j + \epsilon_i - \bm{\gamma_i}'\,\left(\bm{y_i} - \bm{y_j} \right) \leq 0, 
        &\ \text{for} \ i,j = 1,\ldots,n, \ i\neq j \label{eq:CNLSd_multi_a}\\
        && \bm{\gamma_i}'\,\bm{g^{y_{G_1}}}  = 1, 
        &\ \text{for} \; i \in G_1 \label{eq:CNLSd_multi_b}\\
        && \bm{\gamma_i}'\,\bm{g^{y_{G_2}}}  = 1, 
        &\ \text{for} \; i \in G_2 \label{eq:CNLSd_multi_c}\\
        && \bm{\gamma_i} \geq 0, 
        &\ \text{for} \ i = 1,\ldots,n.  \label{eq:CNLSd_multi_d}
    \end{align}
\end{subequations}

Note that using more than one direction for CNLS-d can lead to violations on convexity. Only under very limiting conditions can we allow for multiple directions in CNLS-d and guarantee that the resulting estimated function will maintain convexity. The following theorem formalizes the conditions. 

\begin{restatable}{thm}{ConvCond}
    \label{thm:ConvCond}
    If a CNLS-d estimator is calculated using two groups of observations with different direction vectors as shown in Equation \eqref{eq:CNLSd_multi} and the following condition holds regarding the direction vectors and the noise direction: 
    \begin{equation}
    \label{eq:2dconditions}
        \left(\epsilon_i \, \frac{\bm{g^{y_{k(i)}}}}{\|\bm{g^{y_{k(i)}}}\|}\right)' \ 
        \left[ \frac{\bm{g^{y_{k(j)}}}}{\|\bm{g^{y_{k(j)}}}\|} 
        - \frac{\bm{g^{y_{k(i)}}}}{\|\bm{g^{y_{k(i)}}}\|} \right] \geq 0, \ 
        \text{for} \ i,j = 1,\ldots,n, \ i \neq j,
    \end{equation}
    where
    \begin{equation*}
    k(i) = \begin{cases}
        1, & \text{if $i \in G_1$}\\
        2, & \text{if $i \in G_2$},
    \end{cases}
    \end{equation*}
    then the resulting CNLS-d estimate is a concave function.
\end{restatable}

\begin{proof}

Consider the Afriat inequalities in the context of cost isoquant estimation.  One of the conditions of Equation \eqref{eq:CNLSd_iso} is: 
\begin{equation}
\label{eq:AF_proof}
    \epsilon_i - \epsilon_j - \bm{\gamma_i}'\,\left(\bm{y_i} - \bm{y_j} \right) \leq 0, \ \text{for} \ i,j = 1,\ldots,n, \ i \neq j.
\end{equation}
Knowing that \(
\epsilon_i \, \frac{\bm{g^{y_{k(i)}}}}{\|\bm{g^{y_{k(i)}}}\|} = \bm{\hat{y}_i} - \bm{y_i}
\) means that
\(
\epsilon_i = \left(\bm{\hat{y}_i} - \bm{y_i}\right)' \, \frac{\bm{g^{y_{k(i)}}}}{\|\bm{g^{y_{k(i)}}}\|}
\).

Substituting \(\epsilon_i\) and \(\epsilon_j\) in the inequalities  \eqref{eq:AF_proof} obtains:
\begin{equation}
\label{eq:AF_proof2}
\left(\bm{\hat{y}_i} - \bm{y_i}\right)' \, \frac{\bm{g^{y_{k(i)}}}}{\|\bm{g^{y_{k(i)}}}\|} - \left(\bm{\hat{y}_j} - \bm{y_j}\right)' \, \frac{\bm{g^{y_{k(j)}}}}{\|\bm{g^{y_{k(j)}}}\|} - \bm{\gamma_i}' \, \left(\bm{y_i} - \bm{y_j}\right) \leq 0, \ \text{for} \ i,j = 1,\ldots,n, \ i \neq j. 
\end{equation}

Next, consider the case where both observations have the same direction. Then the expression is:
\begin{equation}
\label{eq:AF_proof3}
    \left[\left(\bm{\hat{y}_i} - \bm{y_i}\right) - \left(\bm{\hat{y}_j} - \bm{y_j}\right) \right]' 
    \, \frac{\bm{g^{y_{k(i)}}}}{\|\bm{g^{y_{k(i)}}}\|}
    - \bm{\gamma_i}' \, \left(\bm{y_i} - \bm{y_j}\right)
    \leq 0, 
    \ \text{for} \ i,j = 1,\ldots,n, \ i \neq j. 
\end{equation}

If Equation \eqref{eq:AF_proof3} is satisfied, we know that the CNLS-d constraints hold. By comparison observe that the condition listed below is a sufficient condition for Equation \eqref{eq:AF_proof3} being satisfied when Equation \eqref{eq:AF_proof2} holds:

\[
\left[\left(\bm{\hat{y}_i} - \bm{y_i}\right) - \left(\bm{\hat{y}_j} - \bm{y_j}\right) \right]' \,
    \frac{\bm{g^{y_{k(i)}}}}{\|\bm{g^{y_{k(i)}}}\|} 
    - \bm{\gamma_i}' \, \left(\bm{y_i} - \bm{y_j}\right)  \ \ -- \textit{from eq.} \eqref{eq:AF_proof3} \ \ \ \ \ \ 
\] 
\[ \ \ \ \ \   \leq \left(\bm{\hat{y}_i} - \bm{y_i}\right)' 
    \, \frac{\bm{g^{y_{k(i)}}}}{\|\bm{g^{y_{k(i)}}}\|} 
    - \left(\bm{\hat{y}_j} - \bm{y_j}\right)' \, 
    \frac{\bm{g^{y_{k(j)}}}}{\|\bm{g^{y_{k(j)}}}\|} - \bm{\gamma_i}' \, \left(\bm{y_i} - \bm{y_j}\right) \ \ -- \textit{from eq.} \eqref{eq:AF_proof2}
\] 
\[\ \text{for} \ i,j = 1,\ldots,n, \ i \neq j,\]

\noindent which, after simplifying, becomes:
\begin{equation}
\label{eq:AF_proof4}
    \left(\bm{\hat{y}_i} - \bm{y_i}\right)' \,
    \left[ \frac{\bm{g^{y_{k(j)}}}}{\|\bm{g^{y_{k(j)}}}\|} 
    - \frac{\bm{g^{y_{k(i)}}}}{\|\bm{g^{y_{k(i)}}}\|} \right] 
    \geq 0 \
    \text{for} \ i,j = 1,\ldots,n, \ i \neq j
\end{equation}
\end{proof}

Thus Theorem \ref{thm:ConvCond} is proved and a sufficient condition is found that, if verified, ensures the concavity property of the estimator even when multiple directions are used in the estimation of the directional distance function.

The following corollary, concerning the convex case, is directly inferred from Theorem \ref{thm:ConvCond}:

\begin{restatable}{corollary}{ConvCond2}
    \label{cor:ConvCond2}
    If a CNLS-d estimator is calculated using two groups of observations with different direction vectors as shown in Equation \eqref{eq:CNLSd_multi}, and the following condition holds regarding the direction vectors and the noise direction: 
    \begin{equation}
    \label{eq:2dconditions_conv}
        \left(\epsilon_i \, \frac{\bm{g^{y_{k(i)}}}}{\|\bm{g^{y_{k(i)}}}\|}\right)' \ 
        \left[ \frac{\bm{g^{y_{k(j)}}}}{\|\bm{g^{y_{k(j)}}}\|} 
        - \frac{\bm{g^{y_{k(i)}}}}{\|\bm{g^{y_{k(i)}}}\|} \right] 
        \leq 0, \ 
        \text{for} \ i = 1,\ldots,n, \ i \neq j,
    \end{equation}
    where
    \begin{equation*}
    k(i) = \begin{cases}
        1, & \text{if $i \in G_1$}\\
        2, & \text{if $i \in G_2$},
    \end{cases}
    \end{equation*}
    then the resulting CNLS-d estimate is a convex function.
\end{restatable}

\begin{proof}

Reverse the inequality sign in Equation \eqref{eq:AF_proof}:
\begin{equation}
\label{eq:AF_proof_cor}
-\epsilon_j + \epsilon_i - \bm{\gamma_i}'\,\left(\bm{y_i} - \bm{y_j} \right) \geq 0, \ \text{for} \ i,j = 1,\ldots,n, \ i \neq j,
\end{equation}

\noindent and follow the logic of the proof of Theorem \ref{thm:ConvCond} to obtain Corollary \ref{cor:ConvCond2} and Equation \eqref{eq:2dconditions_conv}.

\end{proof}

Theorem \ref{thm:ConvCond} clarifies that if the directions for each respective group are orthogonal to each other, then condition \ref{eq:2dconditions} is verified. This means that if the direction for group 1 has a single nonzero component in the output 1 dimension and group 2 has a single nonzero component in the output 2 dimension, then we will not observe violations of the convexity property.


We state a second Corollary that follows from Theorem \ref{thm:ConvCond}, which is useful when there are more than two groups each with their own estimation direction in CNLS-d. 

\begin{restatable}{corollary}{MultiDirs}
    \label{cor:MultiDirs}
    Let $n \in  \mathbb{N}$ the total number of observation.
    Let $Q$ the number of outputs considered.
    Let $\bm{Y} = \{\bm{y}_i \in \mathbb{R}^Q_+, i=1,\ldots,n\}$ the set of observed outputs.
    Let $P_g$ a partition of $\bm{Y}$ of cardinal $N_g \in \mathbb{N}$. 
    Let $\bm{g}^{y} = \{\bm{g}^{y_k}, k=1,\ldots,N_g\}$ the set of directions used for each respective group of the partition.
    If a CNLS-d estimator is calculated using the directions from $\bm{g}^{y}$ based on partition $P_g$, and the following condition holds regarding the direction vectors and the noise direction: 
    \begin{equation}
    \label{eq:2dconditions_v2}
        \left(\epsilon_i \, \frac{\bm{g^{y_{k(i)}}}}{\|\bm{g^{y_{k(i)}}}\|}\right)' \ 
        \left[ \frac{\bm{g^{y_{k(j)}}}}{\|\bm{g^{y_{k(j)}}}\|} 
        - \frac{\bm{g^{y_{k(i)}}}}{\|\bm{g^{y_{k(i)}}}\|} \right] \geq 0, \ 
        \text{for} \ i,j = 1,\ldots,n, \ i \neq j,
    \end{equation}
    where $\text{for each} \ i=1,\ldots,n, \ k(i)$ corresponds to the indicator of the part of the partition $P_g$, in which $\bm{y}_i$ belongs.
    Then the resulting CNLS-d estimate is a concave function.
\end{restatable}

\begin{proof}

We can follow the proof of Theorem \ref{thm:ConvCond}, as the condition does not change. The condition still concerns pairwise observations, the only difference is that now the partition of observations corresponds to more than two groups. This does not affect the proof of the condition.

\end{proof}

Corollary \ref{cor:MultiDirs} extends the statement of Theorem \ref{thm:ConvCond} to provide sufficient conditions to avoid violations of the shape constraints in a scenario where there are more than two groups each with their own estimation direction in CNLS-d estimation.

\paragraph{Simulations to investigate the frequency with which multiple directions leads to violations}
\label{sec:A.4.DetSimsTh2}



We run simulations to investigate the effects of using multiple directions. We use the same DGP as stated in Section \ref{sec:5.DirMatters}, Example 1. 
However, we define two groups and assign different directions for each one of them: 
\begin{eqnarray}
\label{eq:2Dirs_GroupsDef}
    G_1 &=& \{i \in \{1,2,...,n\} | \arctan\left( \tilde{y}_{i2} / \tilde{y}_{i1}\right) \leq \pi/4  \} \\
    G_2 &=& \{i \in \{1,2,...,n\} | \arctan\left( \tilde{y}_{i2} / \tilde{y}_{i1}\right) > \pi/4  \},
\end{eqnarray}

and, 
\begin{equation}
\label{eq:2Dirs_DirsDef}
    \bm{g^y} = \begin{cases}
        \bm{g^{y_{G_1}}}, & \text{if $i \in G_1$}\\
        \bm{g^{y_{G_2}}}, & \text{if $i \in G_2$},
    \end{cases}
\end{equation}

\noindent where $\bm{g^{y_{G_1}}} = \left[\cos(\pi/8), \sin(\pi/8)\right]$ and $\bm{g^{y_{G_2}}} = \left[\cos(3\pi/8), \sin(3\pi/8)\right]$.

We run a total of $100$ simulations. 
For comparison, for each simulation, we also record the estimates when using only the direction based on $\pi/8$ and $3\pi/8$ only for all observations. We identify violations of the monotonicity and concavity by sorting the estimates by $y_1$. We identify all adjacent pairs and triplets, which means 99 pairs and 98 triplets given that we consider 100 observations for each simulation.

As expected, there are no violations when we use a single direction for the estimation. However, when we use two directions violations are observed. For monotonicity, we observe no violations for pairs of observations that are part of the same group. However, for pairs with one member from each group we observe violations of monotonicity for 6\% of the pairs. We use the triplets to analyze concavity. When the members of the triplet are from the same group, we observe violations of concavity for 2\% of the triplets. When one member of the triplet is from a different group, the violations of concavity increase to 45\%.  
These results indicate that for one instance when the conditions of Theorem \ref{thm:ConvCond} do not hold, we see a significant number of violations of the maintained assumptions.

\pagebreak
\section{Additional Experiments}
\label{sec:Ax.B.AdExp}


\subsection{Experiments Related to Section \ref{sec:4.1.ExpIllustrate} - with the Linear Estimator.}
\label{sec:Ax.B.AdExp.1}

\paragraph{Measuring MSE Example, Section \ref{sec:4.1.ExpIllustrate} - Noise Generated in a Common and Prespecified Direction $\theta_f$} \mbox{} \\
This section describes the simulations and the results for the fixed noise direction case referenced in Section \ref{sec:4.1.ExpIllustrate}.

The Data Generation Process (DGP) for observations \(\left(\bm{y}_{i},c_{i}\right) , i = 1,\ldots,n\), is as follows: 

\begin{mdframed}

\begin{enumerate}
    \item{The output, \(\tilde{y}_{i}\),  is drawn from the continuous uniform distribution \(U\left[0,1\right]\)}
    \item{The cost is calculated as \(\tilde{c}_{i} = \beta_0 \ \tilde{y}_{i}\), where $\beta_0 = 1$.}
        \item{In the case of fixed direction, the noise term is determined as:}
            \begin{enumerate}
            \item{\(l_{\epsilon_i}\) is the scalar length that is drawn from a normal distribution, \(N\left(0,\lambda\,\epsilon_0\right)\), \(\lambda\) is prespecified and an initial value for the standard deviation, $\epsilon_0$, is calculated as in Equation \eqref{eq:ep0} in Section \ref{sec:4.1.ExpIllustrate}.}:
            \begin{equation}
            \label{eq:ep0_bis}
            \epsilon_0 = \frac{1}{2} \left[ \sqrt{\frac{1}{n-1}\sum_{i=1}^n\left(\tilde{y}_{i}-\bar{y}\right)^2} +\sqrt{\frac{1}{n-1}\sum_{i=1}^n\left(\tilde{c}_{i}-\bar{c}\right)^2} \right],
            \end{equation}
            where \(\bar{y} = \frac{1}{n}\sum_{i=1}^{n} {\tilde{y}_{i}}\) and \(\bar{c} = \frac{1}{n}\sum_{i=1}^{n}{\tilde{c}_{i}}\) are the mean of the output and the mean of the cost without noise, respectively.
            
            \item{\(\bm{v_f} = \left[\cos(\theta_f),\sin(\theta_f)\right]\)\ is the fixed noise direction that is inferred from the prespecified angle $\theta_f$.}
            
            \item{\( \left(\epsilon_{y_{i}},\epsilon_{c_{i}}\right) = l_{\epsilon_i} \, \bm{v_f}, \ i = 1,\ldots,n. \)}
        \end{enumerate}
    \item{The observations with noise are obtained by appending the noise term: 
    \begin{equation}
    \label{ObsAndNoise_bis}
        \binom{y_{i}}{c_{i}} = \binom{\tilde{y}_{i}}{\tilde{c}_{i}} + \binom{\epsilon_{y_{i}}}{\epsilon_{c_{i}}} , \ i = 1,\ldots,n.
    \end{equation}
    }
\end{enumerate}
\end{mdframed}

\begin{figure}[h!]
	\caption{Linear function data generation process with fixed noise direction}
	\label{fig:LinearDGPStatementFixed}
\end{figure}

Apply the DGP described above to generate a training set, \(\left(y_{tr_{i}},c_{tr_{i}}\right), \ i = 1,\ldots,n_{tr}\), and a testing set \(\left(y_{ts_{i}},c_{ts_{i}}\right), \ i = 1,\ldots,n_{ts}\).
Consider $100$ repetitions of the simulation and set the number of observations in each group to $n_{tr} = n_{ts} = 100$. Set the scaling coefficient for the noise to $\lambda = 0.6$. Consider different DGP since data is generated for the following values of noise direction angles, $\theta_f \in \left\{0, \pi/8, \pi/4, 3\pi/8, \pi/2\right\}$. 

We test the set of directions corresponding to the angle 
$\theta_t \in \left\{0, \pi/8, \pi/4, 3\pi/8, \pi/2\right\}$. If the direction of the noise, $\theta_f$, matches the direction used in the DDF, $\theta_t$, then the smallest MSE results for all cases.

\paragraph{Results: Fixed Noise Direction}
\label{sec:TrainTestResFixedDir}

Table \ref{tab:2.LinFixedNoiseDirTrueMSE} reports the MSE computed by comparing the estimated function to the true function and Table \ref{tab:3.LinFixedNoiseDirNoisyMSE} reports the MSE computed by comparing the estimated function to the testing set.

\begin{table}[h!]
  \centering
  \caption{Average  MSE  over  100  simulations  for  the  Linear  Estimator  compared  to  the true function with a DGP using random noise directions}
    \begin{tabular}{c|c|ccccc}
    \multicolumn{2}{r|}{} & \multicolumn{5}{c}{Average MSE: Estimator } \\
    \multicolumn{2}{r|}{} & \multicolumn{5}{c}{compared to the true function} \\
    \multicolumn{2}{r|}{} & \multicolumn{5}{c}{DDF Direction Angle $\theta_t$} \\
\cmidrule{3-7}    \multicolumn{1}{l|}{Noise Dir Angle $\theta_f$ } & \multicolumn{1}{l|}{MSE Dir Ang $\theta_{\textit{MSE}}$} & \(0\)  &  \( \pi / 8\) &   \( \pi / 4\)    &  \( 3\pi / 8\)  &   \( \pi / 2\) \\
    \midrule
    \(0\)  & \(0\) &  \textbf{0.55} & 1.59  & 3.49  & 6.35  & 12.06 \\
    \(0\)  & \( \pi / 8\)  & \textbf{0.32} & 0.86  & 1.81  & 3.17  & 5.70 \\
    \(0\)  & \( \pi / 4\)  & \textbf{0.27} & 0.69  & 1.42  & 2.44  & 4.23 \\
    \(0\)  & \( 3\pi / 8\) & \textbf{0.32} & 0.77  & 1.54  & 2.58  & 4.36 \\
    \(0\)  & \( \pi / 2\)  & \textbf{0.54} & 1.21  & 2.37  & 3.86  & 6.28 \\
    \midrule
    \( \pi / 8\)  & \(0\)        & 3.22  & \textbf{1.00} & 2.66  & 7.79  & 22.92 \\
    \( \pi / 8\)  & \( \pi / 8\) & 2.16  & \textbf{0.59} & 1.39  & 3.80  & 9.98 \\
    \( \pi / 8\)  & \( \pi / 4\) & 2.04  & \textbf{0.50} & 1.10  & 2.88  & 7.09 \\
    \( \pi / 8\)  & \( 3\pi / 8\)& 2.67  & \textbf{0.59} & 1.21  & 3.02  & 7.02 \\
    \( \pi / 8\)  & \( \pi / 2\) & 5.40  & \textbf{1.03} & 1.88  & 4.45  & 9.68 \\
    \midrule
    \( \pi / 4\)  & \(0\)  & 8.95  & 2.92  & \textbf{1.18} & 2.95  & 15.94 \\
    \( \pi / 4\)  & \( \pi / 8\)  & 6.46  & 1.93  & \textbf{0.70} & 1.53  & 7.21 \\
    \( \pi / 4\)  & \( \pi / 4\)  & 6.49  & 1.81  & \textbf{0.61} & 1.20  & 5.24 \\
    \( \pi / 4\)  & \( 3\pi / 8\)  & 9.10  & 2.35  & \textbf{0.74} & 1.31  & 5.30 \\
    \( \pi / 4\)  & \( \pi / 2\)  & 20.84 & 4.70  & \textbf{1.32} & 2.03  & 7.48 \\
    \midrule
    \( 3\pi / 8\)  & \(0\)  & 9.65  & 4.44  & 1.90  & \textbf{1.13} & 5.70 \\
    \( 3\pi / 8\)  & \( \pi / 8\)  & 6.99  & 3.00  & 1.22  & \textbf{0.65} & 2.83 \\
    \( 3\pi / 8\)  & \( \pi / 4\)  & 7.05  & 2.86  & 1.11  & \textbf{0.55} & 2.17 \\
    \( 3\pi / 8\)  & \( 3\pi / 8\)  & 9.92  & 3.76  & 1.40  & \textbf{0.64} & 2.30 \\
    \( 3\pi / 8\)  & \( \pi / 2\)  & 22.76 & 7.71  & 2.66  & \textbf{1.09} & 3.45 \\
    \midrule
    \( \pi / 2\)  & \(0\)  & 6.15  & 3.76  & 2.29  & 1.16  & \textbf{0.50} \\
    \( \pi / 2\)  & \( \pi / 8\)  & 4.25  & 2.50  & 1.49  & 0.73  & \textbf{0.29} \\
    \( \pi / 2\)  & \( \pi / 4\)  & 4.11  & 2.36  & 1.37  & 0.66  & \textbf{0.25} \\
    \( \pi / 2\)  & \( 3\pi / 8\)  & 5.52  & 3.06  & 1.74  & 0.81  & \textbf{0.29} \\
    \( \pi / 2\)  & \( \pi / 2\)7  & 11.62 & 6.10  & 3.33  & 1.50  & \textbf{0.49} \\
    \midrule
    \multicolumn{7}{l}{\textit{Note: Displayed are the measured values multiplied by} \(10^3\)} \\
    \end{tabular}%
  \label{tab:2.LinFixedNoiseDirTrueMSE}%
\end{table}%

\begin{table}[h!]
  \centering
  \caption{Average  MSE  over  100  simulations  for  the  Linear  Estimator  compared  to an out-of-sample testing set with a DGP using fixed noise directions}
    \begin{tabular}{c|c|ccccc}
    \multicolumn{2}{r|}{}     & \multicolumn{5}{c}{Average MSE: Estimator } \\
    \multicolumn{2}{r|}{}     & \multicolumn{5}{c}{compared to testing set data} \\
    \multicolumn{2}{r|}{}     & \multicolumn{5}{c}{DDF Direction Angle $\theta_t$} \\
\cmidrule{3-7}    \multicolumn{1}{l|}{Noise Dir Angle $\theta_f$} & \multicolumn{1}{l|}{MSE Dir Ang $\theta_{\textit{MSE}}$} & \(0\)  &  \( \pi / 8\) &   \( \pi / 4\)    &  \( 3\pi / 8\)  &   \( \pi / 2\) \\
    \midrule
    \(0\)  & \(0\)  & \textbf{30.02} & 31.22 & 33.23 & 36.21 & 42.08 \\
    \(0\)  & \( \pi / 8\)  & 17.53 & \textbf{17.13} & 17.46 & 18.24 & 20.01 \\
    \(0\)  & \( \pi / 4\)  & 14.95 & 13.99 & \textbf{13.86} & 14.10 & 14.92 \\
    \(0\)  & \( 3\pi / 8\)  & 17.51 & 15.70 & 15.15 & \textbf{15.03} & 15.42 \\
    \(0\)  & \( \pi / 2\)  & 29.93 & 25.30 & 23.55 & 22.64 & \textbf{22.32} \\
    \midrule
    \( \pi / 8\)  & \(0\)  & \textbf{49.89} & 52.78 & 58.59 & 68.39 & 91.28 \\
    \( \pi / 8\)  & \( \pi / 8\)  & 32.41 & \textbf{30.88} & 31.71 & 34.14 & 40.37 \\
    \( \pi / 8\)  & \( \pi / 4\)  & 29.93 & 26.38 & \textbf{25.69} & 26.27 & 28.92 \\
    \( \pi / 8\)  & \( 3\pi / 8\)  & 38.15 & 31.00 & 28.66 & \textbf{27.92} & 28.88 \\
    \( \pi / 8\)  & \( \pi / 2\)  & 74.19 & 53.30 & 45.83 & 41.93 & \textbf{40.19} \\
    \midrule
    \( \pi / 4\)  & \(0\)  & \textbf{51.54} & 53.79 & 59.55 & 70.76 & 101.99 \\
    \( \pi / 4\)  & \( \pi / 8\)  & 36.65 & \textbf{34.53} & 35.21 & 38.14 & 47.22 \\
    \( \pi / 4\)  & \( \pi / 4\)  & 36.39 & 31.60 & \textbf{30.32} & 30.83 & 34.75 \\
    \( \pi / 4\)  & \( 3\pi / 8\)  & 50.32 & 39.87 & 35.91 & \textbf{34.32} & 35.52 \\
    \( \pi / 4\)  & \( \pi / 2\)  & 112.21 & 76.31 & 62.47 & 54.76 & \textbf{50.83} \\
    \midrule
    \( 3\pi / 8\)  & \(0\)  & \textbf{39.37} & 41.09 & 45.01 & 52.54 & 73.64 \\
    \( 3\pi / 8\)  & \( \pi / 8\)  & 28.28 & \textbf{27.35} & 28.14 & 30.56 & 37.89 \\
    \( 3\pi / 8\)  & \( \pi / 4\)  & 28.30 & 25.72 & \textbf{25.22} & 26.01 & 29.73 \\
    \( 3\pi / 8\)  & \( 3\pi / 8\)  & 39.47 & 33.40 & 31.11 & \textbf{30.42} & 32.19 \\
    \( 3\pi / 8\)  & \( \pi / 2\)  & 89.14 & 66.84 & 57.41 & 51.96 & \textbf{49.51} \\
    \midrule
    \( \pi / 2\)  & \(0\)  & \textbf{22.47} & 22.94 & 23.97 & 25.85 & 30.66 \\
    \( \pi / 2\)  & \( \pi / 8\)  & 15.44 & \textbf{15.16} & 15.36 & 15.99 & 17.91 \\
    \( \pi / 2\)  & \( \pi / 4\)  & 14.89 & 14.17 & \textbf{14.01} & 14.21 & 15.27 \\
    \( \pi / 2\)  & \( 3\pi / 8\)  & 19.88 & 18.27 & 17.59 & \textbf{17.35} & 17.88 \\
    \( \pi / 2\)  & \( \pi / 2\)  & 41.52 & 36.04 & 33.31 & 31.51 & \textbf{30.54} \\
    \midrule
    \multicolumn{7}{l}{\textit{Note: Displayed are the measured values multiplied by} \(10^3\)} \\
    \end{tabular}%
  \label{tab:3.LinFixedNoiseDirNoisyMSE}%
\end{table}%

In Table \ref{tab:2.LinFixedNoiseDirTrueMSE}, the direction for the DDF corresponding to the smallest MSE always matches the noise direction in the DGP. Further for more than $70 \%$ of the cases tested there is more than a $50 \%$ decrease in MSE by using the correctly specified direction compared to the next best direction tested, which was not as large in the random direction case in Table \ref{tab:LinRandomNoiseTrueFunction} of Section \ref{sec:4.1.ExpIllustrate}. In other words, when endogeneity is severe, the benefits of using a DDF with a well-selected direction are potentially large. 

Table \ref{tab:3.LinFixedNoiseDirNoisyMSE} is consistent with the results observed in the random noise case, in Table \ref{tab:LinRandomNoiseOSMSE} of Section \ref{sec:4.1.ExpIllustrate}. The DDF directions corresponding to the smallest MSE values are those matching the directions used for the MSE computation. 
Thus, the proposed radial MSE measure addresses the challenge of measuring performance in applications with a testing dataset. 

\paragraph{Monte Carlo Simulations - Experiments, Section \ref{sec:exp3} - Experiment 3. Base case with fixed noise direction and different noise levels}  \mbox{} \\
This section summarizes the results of Experiment 3 with \(\lambda = 0.2\).

\begin{table}[h!]
  \centering
  \caption{Experiment 3--More Noise: Values of radial MSE relative to the true function varying the DGP noise direction and the CNLS-d direction. In the DGP, the standard deviation of the noise distribution, \(\lambda\), is 0.2.}
    \begin{tabular}{crrrrr}
    \toprule
          & \multicolumn{5}{c}{CNLS-d Direction Angle} \\
    \cmidrule{2-6}    
    \multicolumn{1}{l}{Noise Direction Angle} &      \multicolumn{1}{c}{\(0\)}  &  \multicolumn{1}{c}{\( \pi / 8\)} &   \multicolumn{1}{c}{\( \pi / 4\)}    &  \multicolumn{1}{c}{\( 3\pi / 8\)}  &   \multicolumn{1}{c}{\(\pi / 2\)}  \\
    \midrule
    \(0\)         & \textbf{8.15} & 15.62 & 37.66 & 82.16 & 183.39 \\
    \( \pi / 8\)  & 50.60 & \textbf{11.59} & 20.68 & 67.88 & 206.46 \\
    \( \pi / 4\)  & 145.21 & 29.40 & \textbf{11.89} & 33.89 & 149.24 \\
    \( 3\pi / 8\) & 220.24 & 69.87 & 22.28 & \textbf{11.66} & 53.66 \\
    \( \pi / 2\)  & 165.84 & 72.13 & 33.27 & 14.25 & \textbf{7.41} \\
    \midrule
    \multicolumn{6}{l}{\textit{Note: Displayed are measured values multiplied by }\(10^4\)} \\
    \end{tabular}%
  \label{tab:UnifFixedLarger}%
\end{table}%


\subsection{Experiments related to Section \ref{sec:Exps} - with CNLS-d}
\label{sec:Ax.B.AdExp.2}

Here we complete Section \ref{sec:Exps} with additional experiments and we follow the numbering experiments numbering established then.

 \paragraph{Experiment 5: Base case with different distributions for the initial observations on the true function} \mbox{} 
\label{sec:exp5}

In Experiment 5, we extend the analysis performed in Experiment 4. We consider additional distributions of the DGP for the angle, \(\theta_i, \ i = 1,\ldots,n\) and see how it affects the optimal direction. Unlike Experiment 4, we don't consider only normal distributions, instead we consider the following: a normal distribution,  \(N\left(\frac{\pi}{4},\frac{\pi}{16}\right)\), and two gamma distributions, \(\Gamma\left(3,\frac{\pi}{2}\right)\) and \(\Gamma\left(.5,\frac{\pi}{24}\right)\). For the gamma distributions, the first parameter corresponds to the shape coefficient and the second the scale coefficient. Each distribution is later referenced respectively as $Normal$, $Gamma_1$ and $Gamma_2$. We truncate the tails of the distribution so that the generated angles fall within the range \(\left[0, \pi/2\right]\). Noise is specified as in  Experiment 1. In Figure \ref{fig:Dist_HistAndObsAndMed}, the distributions of the angles \(\theta_i\) are illustrated and in particular the median values are highlighted. Table \ref{tab:Iso_gamma} reports the results of this experiment.

\begin{figure}[h!]
\centerline{
    \includegraphics[width=3in]{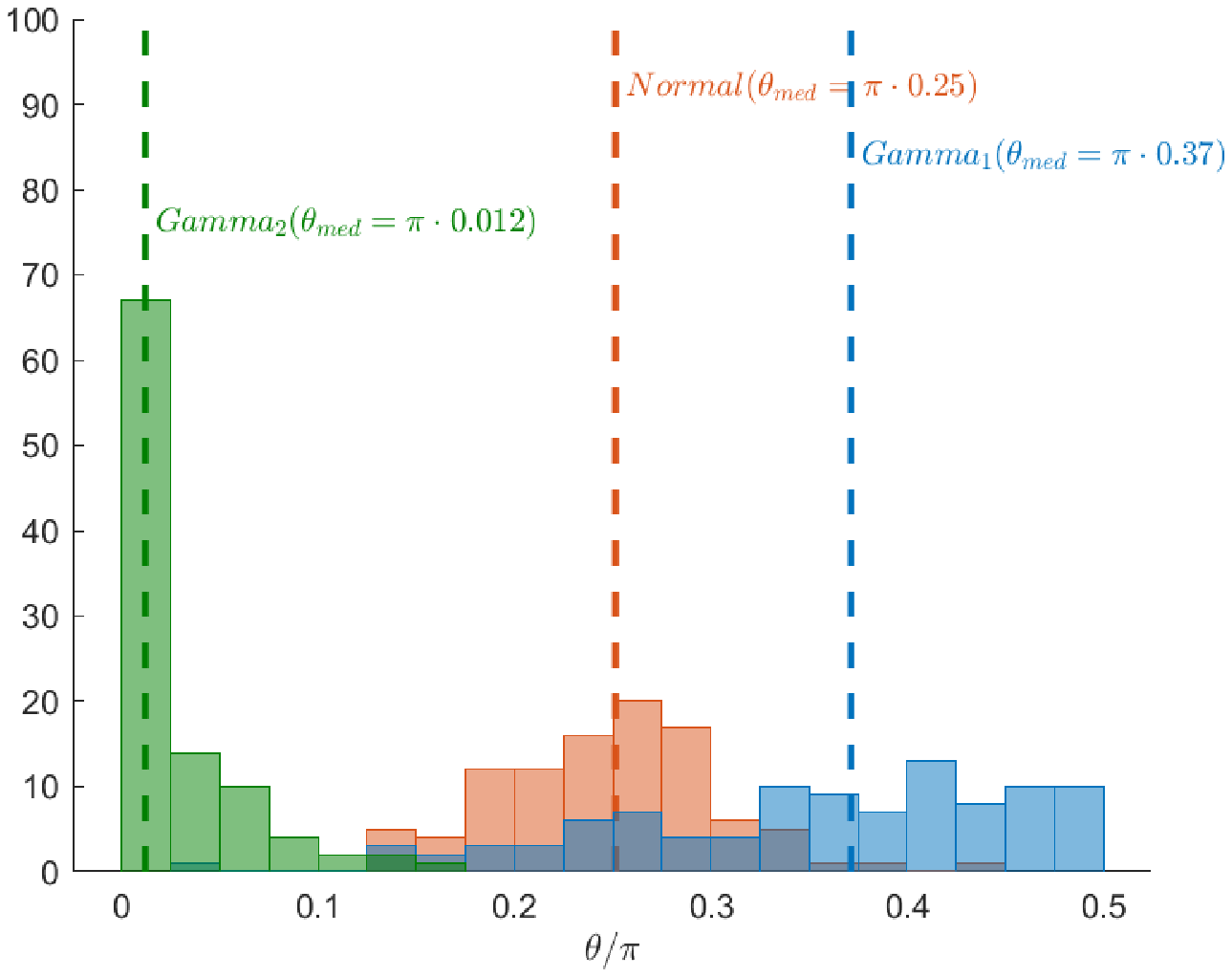}
    \includegraphics[width=3in]{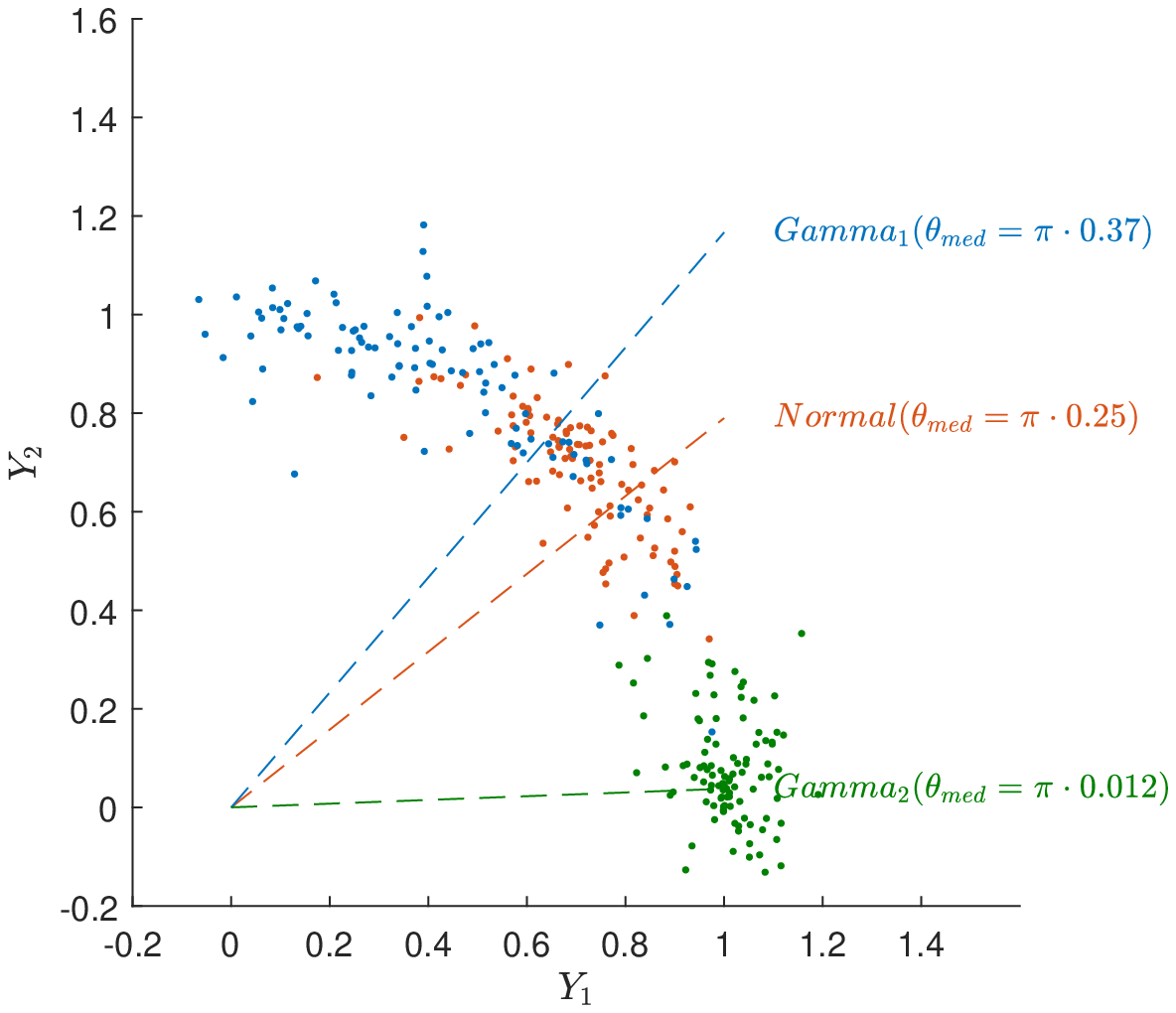}
}
\caption{(a) On the left, histogram of the angles \(\theta_i\) and its median obtained for each distribution when running a simulation with \(100\) observations. (b) On the right, the corresponding observations and the median of the angles \(\theta_i\) for each distribution for the same simulation as the histogram (a).}
\label{fig:Dist_HistAndObsAndMed}
\end{figure}

\begin{table}[h!]
  \centering
  \caption{Experiment 5: Values of radial MSE relative to the true function varying the CNLS-d direction and type of direction used for the DGP.}
    \begin{tabular}{lrrrrr}
    \toprule
    & \multicolumn{5}{c}{CNLS-d Direction Angle} \\
    \cmidrule{2-6}    Distribution \hspace{.1in}  & \multicolumn{1}{c}{\(0\)}  &  \multicolumn{1}{c}{\( \pi / 8\)} &   \multicolumn{1}{c}{\( \pi / 4\)}    &  \multicolumn{1}{c}{\( 3\pi / 8\)}  &   \multicolumn{1}{c}{\(\pi / 2\)} \\
    \midrule
    $Normal$ & 8.45  & 3.04  & \textbf{1.96}  & 3.01  & 8.60 \\
    $Gamma_1$ & 29.34 & 6.92  & 3.27  & \textbf{2.54}  & 3.39 \\
    $Gamma_2$ & \textbf{6.62}  & 9.69  & 19.19 & 72.55 & 598.97 \\
    \midrule
    \multicolumn{6}{l}{\textit{Note: Displayed are measured values multiplied by }\(10^4\).} \\
    \bottomrule
    \end{tabular}%
  \label{tab:Iso_gamma}%
\end{table}%

\bigskip
Two main conclusion can be drawn from the results in Table \ref{tab:Iso_gamma}. First, the smaller the variance of the data distribution, the greater is the importance of direction selection. Looking at the differences between the two gamma distributions, $Gamma_1$ has a larger tail than $Gamma_2$, which means the observations for $Gamma_2$ have a smaller variance. Table \ref{tab:Iso_gamma} indicates that the MSE increases rapidly with deviations from the optimal direction when variance of observations is smaller as with $Gamma_2$ compared to $Gamma_1$. Second, among the directions tested, \(\theta_i\), MSE is minimized for the direction closest to the direction corresponding to the median of the distribution. This second point supports the selection approach proposed in Section \ref{sec:6.DirSelApp}.

\pagebreak


\paragraph{Experiment 6: Adaptation of the Base Case to a 3-Dimensional Case} \mbox{} 
\label{sec:exp6}

We adapt the DGP from Experiment 1, the base case. 
We consider a fixed input level and approximate a three output isoquant, \(Q=3\). 
Indexing the outputs by $q$ and observations by $i$, we define the outputs, 

\begin{equation}
\label{eq:Exp1_TrueVars_3d}
    y_{qi}=\tilde{y}_{qi}+\epsilon_{qi}, \ q=1,\ldots,Q , \ i=1,\ldots,n , 
\end{equation}

\noindent where \(\bm{\tilde{y}_i}\) is the observation on the isoquant and \(\bm{\epsilon_i}\) is the noise. 
The output levels \(\tilde{y}_{qi}, \ q=1,\ldots,Q, \ i=1,\ldots,n\) are generated:

\begin{eqnarray}
\label{eq:Exp1_DetailedOutputs_3d}
    \bm{\tilde{y}}_{i} = \frac{\bm{l}_i}{\lvert\lvert{\bm{l}_i}\rvert\rvert_2}, \ i=1,\ldots,n
\end{eqnarray}

\noindent where \(l_{qi}, \ q=1,\ldots,Q, \ i=1,\ldots,n\), are drawn randomly from a continuous uniform distribution, \(U\left[0,1\right]\). 

The noise terms \(\bm{\epsilon}_{i} , \ i=1,\ldots,n\) is adapted to the 3-dimensional isoquant:

\begin{eqnarray}
\label{eq:Exp1_DetailedNoise_3d}
    \bm{\epsilon}_{i}= l_{\epsilon_i}\; \bm{v}_i, \ i=1,\ldots,n
\end{eqnarray}

\noindent where the length \(l_{\epsilon_i}\) is drawn from the normal distribution \(N\left(0,\lambda \right)\), and \(v_{qi} = \frac{v^*_{qi}}{\lvert\lvert \bm{v}_i \rvert\rvert_2}, \ q=1,\ldots,Q, \ i=1,\ldots,n\) for which \(v^*_{qi}\) are drawn from a continuous uniform distribution \(U\left[-1,1\right]\).

In Experiment 6, 19 directions are considered for the CNLS-d estimators. The directions are determined using the following steps: 
\begin{enumerate}
\item{enumerate all 3 component vectors, corresponding to \(\mathbb{R}^3\) with elements from the set \(\{0,0.5,1\}\) and excluding \(\left(0,0,0\right)\)};
\item{normalize the direction vectors dividing them by their respective \uppercase{e}uclidean norms};
\item{eliminate duplicates} 
\end{enumerate}
The 19 directions are represented by the markers in Figure \ref{fig:Iso_3d} and create a balanced grid on the eighth of a unit sphere, our isoquant. The median direction is \([1/\sqrt{3}, 1/\sqrt{3}, 1/\sqrt{3}]=[.58, 58, .58]\). The standard deviation of the normal distribution is \(\lambda = 0.1\). We perform this experiment \(100\) times for each direction. We report the averaged radial MSE values on a testing set of \(n\) observations lying on the true function in Table \ref{tab:Iso_3d}. In addition to the table, the MSE results are also illustrated in Figure \ref{fig:Iso_3d} where the size of the markers has a positive affine relation with the MSE values and that in the color range from yellow to red, with larger the MSE values associated with more red markers.

\begin{figure}[h!]
\centering
\includegraphics[width=6in]{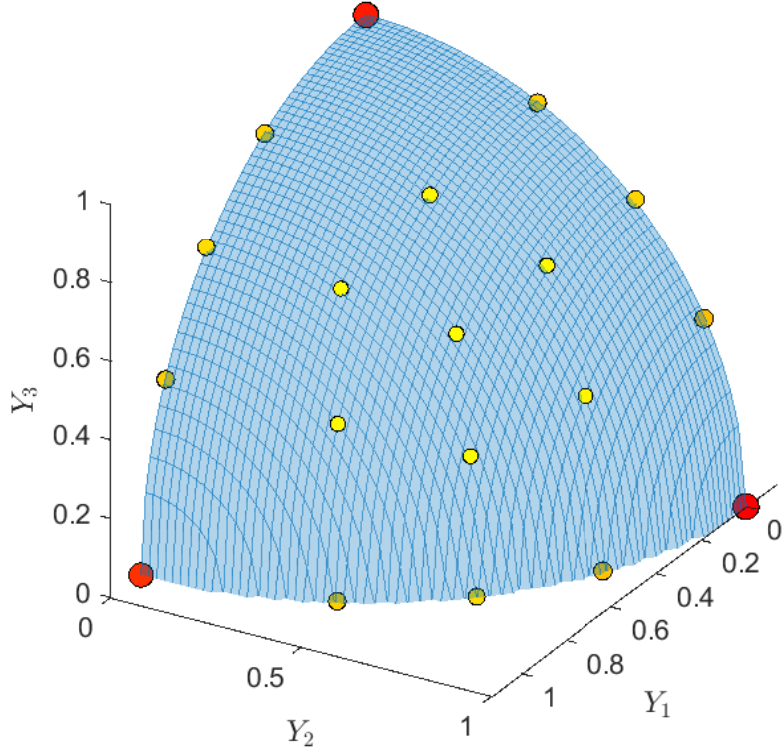}
\caption{Experiment 6: 3-dimensional isoquant case: Representation of the directions tested and the values of averaged radial MSE, the size of the markers having a positive affine relation to the values. The color is an other indicator as the more red the higher the averaged radial MSE is and the more yellow for lower values.}
\label{fig:Iso_3d}
\end{figure}

\begin{table}[h!]
  \centering
  \caption{Experiment 6: Values of radial MSE relative to the true function varying the CNLS-d direction in a 3- dimensional case}
  \begin{tabular}{lc}
  \toprule
    CNLS-d Direction &  \\
    $\left(g^{y_1},g^{y_2},g^{y_3}\right)$ & \multicolumn{1}{c}{Average of radial MSE} \\
    \midrule
    $(0,0,1)$ & 9.07 \\
    $(0,0.45,0.89)$ & 5.23 \\
    $(0,0.71,0.71)$ & 5.04 \\
    $(0,0.89,0.45)$ & 5.53 \\
    $(0,1,0)$ & 9.62 \\
    $(0.33,0.67,0.67)$ & 4.24 \\
    $(0.41,0.41,0.82)$ & 4.29 \\
    $(0.41,0.82,0.41)$ & 4.35 \\
    $(0.45,0,0.89)$ & 5.12 \\
    $(0.45,0.89,0)$ & 5.44 \\
    $(0.58,0.58,0.58)$ & 4.21 \\
    $(0.67,0.33,0.67)$ & 4.15 \\
    $(0.67,0.67,0.33)$ & 4.18 \\
    $(0.71,0,0.71)$ & 4.89 \\
    $(0.71,0.71,0)$ & 4.91 \\
    $(0.82,0.41,0.41)$ & 4.23 \\
    $(0.89,0,0.45)$ & 5.20 \\
    $(0.89,0.45,0)$ & 5.18 \\
    $(1,0,0)$ & 8.58 \\
    \midrule
    \multicolumn{2}{l}{\textit{Note: Displayed are measured values multiplied by }\(10^4\).} \\
    \bottomrule
    \end{tabular}%
  \label{tab:Iso_3d}%
\end{table}%

\bigskip

We can establish three categories of directions that correspond to certain ranges of MSE values. The first category corresponds to the worst MSE values, which are almost twice the smallest values. These are the directions that have only one non-zero component shown with red markers on the corners of the surface shown in Figure \ref{fig:Iso_3d}. The second category is for the MSE values that are above \(5\cdot10^{-4}\) but less than \(8\cdot10^{-4}\). These directions are labeled with the orange markers in Figure \ref{fig:Iso_3d} that are on the edges of the surface but not the corners. One of their directional components, $\left(\bm{g}^x,\bm{g}^y\right)$, is zero but all others are not. The third category of directions, which has the smallest MSEs, correspond to the yellow markers in Figure \ref{fig:Iso_3d}. These directions have only positive components.
Thus, we observe a trend that the directions that have positive components in all variables correspond to the best MSE values. The median value 
direction, \([0.58,0.58,0.58]\), is among the yellow markers. 
These results 
support the selection approach proposed in Section \ref{sec:6.DirSelApp} and confirm the results obtained on the US hospitals data set.

\pagebreak

\section{U.S. Hospital Dataset Application}
\label{sec:Ax.C.HospAppAnnex}




We describe the functional estimates provided by quadratic regression, CNLS-d using a direction with equal components in all dimensions and CNLS-d using the median direction, and the local linear kernel. Table \ref{tab:MPSS_c-level_81rows} provides most productive scale size (MPSS) measurements in cost in $\$M$.  Tables \ref{tab:MC_MajDiag_81rows} and \ref{tab:MC_MajTher_81rows} provide the marginal cost of Minor Therapeutic procedures and the marginal cost of Major Therapeutic procedures, respectively. The units for Tables \ref{tab:MC_MajDiag_81rows} and \ref{tab:MC_MajTher_81rows} are cost  in $\$k$ over Minor and Major Therapeutic procedures, respectively.

Our conclusions are the same as stated in the body of the paper, CNLS-d provides the advantage of being more flexible than the parametric estimator (quadratic regression) while having shape constraints that maintain the interpretability of the results. 

\begin{table}[htbp]
  \tiny
  \centering
  \caption{Most Productive Scale Size (c)}
    \begin{tabular}{cccc|ccc|ccc|ccc|ccc}
    \toprule
    \multicolumn{4}{c|}{Percentile} & \multicolumn{3}{c|}{Quadratic Regression} & \multicolumn{3}{c|}{CNLS-d (median)} & \multicolumn{3}{c}{CNLS-d (equal)} & \multicolumn{3}{c}{LL Kernel}\\
    \midrule
    \multicolumn{1}{c}{MinDiag} & MinTher & MajDiag & MajTher & 2007  & 2008  & 2009  & 2007  & 2008  & 2009  & 2007  & 2008  & 2009 & 2007  & 2008  & 2009\\
    \midrule
    25    & 25    & 25    & 25    & 529   & 234   & 823   & 105   & 58    & 87    & 109   & 93    & 103   & 228   & 301   & 885 \\
    25    & 25    & 25    & 50    & 118   & 118   & 122   & 406   & 81    & 369   & 412   & 91    & 112   & 220   & 143   & 216 \\
    25    & 25    & 25    & 75    & 102   & 110   & 93    & 1214  & 82    & 895   & 1220  & 92    & 104   & 209   & 120   & 189 \\
    25    & 25    & 50    & 25    & 79    & 693   & 560   & 177   & 94    & 166   & 131   & 93    & 140   & 226   & 136   & 162 \\
    25    & 25    & 50    & 50    & 104   & 139   & 141   & 98    & 60    & 94    & 95    & 79    & 84    & 233   & 217   & 210 \\
    25    & 25    & 50    & 75    & 105   & 114   & 103   & 165   & 80    & 340   & 179   & 90    & 114   & 219   & 139   & 208 \\
    25    & 25    & 75    & 25    & 56    & 414   & 335   & 179   & 108   & 554   & 124   & 96    & 384   & 158   & 133   & 27 \\
    25    & 25    & 75    & 50    & 77    & 245   & 176   & 149   & 93    & 194   & 126   & 91    & 132   & 292   & 117   & 185 \\
    25    & 25    & 75    & 75    & 94    & 133   & 115   & 93    & 61    & 107   & 105   & 71    & 85    & 226   & 197   & 197 \\
    25    & 50    & 25    & 25    & 15    & 42    & 63    & 330   & 53    & 78    & 327   & 123   & 127   & 400   & 271   & 1062 \\
    25    & 50    & 25    & 50    & 1074  & 234   & 1333  & 119   & 55    & 92    & 108   & 91    & 89    & 1027  & 306   & 215 \\
    25    & 50    & 25    & 75    & 137   & 131   & 138   & 209   & 78    & 331   & 264   & 98    & 110   & 222   & 153   & 206 \\
    25    & 50    & 50    & 25    & 248   & 330   & 381   & 80    & 57    & 83    & 93    & 74    & 84    & 373   & 304   & 336 \\
    25    & 50    & 50    & 50    & 332   & 273   & 1349  & 70    & 64    & 86    & 127   & 78    & 80    & 903   & 251   & 1152 \\
    25    & 50    & 50    & 75    & 133   & 134   & 141   & 177   & 76    & 304   & 182   & 95    & 109   & 233   & 261   & 231 \\
    25    & 50    & 75    & 25    & 108   & 492   & 718   & 126   & 91    & 144   & 143   & 89    & 129   & 204   & 187   & 79 \\
    25    & 50    & 75    & 50    & 122   & 694   & 1068  & 128   & 87    & 137   & 144   & 95    & 112   & 331   & 159   & 239 \\
    25    & 50    & 75    & 75    & 118   & 154   & 152   & 91    & 59    & 104   & 110   & 77    & 93    & 239   & 232   & 229 \\
    25    & 75    & 25    & 25    & 11    & 13    & 13    & 915   & 53    & 80    & 1015  & 125   & 130   & 246   & 98    & 921 \\
    25    & 75    & 25    & 50    & 11    & 231   & 149   & 192   & 52    & 78    & 197   & 130   & 136   & 537   & 433   & 1129 \\
    25    & 75    & 25    & 75    & 1139  & 223   & 1542  & 112   & 55    & 75    & 101   & 91    & 115   & 1015  & 287   & 215 \\
    25    & 75    & 50    & 25    & 18    & 16    & 5     & 133   & 52    & 79    & 181   & 114   & 118   & 293   & 111   & 887 \\
    25    & 75    & 50    & 50    & 13    & 311   & 217   & 135   & 51    & 77    & 181   & 111   & 125   & 528   & 466   & 1091 \\
    25    & 75    & 50    & 75    & 1155  & 230   & 1563  & 109   & 61    & 75    & 99    & 90    & 114   & 1062  & 272   & 214 \\
    25    & 75    & 75    & 25    & 64    & 220   & 275   & 81    & 57    & 85    & 94    & 82    & 84    & 300   & 199   & 274 \\
    25    & 75    & 75    & 50    & 304   & 483   & 484   & 79    & 56    & 93    & 85    & 73    & 83    & 478   & 400   & 437 \\
    25    & 75    & 75    & 75    & 333   & 265   & 1532  & 77    & 64    & 96    & 115   & 78    & 79    & 963   & 249   & 153 \\
    50    & 25    & 25    & 25    & 143   & 189   & 126   & 165   & 115   & 149   & 173   & 157   & 183   & 132   & 139   & 123 \\
    50    & 25    & 25    & 50    & 119   & 124   & 105   & 126   & 68    & 143   & 110   & 88    & 98    & 287   & 116   & 197 \\
    50    & 25    & 25    & 75    & 103   & 111   & 90    & 289   & 82    & 424   & 265   & 92    & 104   & 218   & 116   & 185 \\
    50    & 25    & 50    & 25    & 84    & 740   & 157   & 136   & 72    & 258   & 140   & 91    & 277   & 128   & 209   & 93 \\
    50    & 25    & 50    & 50    & 106   & 146   & 124   & 96    & 59    & 100   & 101   & 85    & 113   & 245   & 244   & 202 \\
    50    & 25    & 50    & 75    & 106   & 114   & 100   & 173   & 80    & 292   & 212   & 90    & 113   & 229   & 135   & 211 \\
    50    & 25    & 75    & 25    & 58    & 431   & 217   & 205   & 97    & 452   & 119   & 95    & 440   & 161   & 128   & 7 \\
    50    & 25    & 75    & 50    & 79    & 247   & 160   & 140   & 82    & 192   & 114   & 91    & 154   & 150   & 114   & 150 \\
    50    & 25    & 75    & 75    & 95    & 133   & 111   & 93    & 61    & 106   & 104   & 79    & 106   & 233   & 207   & 197 \\
    50    & 50    & 25    & 25    & 10    & 142   & 207   & 99    & 51    & 75    & 107   & 111   & 112   & 462   & 363   & 1031 \\
    50    & 50    & 25    & 50    & 1156  & 232   & 1319  & 109   & 61    & 81    & 114   & 89    & 80    & 1134  & 367   & 264 \\
    50    & 50    & 25    & 75    & 138   & 131   & 135   & 208   & 78    & 240   & 267   & 98    & 110   & 233   & 150   & 212 \\
    50    & 50    & 50    & 25    & 357   & 387   & 450   & 87    & 56    & 90    & 91    & 80    & 90    & 419   & 324   & 194 \\
    50    & 50    & 50    & 50    & 307   & 272   & 1329  & 76    & 63    & 84    & 88    & 77    & 78    & 218   & 269   & 652 \\
    50    & 50    & 50    & 75    & 134   & 135   & 137   & 185   & 76    & 258   & 183   & 95    & 108   & 236   & 170   & 232 \\
    50    & 50    & 75    & 25    & 110   & 508   & 702   & 125   & 90    & 143   & 124   & 89    & 139   & 209   & 178   & 30 \\
    50    & 50    & 75    & 50    & 123   & 646   & 1044  & 128   & 77    & 147   & 119   & 94    & 132   & 333   & 340   & 178 \\
    50    & 50    & 75    & 75    & 119   & 155   & 149   & 91    & 59    & 103   & 110   & 77    & 103   & 240   & 236   & 236 \\
    50    & 75    & 25    & 25    & 18    & 15    & 6     & 274   & 53    & 80    & 282   & 124   & 130   & 291   & 117   & 933 \\
    50    & 75    & 25    & 50    & 14    & 245   & 142   & 191   & 52    & 77    & 188   & 129   & 126   & 566   & 456   & 1134 \\
    50    & 75    & 25    & 75    & 1155  & 224   & 1523  & 111   & 55    & 75    & 101   & 91    & 115   & 1050  & 348   & 247 \\
    50    & 75    & 50    & 25    & 18    & 13    & 10    & 132   & 52    & 79    & 172   & 114   & 118   & 316   & 140   & 894 \\
    50    & 75    & 50    & 50    & 17    & 325   & 209   & 135   & 51    & 76    & 164   & 111   & 124   & 537   & 502   & 680 \\
    50    & 75    & 50    & 75    & 1170  & 230   & 1544  & 109   & 61    & 83    & 106   & 90    & 114   & 1106  & 308   & 252 \\
    50    & 75    & 75    & 25    & 85    & 232   & 264   & 81    & 57    & 84    & 94    & 82    & 84    & 321   & 205   & 245 \\
    50    & 75    & 75    & 50    & 323   & 493   & 471   & 79    & 56    & 92    & 85    & 81    & 82    & 499   & 406   & 306 \\
    50    & 75    & 75    & 75    & 335   & 266   & 1514  & 77    & 64    & 95    & 115   & 78    & 79    & 966   & 252   & 1192 \\
    75    & 25    & 25    & 25    & 75    & 101   & 29    & 548   & 309   & 213   & 620   & 177   & 136   & 20    & 27    & 34 \\
    75    & 25    & 25    & 50    & 100   & 118   & 50    & 129   & 74    & 176   & 142   & 137   & 128   & 100   & 46    & 73 \\
    75    & 25    & 25    & 75    & 102   & 112   & 81    & 101   & 78    & 133   & 104   & 79    & 99    & 242   & 73    & 160 \\
    75    & 25    & 50    & 25    & 74    & 142   & 39    & 244   & 95    & 190   & 322   & 92    & 285   & 49    & 57    & 42 \\
    75    & 25    & 50    & 50    & 95    & 140   & 59    & 112   & 120   & 154   & 189   & 112   & 386   & 123   & 58    & 81 \\
    75    & 25    & 50    & 75    & 106   & 116   & 83    & 107   & 76    & 131   & 110   & 78    & 99    & 259   & 79    & 179 \\
    75    & 25    & 75    & 25    & 60    & 534   & 65    & 163   & 75    & 260   & 178   & 84    & 355   & 79    & 115   & 17 \\
    75    & 25    & 75    & 50    & 80    & 213   & 82    & 139   & 72    & 237   & 179   & 81    & 280   & 129   & 147   & 135 \\
    75    & 25    & 75    & 75    & 96    & 137   & 96    & 91    & 69    & 111   & 99    & 84    & 114   & 242   & 117   & 188 \\
    75    & 50    & 25    & 25    & 233   & 593   & 136   & 232   & 128   & 157   & 229   & 254   & 130   & 109   & 542   & 677 \\
    75    & 50    & 25    & 50    & 185   & 196   & 138   & 171   & 93    & 156   & 145   & 145   & 121   & 154   & 146   & 135 \\
    75    & 50    & 25    & 75    & 137   & 132   & 115   & 107   & 75    & 118   & 101   & 76    & 106   & 243   & 111   & 182 \\
    75    & 50    & 50    & 25    & 175   & 670   & 149   & 133   & 85    & 132   & 133   & 98    & 118   & 135   & 278   & 412 \\
    75    & 50    & 50    & 50    & 169   & 223   & 149   & 120   & 98    & 141   & 121   & 108   & 136   & 179   & 171   & 139 \\
    75    & 50    & 50    & 75    & 133   & 137   & 118   & 104   & 74    & 117   & 107   & 75    & 95    & 300   & 128   & 211 \\
    75    & 50    & 75    & 25    & 101   & 607   & 182   & 106   & 71    & 258   & 156   & 80    & 351   & 139   & 161   & 59 \\
    75    & 50    & 75    & 50    & 117   & 359   & 177   & 102   & 69    & 236   & 150   & 77    & 279   & 171   & 291   & 160 \\
    75    & 50    & 75    & 75    & 120   & 159   & 131   & 97    & 67    & 108   & 97    & 90    & 111   & 274   & 253   & 210 \\
    75    & 75    & 25    & 25    & 11    & 57    & 144   & 92    & 52    & 85    & 101   & 92    & 83    & 380   & 220   & 872 \\
    75    & 75    & 25    & 50    & 12    & 371   & 377   & 88    & 51    & 83    & 105   & 90    & 90    & 642   & 551   & 1051 \\
    75    & 75    & 25    & 75    & 740   & 219   & 452   & 101   & 53    & 81    & 86    & 88    & 88    & 253   & 359   & 237 \\
    75    & 75    & 50    & 25    & 13    & 136   & 202   & 88    & 51    & 84    & 105   & 90    & 92    & 404   & 284   & 866 \\
    75    & 75    & 50    & 50    & 19    & 439   & 428   & 93    & 50    & 82    & 109   & 89    & 90    & 631   & 630   & 1048 \\
    75    & 75    & 50    & 75    & 549   & 225   & 459   & 106   & 59    & 89    & 91    & 88    & 96    & 263   & 349   & 286 \\
    75    & 75    & 75    & 25    & 296   & 328   & 415   & 79    & 48    & 89    & 85    & 86    & 89    & 363   & 233   & 191 \\
    75    & 75    & 75    & 50    & 507   & 564   & 593   & 85    & 48    & 88    & 83    & 77    & 97    & 385   & 423   & 191 \\
    75    & 75    & 75    & 75    & 285   & 261   & 472   & 82    & 56    & 92    & 88    & 76    & 101   & 236   & 264   & 713 \\
    \midrule
    \multicolumn{13}{l}{\textit{Note: The values displayed are in \$M}}\\
    \end{tabular}%
  \label{tab:MPSS_c-level_81rows}%
\end{table}%

\begin{table}[htbp]
  \tiny
  \centering
  \caption{Marginal Cost of Minor Therapeutic Procedures}
    \begin{tabular}{cccc|ccc|ccc|ccc|ccc}
    \toprule
    \multicolumn{4}{c|}{Percentile} & \multicolumn{3}{c|}{Quadratic Regression} & \multicolumn{3}{c|}{CNLS-d (median)} & \multicolumn{3}{c}{CNLS-d (equal)} & \multicolumn{3}{c}{LL Kernel}\\
    \midrule
    \multicolumn{1}{c}{MinDiag} & MinTher & MajDiag & MajTher & 2007  & 2008  & 2009  & 2007  & 2008  & 2009  & 2007  & 2008  & 2009 & 2007  & 2008  & 2009\\
        \midrule
    25    & 25    & 25    & 25    & 8.9   & 6.5   & 13.2  & 1.3   & 1.5   & 2.5   & 0.3   & 0.3   & 1.5   & 6.3   & 10.4  & 4.7 \\
    25    & 25    & 25    & 50    & 8.9   & 6.5   & 13.2  & 0.1   & 0.1   & 0.0   & 0.3   & 0.1   & 0.1   & 6.1   & 9.9   & 3.8 \\
    25    & 25    & 25    & 75    & 8.9   & 6.5   & 13.2  & 0.1   & 0.0   & 0.0   & 0.1   & 0.0   & 0.0   & 5.1   & 7.8   & 4.2 \\
    25    & 25    & 50    & 25    & 8.9   & 6.5   & 13.2  & 0.0   & 0.1   & 0.4   & 0.0   & 0.0   & 0.6   & 6.5   & 10.7  & 5.9 \\
    25    & 25    & 50    & 50    & 8.9   & 6.5   & 13.2  & 0.1   & 0.1   & 0.2   & 0.0   & 0.1   & 0.5   & 6.4   & 10.2  & 5.1 \\
    25    & 25    & 50    & 75    & 8.9   & 6.5   & 13.2  & 0.2   & 0.0   & 0.0   & 0.1   & 0.0   & 0.0   & 5.5   & 8.0   & 4.6 \\
    25    & 25    & 75    & 25    & 8.9   & 6.5   & 13.2  & 0.0   & 0.1   & 0.1   & 0.1   & 0.0   & 0.6   & 6.8   & 10.0  & 8.2 \\
    25    & 25    & 75    & 50    & 8.9   & 6.5   & 13.2  & 0.0   & 0.1   & 0.0   & 0.0   & 0.0   & 0.2   & 6.8   & 9.6   & 7.6 \\
    25    & 25    & 75    & 75    & 8.9   & 6.5   & 13.2  & 0.0   & 0.0   & 0.0   & 0.0   & 0.0   & 0.1   & 5.9   & 7.8   & 6.4 \\
    25    & 50    & 25    & 25    & 8.1   & 6.1   & 12.4  & 7.3   & 8.7   & 10.3  & 8.0   & 8.1   & 9.6   & 5.0   & 10.7  & 6.2 \\
    25    & 50    & 25    & 50    & 8.1   & 6.1   & 12.4  & 2.8   & 7.1   & 8.3   & 4.9   & 4.5   & 8.0   & 4.8   & 9.5   & 4.8 \\
    25    & 50    & 25    & 75    & 8.1   & 6.1   & 12.4  & 1.4   & 0.2   & 0.0   & 0.1   & 0.0   & 0.0   & 4.3   & 7.0   & 3.6 \\
    25    & 50    & 50    & 25    & 8.1   & 6.1   & 12.4  & 6.9   & 5.8   & 7.7   & 5.9   & 5.9   & 6.0   & 5.3   & 10.5  & 7.8 \\
    25    & 50    & 50    & 50    & 8.1   & 6.1   & 12.4  & 4.1   & 5.5   & 7.2   & 2.3   & 3.4   & 6.5   & 5.2   & 9.8   & 6.3 \\
    25    & 50    & 50    & 75    & 8.1   & 6.1   & 12.4  & 0.2   & 0.0   & 0.0   & 0.1   & 0.0   & 0.0   & 4.7   & 6.9   & 4.1 \\
    25    & 50    & 75    & 25    & 8.1   & 6.1   & 12.4  & 0.4   & 1.6   & 1.2   & 1.4   & 0.2   & 1.7   & 6.0   & 9.6   & 10.6 \\
    25    & 50    & 75    & 50    & 8.1   & 6.1   & 12.4  & 0.5   & 1.8   & 0.7   & 1.4   & 0.3   & 0.9   & 5.9   & 9.0   & 9.2 \\
    25    & 50    & 75    & 75    & 8.1   & 6.1   & 12.4  & 0.0   & 0.0   & 0.1   & 0.0   & 0.0   & 0.1   & 5.0   & 6.7   & 6.7 \\
    25    & 75    & 25    & 25    & 6.0   & 5.0   & 10.4  & 9.6   & 13.5  & 14.0  & 9.5   & 11.0  & 14.2  & 4.7   & 8.0   & 16.0 \\
    25    & 75    & 25    & 50    & 6.0   & 5.0   & 10.4  & 9.6   & 13.5  & 14.1  & 9.6   & 11.0  & 14.2  & 3.8   & 7.6   & 14.9 \\
    25    & 75    & 25    & 75    & 6.0   & 5.0   & 10.4  & 5.7   & 10.1  & 5.7   & 4.6   & 8.6   & 6.9   & 3.7   & 6.3   & 9.5 \\
    25    & 75    & 50    & 25    & 6.0   & 5.0   & 10.4  & 9.6   & 13.5  & 14.1  & 9.5   & 10.9  & 13.8  & 4.5   & 7.1   & 16.5 \\
    25    & 75    & 50    & 50    & 6.0   & 5.0   & 10.4  & 9.6   & 13.5  & 14.3  & 9.6   & 10.9  & 13.8  & 4.0   & 6.9   & 15.4 \\
    25    & 75    & 50    & 75    & 6.0   & 5.0   & 10.4  & 5.7   & 9.6   & 5.7   & 4.6   & 8.1   & 6.4   & 3.5   & 5.8   & 9.7 \\
    25    & 75    & 75    & 25    & 6.0   & 5.0   & 10.4  & 8.8   & 12.5  & 13.1  & 8.1   & 10.4  & 12.2  & 4.6   & 7.2   & 18.4 \\
    25    & 75    & 75    & 50    & 6.0   & 5.0   & 10.4  & 8.8   & 12.5  & 13.1  & 7.8   & 10.4  & 12.2  & 4.3   & 6.1   & 17.9 \\
    25    & 75    & 75    & 75    & 6.0   & 5.0   & 10.4  & 4.3   & 8.9   & 4.3   & 2.7   & 5.8   & 4.3   & 3.6   & 3.6   & 13.2 \\
    50    & 25    & 25    & 25    & 8.9   & 6.5   & 13.2  & 0.0   & 0.4   & 0.1   & 0.1   & 0.3   & 0.2   & 6.6   & 10.0  & 4.9 \\
    50    & 25    & 25    & 50    & 8.9   & 6.5   & 13.2  & 0.1   & 0.0   & 0.1   & 0.1   & 0.1   & 0.1   & 6.4   & 9.6   & 4.0 \\
    50    & 25    & 25    & 75    & 8.9   & 6.5   & 13.2  & 0.1   & 0.0   & 0.0   & 0.1   & 0.0   & 0.0   & 5.3   & 7.9   & 4.4 \\
    50    & 25    & 50    & 25    & 8.9   & 6.5   & 13.2  & 0.0   & 0.0   & 0.0   & 0.2   & 0.0   & 0.1   & 6.8   & 10.4  & 6.1 \\
    50    & 25    & 50    & 50    & 8.9   & 6.5   & 13.2  & 0.0   & 0.1   & 0.1   & 0.0   & 0.1   & 0.0   & 6.7   & 10.0  & 5.4 \\
    50    & 25    & 50    & 75    & 8.9   & 6.5   & 13.2  & 0.2   & 0.0   & 0.0   & 0.1   & 0.0   & 0.0   & 5.8   & 7.9   & 5.0 \\
    50    & 25    & 75    & 25    & 8.9   & 6.5   & 13.2  & 0.0   & 0.1   & 0.0   & 0.2   & 0.0   & 0.1   & 7.0   & 9.8   & 8.6 \\
    50    & 25    & 75    & 50    & 8.9   & 6.5   & 13.2  & 0.0   & 0.1   & 0.0   & 0.1   & 0.0   & 0.1   & 7.1   & 9.5   & 7.8 \\
    50    & 25    & 75    & 75    & 8.9   & 6.5   & 13.2  & 0.0   & 0.0   & 0.0   & 0.0   & 0.0   & 0.0   & 6.0   & 8.2   & 6.7 \\
    50    & 50    & 25    & 25    & 8.1   & 6.1   & 12.4  & 8.0   & 8.6   & 9.7   & 7.6   & 6.8   & 9.9   & 5.3   & 10.3  & 6.6 \\
    50    & 50    & 25    & 50    & 8.1   & 6.1   & 12.4  & 3.9   & 7.1   & 7.2   & 4.9   & 4.3   & 7.8   & 5.1   & 9.5   & 5.2 \\
    50    & 50    & 25    & 75    & 8.1   & 6.1   & 12.4  & 1.4   & 0.4   & 0.0   & 0.1   & 0.0   & 0.0   & 4.6   & 7.2   & 4.0 \\
    50    & 50    & 50    & 25    & 8.1   & 6.1   & 12.4  & 6.9   & 5.5   & 7.4   & 5.9   & 6.3   & 7.8   & 5.6   & 10.4  & 8.0 \\
    50    & 50    & 50    & 50    & 8.1   & 6.1   & 12.4  & 4.3   & 4.9   & 7.8   & 2.1   & 3.7   & 7.4   & 5.5   & 9.8   & 6.6 \\
    50    & 50    & 50    & 75    & 8.1   & 6.1   & 12.4  & 0.2   & 0.4   & 0.0   & 0.1   & 0.0   & 0.0   & 4.8   & 7.2   & 4.2 \\
    50    & 50    & 75    & 25    & 8.1   & 6.1   & 12.4  & 0.5   & 1.6   & 0.8   & 0.7   & 0.1   & 1.0   & 6.4   & 9.6   & 10.2 \\
    50    & 50    & 75    & 50    & 8.1   & 6.1   & 12.4  & 0.5   & 1.8   & 0.7   & 0.6   & 0.3   & 0.7   & 6.3   & 9.1   & 9.1 \\
    50    & 50    & 75    & 75    & 8.1   & 6.1   & 12.4  & 0.1   & 0.0   & 0.1   & 0.0   & 0.0   & 0.1   & 5.3   & 7.1   & 7.2 \\
    50    & 75    & 25    & 25    & 6.0   & 5.0   & 10.4  & 9.6   & 13.5  & 14.0  & 9.5   & 11.0  & 14.2  & 4.7   & 7.9   & 15.9 \\
    50    & 75    & 25    & 50    & 6.0   & 5.0   & 10.4  & 9.6   & 13.5  & 14.1  & 9.6   & 11.0  & 14.2  & 3.9   & 7.6   & 13.5 \\
    50    & 75    & 25    & 75    & 6.0   & 5.0   & 10.4  & 5.7   & 10.1  & 6.4   & 4.6   & 8.7   & 7.6   & 3.4   & 6.4   & 9.1 \\
    50    & 75    & 50    & 25    & 6.0   & 5.0   & 10.4  & 9.6   & 13.5  & 14.0  & 9.5   & 10.9  & 14.1  & 4.6   & 7.7   & 16.7 \\
    50    & 75    & 50    & 50    & 6.0   & 5.0   & 10.4  & 9.6   & 13.5  & 14.3  & 9.6   & 10.9  & 13.8  & 4.1   & 6.9   & 15.7 \\
    50    & 75    & 50    & 75    & 6.0   & 5.0   & 10.4  & 5.7   & 10.1  & 6.4   & 4.6   & 8.7   & 6.4   & 3.6   & 6.0   & 9.2 \\
    50    & 75    & 75    & 25    & 6.0   & 5.0   & 10.4  & 8.8   & 12.5  & 13.1  & 8.1   & 10.1  & 12.2  & 4.8   & 7.5   & 18.4 \\
    50    & 75    & 75    & 50    & 6.0   & 5.0   & 10.4  & 8.8   & 12.5  & 13.1  & 8.2   & 10.1  & 12.2  & 4.3   & 6.3   & 17.6 \\
    50    & 75    & 75    & 75    & 6.0   & 5.0   & 10.4  & 4.3   & 8.9   & 4.3   & 2.9   & 5.8   & 4.3   & 3.4   & 4.4   & 13.2 \\
    75    & 25    & 25    & 25    & 8.9   & 6.5   & 13.2  & 0.0   & 0.0   & 0.3   & 0.1   & 0.0   & 0.3   & 6.9   & 9.1   & 6.9 \\
    75    & 25    & 25    & 50    & 8.9   & 6.5   & 13.2  & 0.2   & 0.2   & 0.0   & 0.5   & 0.1   & 0.1   & 6.6   & 9.0   & 6.6 \\
    75    & 25    & 25    & 75    & 8.9   & 6.5   & 13.2  & 0.1   & 0.1   & 0.4   & 0.0   & 0.0   & 0.0   & 6.0   & 7.9   & 5.7 \\
    75    & 25    & 50    & 25    & 8.9   & 6.5   & 13.2  & 0.0   & 0.0   & 0.3   & 0.1   & 0.1   & 0.1   & 7.1   & 9.3   & 7.8 \\
    75    & 25    & 50    & 50    & 8.9   & 6.5   & 13.2  & 0.2   & 0.1   & 0.3   & 0.3   & 0.1   & 0.0   & 7.0   & 9.0   & 7.5 \\
    75    & 25    & 50    & 75    & 8.9   & 6.5   & 13.2  & 0.1   & 0.1   & 0.1   & 0.0   & 0.0   & 0.0   & 6.2   & 8.0   & 5.8 \\
    75    & 25    & 75    & 25    & 8.9   & 6.5   & 13.2  & 0.1   & 0.2   & 0.3   & 0.1   & 0.1   & 0.2   & 7.3   & 8.6   & 9.5 \\
    75    & 25    & 75    & 50    & 8.9   & 6.5   & 13.2  & 0.1   & 0.2   & 0.3   & 0.2   & 0.1   & 0.2   & 7.1   & 8.6   & 8.8 \\
    75    & 25    & 75    & 75    & 8.9   & 6.5   & 13.2  & 0.0   & 0.1   & 0.2   & 0.0   & 0.0   & 0.2   & 6.3   & 8.1   & 8.1 \\
    75    & 50    & 25    & 25    & 8.1   & 6.1   & 12.4  & 3.1   & 2.3   & 2.9   & 2.6   & 1.2   & 4.0   & 6.0   & 9.6   & 8.4 \\
    75    & 50    & 25    & 50    & 8.1   & 6.1   & 12.4  & 3.0   & 0.5   & 3.3   & 1.7   & 0.9   & 1.8   & 5.9   & 9.5   & 7.4 \\
    75    & 50    & 25    & 75    & 8.1   & 6.1   & 12.4  & 0.1   & 0.1   & 0.8   & 0.0   & 0.2   & 0.0   & 5.3   & 7.9   & 5.6 \\
    75    & 50    & 50    & 25    & 8.1   & 6.1   & 12.4  & 2.6   & 2.6   & 0.4   & 1.5   & 2.4   & 0.5   & 6.2   & 9.9   & 9.2 \\
    75    & 50    & 50    & 50    & 8.1   & 6.1   & 12.4  & 2.1   & 0.1   & 0.3   & 0.8   & 0.1   & 0.5   & 6.2   & 9.5   & 8.6 \\
    75    & 50    & 50    & 75    & 8.1   & 6.1   & 12.4  & 0.1   & 0.1   & 0.7   & 0.0   & 0.2   & 0.0   & 5.5   & 7.7   & 6.4 \\
    75    & 50    & 75    & 25    & 8.1   & 6.1   & 12.4  & 0.4   & 0.2   & 0.5   & 0.2   & 0.1   & 0.8   & 6.8   & 8.9   & 10.8 \\
    75    & 50    & 75    & 50    & 8.1   & 6.1   & 12.4  & 0.4   & 0.2   & 0.4   & 0.2   & 0.1   & 0.8   & 6.7   & 8.8   & 10.0 \\
    75    & 50    & 75    & 75    & 8.1   & 6.1   & 12.4  & 0.1   & 0.1   & 0.3   & 0.0   & 0.0   & 0.3   & 5.7   & 7.8   & 7.7 \\
    75    & 75    & 25    & 25    & 6.0   & 5.0   & 10.4  & 9.6   & 13.1  & 14.4  & 9.6   & 11.0  & 12.6  & 5.5   & 8.6   & 14.8 \\
    75    & 75    & 25    & 50    & 6.0   & 5.0   & 10.4  & 9.6   & 13.0  & 14.4  & 9.6   & 11.0  & 12.6  & 4.8   & 8.3   & 14.2 \\
    75    & 75    & 25    & 75    & 6.0   & 5.0   & 10.4  & 4.1   & 9.0   & 7.4   & 3.6   & 5.6   & 6.6   & 3.9   & 6.9   & 8.4 \\
    75    & 75    & 50    & 25    & 6.0   & 5.0   & 10.4  & 9.6   & 13.1  & 14.4  & 9.6   & 11.1  & 12.5  & 5.6   & 8.5   & 15.5 \\
    75    & 75    & 50    & 50    & 6.0   & 5.0   & 10.4  & 9.6   & 13.0  & 14.1  & 9.6   & 11.1  & 12.5  & 4.9   & 8.1   & 15.4 \\
    75    & 75    & 50    & 75    & 6.0   & 5.0   & 10.4  & 4.1   & 7.6   & 7.5   & 3.6   & 5.6   & 6.9   & 3.7   & 6.6   & 9.4 \\
    75    & 75    & 75    & 25    & 6.0   & 5.0   & 10.4  & 7.1   & 8.2   & 9.5   & 7.9   & 6.8   & 10.7  & 6.5   & 8.3   & 18.1 \\
    75    & 75    & 75    & 50    & 6.0   & 5.0   & 10.4  & 7.1   & 8.2   & 9.5   & 7.9   & 6.8   & 10.7  & 5.6   & 7.7   & 17.8 \\
    75    & 75    & 75    & 75    & 6.0   & 5.0   & 10.4  & 4.5   & 7.7   & 7.5   & 3.1   & 5.3   & 6.4   & 4.0   & 5.3   & 12.8 \\
    \midrule
    \multicolumn{13}{l}{\textit{Note: The values displayed are in \$k}}\\
    \end{tabular}%
  \label{tab:MC_MajDiag_81rows}%
\end{table}%

\begin{table}[htbp]
  \tiny
  \centering
  \caption{Marginal Cost of Major Therapeutic Procedures}
    \begin{tabular}{cccc|ccc|ccc|ccc|ccc}
    \toprule
    \multicolumn{4}{c|}{Percentile} & \multicolumn{3}{c|}{Quadratic Regression} & \multicolumn{3}{c|}{CNLS-d (median)} & \multicolumn{3}{c}{CNLS-d (equal)} & \multicolumn{3}{c}{LL Kernel}\\
    \midrule
    \multicolumn{1}{c}{MinDiag} & MinTher & MajDiag & MajTher & 2007  & 2008  & 2009  & 2007  & 2008  & 2009  & 2007  & 2008  & 2009 & 2007  & 2008  & 2009\\
        \midrule
    25    & 25    & 25    & 25    & 10.5  & 11.5  & 9.8   & 1.7   & 2.8   & 4.4   & 3.3   & 1.8   & 4.9   & 18.4  & 14.3  & 22.9 \\
    25    & 25    & 25    & 50    & 11.7  & 13.0  & 10.8  & 17.3  & 17.0  & 20.0  & 16.8  & 15.2  & 18.5  & 17.5  & 11.2  & 18.5 \\
    25    & 25    & 25    & 75    & 15.1  & 17.2  & 14.5  & 19.4  & 22.3  & 24.6  & 19.8  & 21.8  & 24.0  & 15.2  & 10.2  & 12.7 \\
    25    & 25    & 50    & 25    & 10.5  & 11.5  & 9.8   & 0.0   & 0.0   & 0.2   & 0.1   & 0.0   & 0.1   & 17.4  & 14.6  & 21.7 \\
    25    & 25    & 50    & 50    & 11.7  & 13.0  & 10.8  & 9.6   & 10.6  & 13.5  & 11.2  & 10.3  & 13.4  & 16.8  & 12.2  & 18.1 \\
    25    & 25    & 50    & 75    & 15.1  & 17.2  & 14.5  & 19.8  & 22.2  & 24.6  & 19.8  & 21.8  & 24.0  & 15.8  & 10.5  & 13.9 \\
    25    & 25    & 75    & 25    & 10.5  & 11.5  & 9.8   & 0.1   & 0.9   & 0.1   & 0.1   & 0.0   & 0.2   & 17.4  & 14.9  & 17.2 \\
    25    & 25    & 75    & 50    & 11.7  & 13.0  & 10.8  & 1.3   & 1.7   & 1.3   & 2.9   & 0.1   & 5.1   & 17.3  & 14.0  & 17.2 \\
    25    & 25    & 75    & 75    & 15.1  & 17.2  & 14.5  & 16.1  & 18.0  & 23.5  & 16.8  & 16.5  & 23.8  & 16.9  & 11.8  & 14.2 \\
    25    & 50    & 25    & 25    & 10.5  & 11.5  & 9.8   & 0.1   & 0.1   & 0.2   & 0.1   & 0.1   & 0.5   & 17.9  & 12.6  & 20.0 \\
    25    & 50    & 25    & 50    & 11.7  & 13.0  & 10.8  & 12.9  & 7.9   & 8.2   & 10.0  & 7.0   & 6.2   & 17.1  & 9.7   & 16.8 \\
    25    & 50    & 25    & 75    & 15.1  & 17.2  & 14.5  & 19.8  & 22.3  & 24.6  & 19.8  & 21.8  & 24.0  & 15.0  & 8.7   & 12.1 \\
    25    & 50    & 50    & 25    & 10.5  & 11.5  & 9.8   & 0.4   & 0.2   & 0.4   & 0.1   & 0.5   & 0.3   & 17.3  & 13.3  & 18.9 \\
    25    & 50    & 50    & 50    & 11.7  & 13.0  & 10.8  & 5.2   & 5.2   & 1.4   & 10.5  & 8.1   & 5.8   & 16.6  & 10.8  & 16.5 \\
    25    & 50    & 50    & 75    & 15.1  & 17.2  & 14.5  & 19.8  & 22.2  & 24.6  & 19.8  & 21.8  & 24.0  & 15.4  & 8.5   & 12.6 \\
    25    & 50    & 75    & 25    & 10.5  & 11.5  & 9.8   & 0.1   & 0.3   & 0.1   & 0.2   & 0.1   & 0.2   & 17.3  & 14.1  & 15.7 \\
    25    & 50    & 75    & 50    & 11.7  & 13.0  & 10.8  & 0.1   & 0.5   & 0.9   & 0.2   & 0.2   & 4.6   & 16.9  & 13.0  & 16.3 \\
    25    & 50    & 75    & 75    & 15.1  & 17.2  & 14.5  & 16.1  & 18.0  & 22.8  & 16.8  & 16.5  & 23.8  & 15.9  & 10.2  & 14.3 \\
    25    & 75    & 25    & 25    & 10.5  & 11.5  & 9.8   & 0.0   & 0.0   & 0.1   & 0.2   & 0.1   & 0.1   & 17.1  & 9.3   & 9.9 \\
    25    & 75    & 25    & 50    & 11.7  & 13.0  & 10.8  & 1.6   & 0.0   & 0.3   & 0.7   & 0.1   & 0.1   & 15.4  & 7.0   & 9.8 \\
    25    & 75    & 25    & 75    & 15.1  & 17.2  & 14.5  & 18.3  & 12.4  & 20.9  & 16.2  & 10.9  & 14.3  & 15.3  & 6.2   & 6.7 \\
    25    & 75    & 50    & 25    & 10.5  & 11.5  & 9.8   & 0.2   & 0.0   & 0.2   & 0.0   & 0.1   & 0.2   & 16.3  & 9.2   & 9.4 \\
    25    & 75    & 50    & 50    & 11.7  & 13.0  & 10.8  & 0.2   & 0.3   & 0.4   & 0.8   & 0.1   & 0.3   & 15.6  & 6.6   & 8.4 \\
    25    & 75    & 50    & 75    & 15.1  & 17.2  & 14.5  & 18.3  & 12.8  & 20.9  & 16.2  & 11.3  & 15.2  & 15.2  & 6.2   & 5.8 \\
    25    & 75    & 75    & 25    & 10.5  & 11.5  & 9.8   & 0.1   & 0.1   & 0.1   & 0.2   & 0.1   & 0.1   & 15.7  & 11.1  & 9.8 \\
    25    & 75    & 75    & 50    & 11.7  & 13.0  & 10.8  & 0.1   & 0.1   & 0.1   & 0.6   & 0.1   & 0.1   & 15.6  & 8.4   & 9.7 \\
    25    & 75    & 75    & 75    & 15.1  & 17.2  & 14.5  & 15.5  & 10.4  & 19.7  & 17.2  & 10.8  & 16.7  & 14.7  & 6.9   & 8.0 \\
    50    & 25    & 25    & 25    & 10.5  & 11.5  & 9.8   & 0.3   & 0.1   & 0.0   & 0.2   & 0.3   & 2.7   & 18.6  & 14.1  & 21.4 \\
    50    & 25    & 25    & 50    & 11.7  & 13.0  & 10.8  & 17.8  & 17.7  & 17.0  & 16.3  & 15.4  & 19.3  & 18.0  & 11.4  & 17.2 \\
    50    & 25    & 25    & 75    & 15.1  & 17.2  & 14.5  & 19.2  & 22.3  & 24.6  & 19.8  & 21.8  & 24.0  & 15.1  & 11.1  & 13.1 \\
    50    & 25    & 50    & 25    & 10.5  & 11.5  & 9.8   & 0.1   & 0.0   & 0.1   & 0.2   & 0.0   & 0.1   & 17.6  & 14.7  & 20.4 \\
    50    & 25    & 50    & 50    & 11.7  & 13.0  & 10.8  & 11.3  & 11.8  & 15.7  & 10.5  & 10.3  & 14.6  & 17.2  & 12.5  & 17.2 \\
    50    & 25    & 50    & 75    & 15.1  & 17.2  & 14.5  & 19.8  & 22.1  & 24.6  & 19.8  & 21.8  & 24.0  & 15.9  & 10.7  & 13.5 \\
    50    & 25    & 75    & 25    & 10.5  & 11.5  & 9.8   & 0.1   & 1.0   & 0.3   & 0.1   & 0.0   & 1.3   & 17.3  & 15.1  & 16.8 \\
    50    & 25    & 75    & 50    & 11.7  & 13.0  & 10.8  & 0.9   & 1.5   & 2.1   & 0.5   & 0.2   & 1.3   & 17.2  & 14.5  & 16.9 \\
    50    & 25    & 75    & 75    & 15.1  & 17.2  & 14.5  & 16.1  & 18.0  & 23.5  & 16.8  & 16.5  & 23.6  & 16.7  & 13.0  & 14.3 \\
    50    & 50    & 25    & 25    & 10.5  & 11.5  & 9.8   & 0.2   & 0.1   & 0.2   & 0.1   & 0.4   & 0.5   & 18.2  & 12.8  & 18.6 \\
    50    & 50    & 25    & 50    & 11.7  & 13.0  & 10.8  & 11.1  & 7.9   & 10.0  & 9.4   & 7.7   & 5.5   & 17.6  & 10.3  & 15.8 \\
    50    & 50    & 25    & 75    & 15.1  & 17.2  & 14.5  & 19.8  & 22.2  & 24.6  & 19.8  & 21.8  & 24.0  & 14.9  & 9.4   & 12.1 \\
    50    & 50    & 50    & 25    & 10.5  & 11.5  & 9.8   & 0.4   & 0.2   & 0.5   & 0.1   & 0.1   & 0.4   & 17.5  & 13.6  & 17.6 \\
    50    & 50    & 50    & 50    & 11.7  & 13.0  & 10.8  & 3.7   & 7.7   & 1.7   & 6.9   & 7.1   & 3.7   & 16.9  & 11.5  & 15.6 \\
    50    & 50    & 50    & 75    & 15.1  & 17.2  & 14.5  & 19.8  & 22.0  & 24.6  & 19.8  & 21.8  & 24.0  & 15.2  & 9.5   & 12.8 \\
    50    & 50    & 75    & 25    & 10.5  & 11.5  & 9.8   & 0.1   & 0.3   & 0.2   & 0.0   & 0.0   & 0.4   & 17.4  & 14.6  & 14.8 \\
    50    & 50    & 75    & 50    & 11.7  & 13.0  & 10.8  & 0.1   & 0.5   & 0.3   & 0.1   & 0.2   & 1.0   & 17.0  & 13.6  & 15.3 \\
    50    & 50    & 75    & 75    & 15.1  & 17.2  & 14.5  & 17.4  & 18.0  & 22.8  & 16.8  & 16.5  & 23.8  & 16.0  & 11.5  & 14.8 \\
    50    & 75    & 25    & 25    & 10.5  & 11.5  & 9.8   & 0.0   & 0.0   & 0.1   & 0.2   & 0.1   & 0.0   & 17.0  & 9.3   & 9.5 \\
    50    & 75    & 25    & 50    & 11.7  & 13.0  & 10.8  & 1.6   & 0.0   & 0.3   & 0.7   & 0.1   & 0.1   & 15.7  & 7.6   & 6.3 \\
    50    & 75    & 25    & 75    & 15.1  & 17.2  & 14.5  & 18.3  & 12.4  & 19.8  & 16.2  & 11.0  & 13.4  & 14.4  & 6.6   & 6.5 \\
    50    & 75    & 50    & 25    & 10.5  & 11.5  & 9.8   & 0.2   & 0.0   & 0.1   & 0.0   & 0.1   & 0.1   & 16.6  & 10.0  & 8.9 \\
    50    & 75    & 50    & 50    & 11.7  & 13.0  & 10.8  & 0.2   & 0.2   & 0.4   & 0.8   & 0.1   & 0.3   & 15.8  & 8.2   & 8.1 \\
    50    & 75    & 50    & 75    & 15.1  & 17.2  & 14.5  & 18.3  & 12.4  & 19.8  & 16.2  & 11.0  & 15.2  & 14.8  & 6.7   & 5.3 \\
    50    & 75    & 75    & 25    & 10.5  & 11.5  & 9.8   & 0.1   & 0.1   & 0.1   & 0.2   & 0.1   & 0.1   & 16.1  & 11.6  & 9.6 \\
    50    & 75    & 75    & 50    & 11.7  & 13.0  & 10.8  & 0.1   & 0.1   & 0.1   & 0.3   & 0.1   & 0.1   & 15.5  & 9.2   & 9.1 \\
    50    & 75    & 75    & 75    & 15.1  & 17.2  & 14.5  & 15.5  & 10.4  & 19.7  & 16.3  & 10.8  & 16.7  & 14.6  & 8.3   & 7.2 \\
    75    & 25    & 25    & 25    & 10.5  & 11.5  & 9.8   & 0.3   & 0.0   & 0.3   & 0.3   & 0.0   & 1.5   & 18.9  & 14.7  & 15.0 \\
    75    & 25    & 25    & 50    & 11.7  & 13.0  & 10.8  & 2.4   & 9.7   & 4.0   & 6.7   & 6.2   & 4.3   & 18.0  & 13.9  & 13.9 \\
    75    & 25    & 25    & 75    & 15.1  & 17.2  & 14.5  & 19.6  & 19.5  & 24.7  & 19.3  & 18.3  & 24.4  & 15.7  & 13.2  & 11.0 \\
    75    & 25    & 50    & 25    & 10.5  & 11.5  & 9.8   & 0.1   & 0.0   & 0.2   & 0.3   & 0.0   & 0.2   & 18.0  & 15.4  & 14.8 \\
    75    & 25    & 50    & 50    & 11.7  & 13.0  & 10.8  & 3.9   & 5.5   & 0.8   & 4.5   & 2.8   & 3.3   & 17.6  & 14.7  & 13.8 \\
    75    & 25    & 50    & 75    & 15.1  & 17.2  & 14.5  & 19.6  & 19.5  & 24.7  & 19.3  & 18.9  & 24.4  & 16.4  & 13.4  & 11.9 \\
    75    & 25    & 75    & 25    & 10.5  & 11.5  & 9.8   & 0.1   & 0.1   & 0.3   & 0.0   & 0.1   & 0.4   & 17.1  & 16.1  & 14.0 \\
    75    & 25    & 75    & 50    & 11.7  & 13.0  & 10.8  & 0.1   & 0.1   & 0.3   & 0.2   & 0.1   & 0.4   & 17.1  & 16.5  & 14.3 \\
    75    & 25    & 75    & 75    & 15.1  & 17.2  & 14.5  & 19.5  & 11.6  & 23.4  & 19.1  & 18.5  & 20.8  & 17.5  & 15.4  & 16.5 \\
    75    & 50    & 25    & 25    & 10.5  & 11.5  & 9.8   & 0.3   & 0.1   & 0.5   & 1.7   & 0.1   & 0.7   & 18.6  & 13.9  & 13.2 \\
    75    & 50    & 25    & 50    & 11.7  & 13.0  & 10.8  & 0.9   & 7.4   & 1.4   & 3.1   & 4.5   & 3.5   & 17.8  & 13.3  & 12.5 \\
    75    & 50    & 25    & 75    & 15.1  & 17.2  & 14.5  & 19.6  & 19.5  & 24.7  & 19.3  & 19.2  & 24.4  & 15.0  & 12.4  & 10.9 \\
    75    & 50    & 50    & 25    & 10.5  & 11.5  & 9.8   & 0.5   & 0.1   & 0.2   & 0.7   & 0.1   & 0.1   & 17.8  & 15.0  & 12.8 \\
    75    & 50    & 50    & 50    & 11.7  & 13.0  & 10.8  & 0.7   & 5.5   & 0.8   & 2.5   & 2.8   & 3.3   & 17.3  & 14.2  & 12.2 \\
    75    & 50    & 50    & 75    & 15.1  & 17.2  & 14.5  & 19.6  & 19.5  & 24.7  & 19.3  & 19.8  & 24.4  & 15.7  & 12.3  & 12.3 \\
    75    & 50    & 75    & 25    & 10.5  & 11.5  & 9.8   & 0.1   & 0.1   & 0.2   & 0.2   & 0.1   & 0.2   & 17.3  & 16.0  & 12.1 \\
    75    & 50    & 75    & 50    & 11.7  & 13.0  & 10.8  & 0.1   & 0.1   & 0.3   & 0.2   & 0.1   & 0.2   & 17.0  & 16.1  & 12.9 \\
    75    & 50    & 75    & 75    & 15.1  & 17.2  & 14.5  & 19.2  & 11.6  & 24.2  & 19.1  & 18.5  & 20.0  & 16.5  & 15.0  & 15.5 \\
    75    & 75    & 25    & 25    & 10.5  & 11.5  & 9.8   & 0.1   & 0.2   & 0.1   & 0.1   & 0.2   & 0.3   & 16.8  & 11.6  & 6.9 \\
    75    & 75    & 25    & 50    & 11.7  & 13.0  & 10.8  & 0.1   & 0.4   & 0.1   & 0.1   & 0.2   & 0.3   & 16.2  & 10.4  & 6.5 \\
    75    & 75    & 25    & 75    & 15.1  & 17.2  & 14.5  & 18.6  & 12.6  & 15.4  & 15.9  & 15.0  & 14.1  & 14.6  & 10.1  & 4.5 \\
    75    & 75    & 50    & 25    & 10.5  & 11.5  & 9.8   & 0.1   & 0.2   & 0.1   & 0.1   & 0.1   & 0.1   & 16.5  & 12.4  & 7.2 \\
    75    & 75    & 50    & 50    & 11.7  & 13.0  & 10.8  & 0.1   & 0.4   & 0.1   & 0.1   & 0.1   & 0.1   & 15.8  & 11.5  & 7.2 \\
    75    & 75    & 50    & 75    & 15.1  & 17.2  & 14.5  & 18.6  & 13.4  & 15.9  & 15.9  & 15.0  & 13.6  & 14.4  & 10.6  & 5.3 \\
    75    & 75    & 75    & 25    & 10.5  & 11.5  & 9.8   & 0.1   & 0.1   & 0.1   & 0.1   & 0.1   & 0.2   & 15.7  & 14.1  & 8.5 \\
    75    & 75    & 75    & 50    & 11.7  & 13.0  & 10.8  & 0.1   & 0.2   & 0.1   & 0.1   & 0.1   & 0.2   & 15.3  & 13.1  & 7.4 \\
    75    & 75    & 75    & 75    & 15.1  & 17.2  & 14.5  & 13.5  & 7.2   & 14.2  & 12.2  & 11.7  & 12.1  & 14.6  & 12.7  & 7.9 \\
    \midrule
    \multicolumn{13}{l}{\textit{Note: The values displayed are in \$k}}\\
    \end{tabular}%
  \label{tab:MC_MajTher_81rows}%
\end{table}%

\end{document}